\documentclass[11pt,letterpaper]{article}
\usepackage[utf8]{inputenc}
\usepackage[english]{babel}
\usepackage[margin=1in]{geometry} 
\usepackage{amsmath,amsthm,amssymb,amsfonts,bbm}
\usepackage{titling}
\usepackage{graphicx}
\usepackage{verbatim}
\usepackage{xcolor}
\usepackage{calc}
\usepackage[ruled]{algorithm2e}   
\usepackage{bbm}
\usepackage{mathrsfs}
\usepackage{tikz}
\usepackage{authblk}
\usepackage{ulem}

\usepackage{caption}  
\usepackage{subcaption} 

\usepackage{hyperref}
\hypersetup{colorlinks=true, linkcolor=darkred, citecolor=darkgreen, urlcolor=darkblue} 
\usepackage[round]{natbib} 

\usepackage{booktabs} 
\usepackage{arydshln} 

\usepackage{mathtools}

\usepackage{rotating}

\usepackage{tikz}
\usetikzlibrary{positioning}

\usepackage{soul} 

\usepackage{color}
\definecolor{darkred}{RGB}{100,0,0}
\definecolor{darkgreen}{RGB}{0,100,0}
\definecolor{darkblue}{RGB}{0,0,150}

\usepackage{pagecolor}
\definecolor{softgreen}{RGB}{230, 245, 230}

\usepackage{accents}

\newcommand{\E}{E}
\newcommand{\PP}{\mathbb{P}}
\newcommand{\var}{\text{var}}
\newcommand{\R}{\mathbb{R}}

\newcommand{\Z}{\mathcal{Z}}
\newcommand{\I}{\mathcal{I}}
\newcommand{\Hc}{\mathcal{H}}
\newcommand{\A}{\mathcal{A}}
\newcommand{\B}{\mathcal{B}}
\newcommand{\D}{\mathcal{D}}

\newcommand{\Ec}{\mathcal{E}}

\newcommand{\ind}{\mathbbm{1}}

\newcommand{\V}{\mathcal{V}}
\newcommand{\Gc}{\mathcal{G}}
\newcommand{\T}{\mathcal{T}}

\newcommand{\Oc}{\mathcal{O}}
\newcommand{\Rc}{\mathcal{R}}

\newcommand{\cc}{c}

\newcommand{\crds}{\rho}

\newcommand{\BB}{B}
\newcommand{\BBB}{\phi}
\newcommand{\UU}{\phi}
\newcommand{\ksup}{d}
\newcommand{\kstar}{k^*}
\newcommand{\Dlstar}{\tilde M^*}
\newcommand{\Dlstarz}{M^*}
\newcommand{\betatwo}{c}

\newcommand{\consta}{C_1}
\newcommand{\constb}{C_2}
\newcommand{\constc}{C_3}
\newcommand{\constd}{C_4}
\newcommand{\constg}{C_6}
\newcommand{\consth}{C_7}
\newcommand{\consti}{C_8}

\newcommand{\gammaone}{\varphi_1}
\newcommand{\gammatwo}{\varphi_2}

\newcommand{\sstar}{\text{star}}

\newcommand{\RR}{\text{R}}
\newcommand{\DR}{\text{DR}}

\newcommand{\cconst}{\text{const}}

\newcommand{\TV}{\text{TV}}
\newcommand{\convp}{\stackrel{p}{\rightarrow}}  
\newcommand{\convd}{\stackrel{d}{\rightarrow}}  

\newcommand{\bigCI}{\mathrel{\text{\scalebox{1.07}{$\perp\mkern-10mu\perp$}}}}

\usepackage{mathtools}

\DeclareMathOperator*{\argmin}{arg\,min}

\newtheorem{theorem}{Theorem}
\newtheorem{proposition}{Proposition}
\newtheorem{corollary}{Corollary}
\newtheorem{lemma}{Lemma}
\newtheorem{assumption}{Assumption}
\newtheorem{remark}{Remark}

\newcommand{\red}[1]{\textcolor{red}{#1}}

\newcommand{\org}[1]{\textcolor{orange}{#1}}

\newcommand{\mis}[1]{\textcolor{red}{#1}}
\newcommand{\smis}[1]{\textcolor{orange}{#1}}

\renewcommand{\baselinestretch}{1.66}

\newtheorem{appxlemma}{Lemma}

\newtheorem{appxassumption}{Assumption}

\usepackage{xr}            
\makeatletter
\newcommand*{\addFileDependency}[1]{
	\typeout{(#1)}
	\@addtofilelist{#1}
	\IfFileExists{#1}{}{\typeout{No file #1.}}
}
\makeatother

\title{\bf  \vspace{-2em} A Liberating Framework from Truncation and Censoring, with Application to Learning Treatment Effects}
  \author{Yuyao Wang, Andrew Ying,  and   Ronghui Xu 
  \thanks{The authors gratefully acknowledge Drs. Steve Edland, Jue (Marquis) Hou and Kendrick Li for helpful discussions. This research was done using services provided by the Open Science Grid (OSG) Consortium, 
   which is supported by the National Science Foundation awards \#2030508 and \#1836650. }\hspace{.2cm}\\
    Department of Mathematics, Herbert Wertheim School of Public Health \\
     and Halicioglu Data Science Institute, University of California San Diego
    }
\date{}

\begin{document}
\maketitle

\renewcommand{\baselinestretch}{1}
\vspace{-3em}
\begin{abstract}
Time-to-event outcomes are often subject to left truncation and right censoring. While many survival analysis methods have been developed to handle truncation and censoring, majority of the past works require strong independence assumptions. We relax these stringent assumptions through leveraging covariate information together with orthogonal learning, and develop a liberating framework from left truncation and right censoring so that desirable properties like double robustness can be immediately transferred from settings without truncation or censoring. To illustrate its generality and ease to use, the framework is applied to estimation of the average treatment effect (ATE) and the conditional average treatment effect (CATE). For the ATE, we establish both model  and rate double robustness under confounding, truncation and censoring; for the CATE, we show that the orthogonal and the doubly robust learners under these three sources of bias can achieve oracle rate of convergence. We study the estimators both theoretically and through extensive simulation, and apply them to analyzing the effect of mid-life heavy drinking on late life cognitive impairment free survival, using data from the Honolulu Asia Aging Study.
\end{abstract}

\vfill

\newpage
\renewcommand{\baselinestretch}{1.66}

\section{Introduction}\label{sec:intro}

In observational studies with delayed entry, time-to-event outcomes are often subject to left truncation, if only subjects who have not experienced the event of interest at the study entry are included. For example, in aging studies where age is the time scale of interest, participants typically enroll at various ages.  Individuals who have experienced early events such as death will not to be included, leading to selection bias. 
Similarly in pregnancy studies, women typically enroll after their clinical recognition of pregnancies, and women with early pregnancy losses tend not to be captured \citep{xu:cham, ying2024causal}. 
In addition to left truncation, time-to-event outcomes are also commonly subject to right censoring. 

Until recently the literature on left truncation has largely  
relied on the random left truncation or quasi-independence assumption that the left truncation time and the event time are independent on the observed data region \citep{woodroofe1985estimating, wang1986asymptotic, wang1989semiparametric, wang1991nonparametric, gross1996weighted, gross1996nonparametric}. 
This may be extended to the conditional (quasi-)independence assumption under the regression setting when the dependence-inducing covariates are 
included as regressors \citep{wang1993statistical, shen2009analyzing, 
qin2011maximum}.
The resulting estimand, however, now has a conditional interpretation instead of a marginal one. 
For marginal estimands of interest, inverse probability of truncation weighting has been proposed \citep{vakulenko2022nonparametric}, and 
\citet{wang2024doubly} further derived the efficient influence function 
and doubly robust (DR) approaches in that setting. 

Right censoring in our context has two layers of complexity. First, similar to the above,
 in a regression setting covariate dependent right censoring is non-informative when the covariates are included in the regression model. 
But for a marginal estimand, covariate dependent right censoring is informative. 
Informative right censoring for estimating a marginal parameter has recently gained attention in the literature  \citep{campigotto2014impact, templeton2020informative, olivier2021informative, van2021principled, luo2023doubly}. 
Secondly, left truncation has been known to
induce dependence between the right censoring time and the event time \citep{asgharian2005asymptotic, shen2009analyzing}.  
This is because 
when loss to follow-up occurs after study entry, the censoring time and event time ``share'' a common time to enrollment into the study.

Handling covariate dependent left truncation  and right censoring (LTRC)  often presents substantial statistical challenges. 
In this paper, we introduce a liberating framework designed to specifically address these challenges. By leveraging recent advancements in orthogonal learning, our framework allows existing estimating and loss functions from LTRC-free settings to be effortlessly applied to the observed LTRC data, while
maintaining their desirable statistical properties. As a key illustration of this generality and ease of use, we apply the framework to problems in causal inference,  for the estimation of the average treatment effect (ATE) and the conditional average treatment effect (CATE). 
Here orthogonality is important when nonparametric or machine learning (ML) approaches are used to estimate the nuisance functions which might be complex, so that their slow convergence does not propagate into that of the estimation of the estimand itself. For the ATE, we develop both model DR and rate DR estimators. While the concept of model DR might be more familiar to the reader, rate DR is useful when ML approaches are used to estimate the nuisance functions. For the CATE, starting with two concrete learners for LTRC-free data, we develop orthogonal and doubly robust learners that can achieve the oracle rate of convergence under covariate dependent LTRC. Here oracle rate refers to the estimation error rate if the nuisance functions are known.

\subsection{Related work} 

There has been a very broad semiparametric literature on doubly robust approaches for missing data under the coarsening at random framework, which includes estimation of treatment effects in observational studies, as well as  handling of informative right censoring for time-to-event outcomes \citep{tsiatis2006semiparametric, rotnitzky2005inverse}. 
On the other hand, separate development in \citet{wang2024doubly} provides a DR approach based on efficient influence curve under left truncation which does not fall under the established framework of coarsened data. 

Recently
CATE has received much attention in the causal inference literature. 
When there is no left truncation or right censoring, 
\citet{nie2021quasi}, \citet{foster2023orthogonal}, and \citet{kennedy2023towards}
 provided formal guarantee that various ML and nonparametric approaches can actually succeed in identifying CATE with certain error bounds. 
These works are closely related to the orthogonal and doubly robust approaches developed under the semiparametric theory referenced above, 
but now for an infinite dimensional estimand. Just like for the finite dimensional estimand under semiparametric theory, orthogonality and double robustness guarantee that the rate of convergence to the estimand is not  affected by the estimation error of the nuisance parameters, i.e.~$n^{-1/2}$ for the finite dimensional estimand, and  oracle rate for the infinite dimensional estimand.

As will be seen, 
our  development  for CATE estimation under LTRC  is  most closely related to those of  \citet{foster2023orthogonal} and 
\citet{kennedy2023towards}. \citet{foster2023orthogonal} focused on Neyman orthogonality under sample splitting, 
while \citet{kennedy2023towards} considered doubly robust estimation. 
Here we consider both.  
Our learners do not fall under the pseudo-outcome regression framework studied in \citet{kennedy2023towards}. 
In addition, \citet{morzywolek2023general} studied a general class of Neyman orthogonal learners and their error bounds, which included the T-, R- and DR-learners as special cases, and  provided insights into the relationship among them.

Finally in the absence of truncation, for CATE estimation with right censored time-to-event outcomes, various ML approaches have been considered, 
and we refer to \citet{xu2023treatment} for a comprehensive review and tutorial on how the metalearners, including the S-, T-, X-, M-, and R-learners, can be extended to handle right censoring. 
None of the above approaches, however,  leveraged Neyman orthogonality or doubly robust estimation. 
More recently, \citet{cui2023estimating} developed causal survival forests that adjust for right censoring using  orthogonal estimating equations.

\subsection{Organization of the Paper}

In this paper we provide a unified approach to handle covariate dependent LTRC, which can be applied to the estimation of both ATE and CATE. 
Section \ref{sec:preliminary} introduces the problem setup, assumptions, and relevant martingales.
In Section \ref{sec:orthogonal_DR_appraoches} we describe an orthogonal and doubly robust operator for handling LTRC, so that it preserves double robustness for estimation of the  ATE and  orthogonality and double robustness for estimation of the CATE. 
Section \ref{sec:nuisance_estimation} discusses the estimation of the nuisance parameters 
and studies the theory for both the ATE and the CATE estimators. In 
Section \ref{sec:simu} we study their finite sample performance using simulation. In Section \ref{sec:application}  the proposed methods are applied to analyzing the effect of mid-life heavy drinking on late life cognitive impairment free survival using data from the Honolulu Asia Aging Study. 
Section \ref{sec:discussion} concludes with a discussion.

\section{Preliminaries} \label{sec:preliminary}

\subsection{Estimands and assumptions} \label{sec:notation}

In the following  we use variables with superscript asterisk ``*'' to denote those in the population before truncation, and without ``*'' to denote variables of the observed truncated data. Let $A^*\in\{1,0\}$ denote the binary treatment assignment and $Z^*$ the baseline covariates. 
Following the potential outcome framework 
\citep{neyman1923application, rubin1974estimating}, for $a\in\{1,0\}$ let $T^*(a)$, $Q^*(a)$, and $C^*(a)$ denote the potential event time of interest, the potential left truncation time, and the potential censoring time had the subject received treatment $a$. 
For exposition we will assume $C^*(a)>Q^*(a)$ almost surely, and discuss otherwise in the last section. 
Let $D^*(a) = C^*(a)-Q^*(a)$ denote the residual censoring time. 
Let $Q^* $, $T^*$, $C^*$  
and $D^*$  be the corresponding ``factual'' outcomes once treatment $A$ is received. 
A subject is observed only if $Q^*<T^*$. 

We consider $(T^*(1), T^*(0), A^*, Z^*)$ as the full data,   $(T^*, A^*, Z^*)$  as the LTRC-free data, $(Q, T, A, Z)$ as the truncated but censoring-free data, and finally $O = (Q, X, \Delta, A, Z)$  the observed data, where $X = \min(T,C)$ and $\Delta = \ind(T<C)$.

We are interested in the following estimand which is the conditional treatment effect (CATE): 
\begin{align}
    \tau(v) = \E\left[\nu\{T^*(1)\} - \nu\{T^*(0)\} | V^* = v\right], \label{eq:CATE_definition}
\end{align}
where $\nu$ denotes a known bounded transformation, 
and $V^*$ is generally a subset of $Z^*$. For example, 
when $\nu(t) = \ind(t>t_0)$ for some $t_0$, 
$E[ \nu\{T^*(a)\} | V^*=v ] = \PP\{T^*(a) > t_0 | V^*=v\}$; 
when $\nu(t) = \min(t,t_0)$, $\tau(v)$ 
 is the difference of the potential conditional restricted mean survival times (RMST). 
For the special case $V = \varnothing$ (the empty set), $\theta = \tau(\varnothing) = \E\left[\nu\{T^*(1)\} - \nu\{T^*(0)\}\right]$ is the ATE.  
While a special case, the estimation of the scalar ATE is very different from that  of the funcational CATE, which will be treated separately via estimating and loss functions, respectively; nonetheless, our unifying framework for handling LTRC also shows that there is much parallel between the two as will be seen below.

Let $F(t|a,z)$, $G(t|a,z)$ and $H(a,z)$ denote the conditional cumulative distribution function (CDF) of $T^*$ given $(A^*,Z^*)$, $Q^*$ given $(A^*,Z^*)$, and the CDF of $(A^*,Z^*)$, respectively.
 Let $\pi(z) = \PP(A^*=1|Z^*=z)$ be the propensity score, 
and $S_D(t|Q,A,Z) = \PP(D>t|Q,A,Z)$. 
Finally let  
$\mu(a,z;F) = \E\{ \nu(T^*)|A^*=a,Z^*=z\} = \int_0^\infty \nu(t) dF(t|a,z)$, 
and $\beta = \PP(Q^*<T^*)$.

Throughout this paper, 
we assume the Stable Unit Treatment Value Assumption (SUTVA) that each subject's potential times are not affected by other subjects' treatments and that there is only one version for each treatment level.
We also make the following assumptions. 
\begin{assumption}[Consistency]\label{ass:consistency}
	$T^* = A^*T^*(1) +(1-A^*)T^*(0)$, 
 $Q^* = A^*Q^*(1) +(1-A^*)Q^*(0)$, $C^* = A^*C^*(1) +(1-A^*)C^*(0)$, which implies  
 $D^* = A^*D^*(1) +(1-A^*)D^*(0)$.
\end{assumption}
\begin{assumption}[No unmeasured confounding]\label{ass:noUnmConf}
	$A^*\bigCI \left\{T^*(a), Q^*(a), D^*(a)\right\} \mid Z^*$.
\end{assumption}
\begin{assumption}[Conditional independent truncation]\label{ass:trunc}
	$Q^*(a) \bigCI T^*(a) \mid Z^*$.
\end{assumption}
\begin{assumption}[Conditional independent residual censoring]\label{ass:cen}
$D^*(a) \bigCI \{T^*(a), Q^*(a)\}\mid Z^*$. 
\end{assumption}
\begin{assumption}[Strict positivity] \label{ass:strict_positivity}
There exist $\delta>0$ and $0<\tau_1<\tau_2<\infty$ such that (i) $\delta\leq \pi(Z^*)\leq 1-\delta$ a.s.; (ii) $T^*\geq \tau_1$ a.s.~and $Q^*<\tau_2$ a.s., $1-F(\tau_2|A^*,Z^*) \geq \delta$ a.s.,  $G(\tau_1|A^*,Z^*)\geq \delta$ a.s., and $S_D(\tau_2-Q|Q,A,Z) \geq \delta$ a.s..
\end{assumption}
Assumptions \ref{ass:consistency}, \ref{ass:noUnmConf}, and \ref{ass:strict_positivity}(i) are commonly used assumptions in the causal inference literature. 
Assumption \ref{ass:trunc} can be weakened to a similar conditional quasi-independence assumption 
 as in \citet{wang2024doubly}. 
 Assumption \ref{ass:cen} can be weakened to the conditional independent censoring assumption in the truncated data:
 $D \bigCI (T-Q) \mid Q,A,Z$, which was used in \citet{cheng2012estimating} and \citet{morenz2024debiased}  (although \citet{cheng2012estimating} appeared to have  a typo and only stated $D \bigCI (T-Q) \mid A,Z$). 
Assumption \ref{ass:strict_positivity}(ii) 
on the upper and lower tails of the distributions  are commonly imposed for left truncated \citep{wang1989semiparametric, wang1991nonparametric} and right censored data, and ensures that the left and right tails of  $Q^*$, $T^*$ and $C^*$ can be identified and estimated. 
Similar  assumptions were also made in  \citet{wang2024doubly} 
and   \cite{luo2023doubly}.

\subsection{Relevant martingales}\label{sec:martingales}

Here we introduce the martingales associated with $Q$ and $D$, which will be useful later.
For $t\geq 0$, let
$\bar N_Q(t) = \mathbbm{1}(t\leq Q<T)$,
and let $\bar\lambda_Q$ denote the conditional {\it reverse time hazard function} of $Q$ given $Z$ in the LTRC-free data: 
\begin{align}
\bar\lambda_Q(q|a,z) &= \lim_{h\to 0+} \frac{\PP(q-h< Q^*\leq q|Q^*\leq q, A^*=a, Z^* = z)}{h}. \label{eq:alpha_def} 
\end{align}
Then following a similar proof as in \citet{wang2024doubly},  we have that 
\begin{align}
    \bar M_Q(t; G) = \bar N_Q(t) - \int_t^\infty \mathbbm{1}(Q\leq s <T) \bar\lambda_Q(s|A,Z) ds \label{eq:def_MQ}
\end{align}
is a {\it backwards martingale} 
 in the truncated but censoring-free data 
with respect to $\{\mathcal{\bar F}_t\}_{t\geq 0}$, where $\mathcal{\bar F}_t = \sigma \left\{A,Z, \mathbbm{1}(Q<T), \mathbbm{1}(s\leq T),\mathbbm{1}(s\leq Q<T): s\geq t \right\}$.

We now introduce the martingale associated with $D$ in the observed data.
Let 
$N_D(t) = \mathbbm{1}(X-Q\leq t, \Delta = 0)$,
and let $\lambda_D(t|Q,A,Z)$ denote the conditional hazard function of $D$ given $(Q,A,Z)$. Then 
\begin{align}
M_D(t;S_D) = N_D(t) - \int_0^t \mathbbm{1}(X-Q\geq u) \lambda_D(u|Q,A,Z) du
\end{align}
is a martingale with respect to the filtration $\{\mathcal{F}_t\}_{t\geq 0}$, where $\mathcal{F}_t = \sigma\{Q,A,Z, \ind(T\leq u),\ind(C\leq v): u\leq t,v\leq t \}$.

\section{Orthogonal and DR approaches under LTRC}\label{sec:orthogonal_DR_appraoches}

\subsection{Doubly robust LTRC operators}\label{sec:AIPW_operators}

We propose a doubly robust operator for handling covariate-dependent left truncation and right censoring, such that if we have an estimating or loss function for the LTRC-free data, we can obtain its counterpart under LTRC, to be elaborated below. 
In this paper we apply the operator to causal inference; however, it can be more broadly applied to estimation problems under LTRC based on existing estimating functions or loss functions in the LTRC-free data. 

Motivated by {the augmented inverse probability of truncation weighting approach developed in \citet{wang2024doubly}},
we consider the following operator $\V_Q$ for handling left truncation. The operator $\V_Q$ maps an arbitrary {$d$-dimensional} function {$\zeta= \zeta(T^*,A^*,Z^*)$} of the LTRC-free data to a {$d$-dimensional} function of the truncated but censoring-free data $ (Q,T,A,Z) $: 
\begin{align}
    \V_Q(\zeta; F,G) 
    & = \frac{\zeta(T,A,Z)}{G(T|A,Z)} -\int_0^\infty  m_\zeta(v,A, Z; F) \cdot  \frac{F(v|A,Z)}{ 1-F(v|A,Z) } \cdot \frac{d\bar M_Q(v;G)}{G(v|A,Z)}, \label{eq:AIPW_Q}
\end{align}
where $m_\zeta(v,a,z;F) = \E\{\zeta(T^*,A^*,Z^*) \mid T^*<v, A^*=a, Z^*=z\} = \int_0^v \zeta(t,a,z) d F(t|a,z)/F(v|a,z)$.
The operator $\V_Q$ takes the form of  augmented inverse probability weighting (AIPW), where the first term is the inverse probability of truncation weighted $\zeta$. 
The doubly robust estimating function developed in \citet{wang2024doubly} is a special case of \eqref{eq:AIPW_Q} with $\zeta$ equal to $\nu(T^*) $ minus the estimand which is $E\{ \nu(T^*) \}$.

Next we adapt the augmented inverse probability of censoring weighting \citep[AIPCW]{rotnitzky2005inverse} approach to the residual censoring time, and  consider the operator $\V_C$ that maps an arbitrary {$d$-dimensional} function {$\xi = \xi (Q,T,A,Z) $} of the truncated but censoring-free data  to a {$d$-dimensional} function of the observed data $(Q,X,\Delta,A,Z) $:
\begin{align}
    \V_C(\xi; F,S_D) 
    &= \frac{\Delta \ \xi(Q,X,A,Z)}{S_D(X-Q|Q,A,Z)} 
    +  \int_0^\infty \bar m_\xi(u,Q,A,Z;F) \cdot \frac{dM_D(u;S_D)}{S_D(u|Q,A,Z)},  \label{eq:AIPCW}
\end{align}
where 
\begin{eqnarray*}
\bar m_\xi(u,q,a,z;F) &=&  \E\{\xi(Q,T,A,Z)|T-Q \geq u,Q=q,A=a,Z=z\} \\
&=&  \E\{\xi(Q^*,T^*,A^*,Z^*)|T^*-Q^* \geq u,Q^*=q,A^*=a,Z^*=z\} \\
&=& \frac{ \int_{q+u}^\infty \xi(q,t,a,z) dF(t|a,z) }{1-F(q+u|a,z)}.
\end{eqnarray*}
The last equality above holds because  $Q^*\bigCI T^*\mid A^*,Z^*$, which is implied by Assumptions \ref{ass:consistency} - \ref{ass:trunc}.
We note that the operator $\V_C$ naturally 
handles the truncation-induced dependence between the right censoring time and the event time by considering the residual censoring time.

Finally we combine the above two operators and obtain $\V = \V_C \circ\V_Q$, which maps an arbitrary {$d$-dimensional} function $\zeta$ of the LTRC-free data $(T^*,A^*,Z^*)$ to a {$d$-dimensional} function of the observed data $(Q,X,\Delta,A,Z)$:
\begin{align}
    \V(\zeta; F,G,S_D) = \V_C\{ \V_Q(\zeta; F,G); F, S_D \}. 
\end{align}
In this paper we have $d = 1$ for both the ATE and the CATE. 

We say that the operator $\V$ is doubly robust, in the sense stated in Lemma \ref{lem:DR_V} below. 
We will use subscript ``0'' to denote the true parameters for the rest of the paper.

\begin{lemma}[Double robustness of $\V$] \label{lem:DR_V}
    Under Assumptions \ref{ass:consistency} -  \ref{ass:cen} and  \ref{ass:strict_positivity}(ii), for any function $\zeta(T^*,A^*,Z^*)$ that is bounded almost surely,
\begin{align}
    \E\{\V(\zeta; F,G,S_D) 
    \} = \beta^{-1} \E\{\zeta(T^*,A^*,Z^*)\} \nonumber
\end{align}
if either $F = F_0$ or $(G,S_D) = (G_0,S_{D0})$; recall that $\beta = \PP(Q^*<T^*)$. 
\end{lemma}

\begin{remark}
    In fact, each of the operators $\V_Q$ and $\V_C$ is doubly robust. In particular, we show in Appendix \ref{app:DR_lemmas_proofs} that the operator $\V_Q$ maintains the expectation of the input function (up to a constant factor) in the LTRC-free data if either $F = F_0$ or $G = G_0$; and the operator $\V_C$ maintains the expectation of the input function in the truncated but censoring-free data if either $F = F_0$ or $S_D = S_{D0}$. 
\end{remark}

Lemma \ref{lem:DR_V} separates the nuisance parameters into two sets, where $F$ is the conditional event time distribution, and $(G,S_D)$ represent the missing data mechanism. As will be seen below, with a proper nuisance estimation scheme, 
Lemma \ref{lem:DR_V} enables a general doubly robust approach to using existing estimating functions or loss functions in the LTRC-free data, {while preserving the original function's orthogonality or double robustness}.
This allows us to construct estimators where the estimation errors of the nuisance parameters only have higher order impacts on the estimation error rate of the estimand, and therefore enables the estimation to achieve fast rate of convergence when incorporating nonparametric or ML approaches that might have slower rates of convergence.  
As mentioned above, this framework is not restricted to the estimation of treatment effects but can also be applied to other  problems with LTRC data.

\subsection{DR estimating function for ATE}\label{sec:DR_EF_ATE}

We first consider the average treatment effect  $\theta = \tau(\varnothing) = \E\left[\nu\{T^*(1)\} - \nu\{T^*(0)\} \right]$. 
In the absence of LTRC, an augmented inverse probability of treatment weighting (AIPTW) approach \citep{van2003unified, tsiatis2006semiparametric} has been proposed for $\theta$ and the estimating function for a single copy of the data is 
\begin{align}
    U_0(\theta; \pi,F) 
    & =  \frac{A^* - \pi(Z^*)}{\pi(Z^*)\{1-\pi(Z^*)\}} \left\{ \nu(T^*) - \mu(A^*,Z^*;F) \right\}  + \mu(1,Z^*;F) - \mu(0,Z^*;F) - \theta, \label{eq:U0}
\end{align}
where $\mu(a,z;F) = \E\{ \nu(T^*)|A^*=a,Z^*=z\} = \int_0^\infty \nu(t) dF(t|a,z)$.
It is known that $ U_0(\theta; \pi,F)$  is doubly robust in the sense that its expectation is zero at $\theta_0$ if either $F= F_0$ or $\pi = \pi_0$. 

Applying $\V$ to $U_0$ we obtain an estimating function for $\theta$ under LTRC: 
\begin{align}
    U(\theta; \pi,F,G,S_D)  
    & = \frac{A - \pi(Z)}{\pi(Z)\{1-\pi(Z)\}} \left\{ \V(\nu; F,G,S_D) - \V(1;F,G,S_D)\mu(A,Z;F) \right\} \nonumber \\
    & \quad + \V(1;F,G,S_D) \{\mu(1,Z;F) - \mu(0,Z;F) - \theta\}, \label{eq:AIPW_3bias}
\end{align}
where the expressions for $\V(\nu; F,G,S_D)$ and $\V(1; F,G,S_D)$, 
as well as the derivation of $U$ is in Appendix \ref{sec:crossprod_norm_properties}.

Lemma \ref{lem:DR_3bias} below shows that $U$ is doubly robust. 
\begin{lemma}[Double robustness of $U$] \label{lem:DR_3bias}
    Under Assumptions \ref{ass:consistency} -  \ref{ass:strict_positivity}, 
    $\E\left[U(\theta_{0}; \pi,F,G,S_D) \right] =0$
    if either $F = F_0$ or $(\pi,G,S_D) = (\pi_0,G_0,S_{D0})$.
\end{lemma}

\subsection{Orthogonal and DR loss for CATE}\label{sec:orthogonal_loss}

We first review the more general orthogonal loss and  doubly robust loss functions for the LTRC-free data \citep{foster2023orthogonal, kennedy2023towards}. 
Let $\ell(T^*,A^*,Z^*; \tau,\eta)$ be a loss function for $\tau$ such that 
$\tau_0 = \min_{ \tau\in\T} \E \{\ell(T^*,A^*,Z^*; \tau,\eta_0)\}$, with possible nuisance parameter $\eta\in \mathcal{N}$. 
For a vector space $\mathcal{F}$  of functions and 
a functional $\Phi:\mathcal{F} \to \mathbb{R}$,   define
 the directional derivative $D_f \Phi(f)[h] = d \Phi(f+th)/ dt |_{t=0}$ for any $f,h\in\mathcal{F}$.
We say that $\ell$ is {\it Neyman orthogonal} \citep{neyman1959optimal,neyman1979c} if $D_\eta D_\tau \E\{\ell(T^*,A^*,Z^*; \tau,\eta)\}[\tau-\tau_0, \eta - \eta_0] = 0$ for any $\tau\in\T$ and $\eta\in\mathcal{N}$.
Note that when $\tau$ is finite dimensional, the derivative $D_\tau \E\{\ell(T^*,A^*,Z^*; \tau,\eta)\}$ can be viewed as an estimating function for $\tau$. 
In addition, if $\eta = (\eta_1,\eta_2)$, we say that $\ell$ is  {\it doubly robust} if 
$\E\{\ell(T^*,A^*,Z^*; \tau,\eta_{10},\eta_2)\} = \E\{\ell(T^*,A^*,Z^*; \tau,\eta_{1},\eta_{20})\} = \E\{\ell(T^*,A^*,Z^*; \tau,\eta_{10},\eta_{20})\}$.

Applying the operator $\V$ to $\ell$, 
we obtain a loss function in the observed data  under LTRC:
  $\tilde\ell(\tau;\eta, F,G,S_D) = \V(\ell;F,G,S_D).$ 
Lemma \ref{lem:orthogonal_DR_loss}  
below shows that $\V$ preserves the orthogonality and double robustness of $l$.

\begin{lemma}[$\V$ preserves orthogonality and DR of loss]\label{lem:orthogonal_DR_loss}
    Under Assumptions \ref{ass:consistency} -  \ref{ass:cen} and \ref{ass:strict_positivity}(ii),
   \\
    (i) if $\ell(T^*,A^*,Z^*;\tau,\eta)$ is a Neyman orthogonal loss in the LTRC-free data, then $\tilde\ell(\tau;\eta,F,G,S_D)$
     is a Neyman orthogonal loss in the observed data, with $\eta,F,G,S_D $ as the nuisance parameters; \\
    (ii) if $\eta = (\eta_1,\eta_2)$ and $\ell(T^*,A^*,Z^*; \tau,\eta_1,\eta_2)$ is a doubly robust loss with respect to $\eta_1$ and $\eta_2$ in the LTRC-free data, 
    then 
    $\tau_0 = \argmin_\tau \E\{\tilde\ell(\tau;\eta_1,\eta_2,F,G,S_D)\}$
    if (a) either $\eta_1 = \eta_{10}$ or $\eta_2 = \eta_{20}$, and (b) either $F = F_0$ or $(G,S_D) = (G_0,S_{D0})$. 
\end{lemma}
We note that whether part (ii) above translates to double robustness in estimation depends on the relationship between the parameters $(\eta_1,\eta_2) $ and $(F,G,S_D) $, as elaborated in the next section. 

In the following  we consider two specific loss functions  for CATE without LTRC: the R-loss and the DR-loss. 
For the special case of $V^* = Z^*$, the R-loss considered in \cite{nie2021quasi} 
is:  
\begin{align}
    \ell_\RR(\tau;\pi,F) 
    & = \left[\{\nu(T^*) - \tilde\mu(Z^*;\pi,F)\} - \{A^*-\pi(Z^*)\} \tau(Z^*) \right]^2, \label{eq:R_loss}
\end{align}
where $\tilde\mu(z;\pi,F) = \mu(1,z;F)\pi(z) + \mu(0,z;F)\{1-\pi(z)\}$. It is a Neyman orthogonal loss for the CATE defined in \eqref{eq:CATE_definition}. 
We note that when $V^*$ is a proper subset of  $Z^*$,  \citet{morzywolek2023general} showed that  \eqref{eq:R_loss} with $\tau(Z^*)$   replaced by $\tau(V^*)$, leads to a weighted CATE estimand.  
In the following we will only consider $V^* = Z^*$ for the R-loss.

\cite{kennedy2023towards} considered 
the DR-loss:  
\begin{align}
    \ell_{\DR}(\tau;\pi,F) 
    &= \left[ \frac{A^* - \pi(Z^*)}{\pi(Z^*)\{1-\pi(Z^*)\}} \{ \nu(T^*) - \mu(A^*,Z^*; F) \} + \mu(1,Z^*; F) - \mu(0,Z^*; F) - \tau(V^*)\right]^2.  \label{eq:DR_loss}
\end{align}
It is Neyman orthogonal and doubly robust in terms of $F$ and $\pi$. 
We note that for the special case with $V^* = \varnothing$ and $\theta = \tau(\varnothing)$  the ATE,  $\partial\ell_{\DR}(\theta;\mu, \pi)/\partial \theta$ gives the AIPTW estimating function \eqref{eq:U0} for $\theta $ in the LTRC-free data.

Applying $\V$ to the R- and  DR-loss functions, we obtain (after dropping a function of the observed data $O$ that does not involve $\tau$; see Appendix \ref{app:derivations}): 
\begin{align}
    \tilde\ell_\RR(\tau;\pi,F,G,S_D) 
     & = \V(1;F,G,S_D) 
      \left[ \frac{\V(\nu;F,G,S_D)}{\V(1;F,G,S_D)} - \tilde\mu(Z;\pi,F) - \{A-\pi(Z)\}\tau(Z) \right]^2,  
\label{eq:ltrcR-loss} \\
    \tilde\ell_\DR(\tau;\pi,F,G,S_D) 
     & = \V(1;F,G,S_D) 
      \left[ \frac{A - \pi(Z)}{\pi(Z)\{1-\pi(Z)\}} \left\{ \frac{\V(\nu;F,G,S_D)}{\V(1;F,G,S_D)} - \mu(A,Z;F) \right\} \right. \nonumber \\
     & \qquad\qquad\qquad\qquad\quad  + \mu(1,Z;F) - \mu(0,Z;F) - \tau(V) \Bigg]^2. 
    \label{eq:ltrcDR-loss}
\end{align}
We refer to the above as the  ltrcR-loss and the ltrcDR-loss, respectively.
By the Neyman orthogonality of the R- and DR-loss,  and the double robustness of the DR-loss in the LTRC-free data, we have immediately the following corollary  from Lemma \ref{lem:orthogonal_DR_loss}.

\begin{corollary}\label{cor:ltrcR-ltrcDR-loss}
    Under Assumptions \ref{ass:consistency} - \ref{ass:strict_positivity}, we have \\
    (i) The ltrcR-loss \eqref{eq:ltrcR-loss} and the ltrcDR-loss \eqref{eq:ltrcDR-loss} are Neyman orthogonal  in the observed data; \\
    (ii) The ltrcDR-loss \eqref{eq:ltrcDR-loss} is doubly robust in that for any $\tau\in\T$,
    $$\E\{\tilde\ell_{\DR}(\tau;\pi,F_0,G,S_D)\} = \E\{\tilde\ell_{\DR}(\tau;\pi_0,F,G_0,S_{D0})\} = \E\{\tilde\ell_{\DR}(\tau;\pi_0,F_0,G_0,S_{D0})\}.
    $$
\end{corollary}

\section{Estimation and theory} \label{sec:nuisance_estimation}

The doubly robust estimation of the ATE using $U$ in \eqref{eq:AIPW_3bias}, or the orthogonal and  doubly robust estimation of the CATE using the ltrcR-loss \eqref{eq:ltrcR-loss} or the ltrcDR-loss \eqref{eq:ltrcDR-loss}, involves two-stage algorithms where in the first stage, the four  nuisance parameters $\pi$, $F$, $G$ and $S_D$ are estimated.  
Figure \ref{fig:DGM_nuisance_summary} in Appendix \ref{supp:nuis} 
shows four different ways to estimate the nuisance parameters, with two each for estimating $G$ and $\pi$, respectively. This is because the estimation of $G$ needs to account for censoring \citep{wang1991nonparametric}, and the estimation of $\pi$ needs to account for 
truncation \citep{cheng2012estimating}; 
more details are also explained in the Appendix.  

The double robustness property of $U$ and $\tilde\ell_\DR$ seperates the nuisance parameters into $F$ and $(\pi,G,S_D)$. In Figure \ref{fig:DGM_nuisance_summary} only scheme (a)  preserves the double robustness because the estimation of $(\pi,G,S_D)$ is separate from the estimation of $F$.  Scheme (b) breaks doubly robustness because the estimation error of $F$ 
will propagate to the estimation of $\pi$ and $G$; schemes (c) and  (d) also breaks doubly robustness for similar reasons. 
For the rest of this paper we only consider scheme (a).


\subsection{ATE}\label{sec:ATE}

With a random sample of size $n$ and  
after estimating the nuisance parameters, we have the estimating equation 
  $   \sum_{i=1}^n U_{i}(\theta; \hat\pi,\hat F,\hat G,\hat S_D) = 0$. The solution  $\hat\theta$ has a closed form which is provided in Appendix \ref{supp:expression}.

When all the nuisance estimators are asymptotically linear, which is often the case when parametric or semiparametric models are used, the estimator $\hat\theta$ enjoys model double robustness, as formally stated in Theorem \ref{thm:mdr} below with the proof in the Appendix.
We use $\convp$ and $\convd$ to denote convergence in probability and in distribution, respectively.
Assumption \ref{ass:uniformcons1} in the Appendix assumes that  $\hat\pi, \hat F, \hat G, \hat S_D$ converge uniformly to $\pi^\divideontimes,F^\divideontimes,G^\divideontimes,S_D^\divideontimes$, respectively. 

\begin{theorem}[Model double robustness]\label{thm:mdr}
	(i) Under Assumptions \ref{ass:consistency} 
    - \ref{ass:strict_positivity} and \ref{ass:uniformcons1} in the Appendix, assume 
  additionally that $\pi^\divideontimes,F^\divideontimes, G^\divideontimes,S_D^\divideontimes$ also satisfy Assumption \ref{ass:strict_positivity}. 
 Then   $\hat \theta \convp \theta_{0}$ if either $F^\divideontimes = F_0$ or $(\pi^\divideontimes,G^\divideontimes,S_D^\divideontimes) = (\pi_0,G_0,S_{D0})$.
(ii) In addition, if $\hat\pi, \hat F, \hat G, \hat S_D$ are asymptotically linear satisfying Assumption \ref{ass:AL} in the Appendix, then
	$\sqrt{n}(\hat \theta - \theta_{0}) \convd N(0, \sigma^2)$.
	Furthermore, when both $F^\divideontimes = F_0$ and $(\pi^\divideontimes,G^\divideontimes,S_D^\divideontimes) = (\pi_0,G_0,S_{D0})$, we have $\sigma^2 = \beta^2 \E\{U(\theta_{0}; \pi_0,F_0,G_0,S_{D0})^2\}$, which can be consistently estimated by 
$  \hat\sigma^2 = n  \sum_{i=1}^n U_{i}^2( \hat\theta; \hat\pi, \hat F, \hat G, \hat S_D) /
  \{ \sum_{i=1}^n \V_i(1;\hat F,\hat G,\hat S_D)\}^2  $. 
\end{theorem}

In practice, flexible nonparametric or ML methods are increasingly being used. 
These methods typically have slower than  $n^{-1/2}$ convergence rate and, when used  to estimate the nuisance parameters, cross-fitting  \citep{hasminskii1978nonparametric, bickel1982adaptive, robins2008higher, chernozhukov2018double} is needed in order to obtain $n^{-1/2}$ consistent estimators of the ATE.  
The resulting estimator, denoted by $\hat\theta_{cf}$,  enjoys a rate double robustness property \citep{rotnitzky2021characterization, hou2021treatment, wang2024doubly} under the assumptions below.

With norms and integral products defined in Appendix \ref{app:integral_prod}, we have the following. 
\begin{assumption}[Uniform Convergence]\label{ass:uniformcons2}
	$ \|\hat\pi- \pi_0\|_{2} = o_p(1) $, 
    $ \|\hat F- F_0\|_{\sup, 2} = o_p(1) $, 
    $ \|\hat G- G_0\|_{\sup, 2} = o_p(1)$, 
  and  $ \|\hat S_D - S_{D0}\|_{\sup, 2} = o_p(1) $. 
\end{assumption}
\begin{assumption}[Product rate condition] \label{ass:prodrate}
\begin{align}
    & \|\hat F-F_0\|_{\sup,2} \cdot \left\{ \|\hat\pi - \pi_0\|_{2} + \|\hat G-G_0\|_{\sup,2} + \|\hat S_D - S_{D0}\|_{\sup,2} \right\} \nonumber  \\
    &\qquad\qquad\quad + \|K_1(\hat g,g_0)\|_1 + \|K_2(\hat g,g_0)\|_1 + \|K_3(\hat g,g_0)\|_1 = o_p(n^{-1/2}). \nonumber 
\end{align}
\end{assumption}
In the above $g = (F,G,S_D)$ and $g_0 = (F_0,G_0,S_{D0})$, and $K_1(g, g_0)$, $K_2(g, g_0)$, and $K_3(g, g_0)$ are integral products between the estimation errors of $F$ and $(G, S_D)$, with detailed expressions  in Appendix \ref{app:integral_prod}. 
These integral products are encountered when considering  rate DR for time-to-event data where more than one set of nuisance parameters are functions of time  \citep{ying2023cautionary,  wang2024doubly, westling2024inference, luo2023doubly}.

The theorem below assumes $K$-fold cross-fitting, with  the data partitioned into 
$\I_1,...,\I_K$; 
$\hat\pi^{(-k)}, \hat F^{(-k)}, \hat G^{(-k)}, \hat S_D^{(-k)}$ are estimated using the out-of-$\I_k$ data ($k=1,...,K$). The estimator $\hat\theta_{cf}$ is then obtained by solving 
$     \sum_{k=1}^K \sum_{i\in \I_k} U_i\left(\theta; \hat\pi^{(-k)}, \hat F^{(-k)}, \hat G^{(-k)}, \hat S_D^{(-k)}\right) = 0$.
\begin{theorem}[Rate double robustness] \label{thm:rdr}
    Under Assumptions \ref{ass:consistency} - 
    \ref{ass:uniformcons2},  $\hat\theta_{cf} \convp \theta_{0}$. In addition, if Assumption \ref{ass:prodrate} holds, then 
    $n^{1/2}(\hat\theta_{cf} - \theta_{0}) \convd N(0, \sigma^2)$, where $\sigma^2 = \beta^2 \E\{U(\theta_{0}; \pi_0,F_0,G_0,S_{D0})^2\}$, which can be consistently estimated by 
     $\hat\sigma_{cf}^2 = \\ 
 n  \left.  \sum_{i,k} U_{i}^2 \left( \hat\theta_{cf}; \hat\pi^{(-k)}, \hat F^{(-k)}, \hat G^{(-k)}, \hat S_D^{(-k)}  \right) \right/ \left\{  \sum_{i,k} \V_i(1;\hat F^{(-k)},\hat G^{(-k)},\hat S_D^{(-k)}) /n \right\}^2$.  
\end{theorem}

The proof of Theorem \ref{thm:rdr} is in the Appendix.
Being an influence function (i.e.~orthogonal score),  with cross-fitting $U$ allows $n^{-1/2}$ inference for $\theta$ if the nuisance parameters are estimated at faster than  $n^{-1/4}$ rate \citep{rotnitzky2021characterization}.  
Rate DR, on the other hand,  
only requires that the integral product  rates to be faster than $n^{-1/2}$, without requiring that each nuisance parameters be estimated at faster than $n^{-1/4}$ rate.

\subsection{CATE}\label{sec:CATE}

In the following we show that both the ltrcR-loss \eqref{eq:ltrcR-loss} and the ltrcDR-loss \eqref{eq:ltrcDR-loss} allow oracle rate estimation of the CATE under suitable assumptions; that is, the rate is unaffected by the first-stage nuisance estimation rate. 
For the theory we consider empirical risk minimization  with sample splitting; that is,  
$
\hat\tau = \argmin_{\tau\in\T} \left[ \sum_{i=1}^k \tilde\ell_i(\tau; \hat\pi,\hat F,\hat G,\hat S_D) / k \right]
$, where $\hat\pi,\hat F,\hat G$ and $\hat S_D$ are estimated from data $\{O_i:i=k+1,...,n\}$, and $k$ is roughly $n/2$.

Our main result here utilizes an extension (Lemma \ref{lem:CATE_rate_general} in the Appendix) of the  error bound in \citet{foster2023orthogonal}. 
Since both the ltrcR-loss  and the ltrcDR-loss can be seen as a  weighted squared loss (Appendix \ref{sec:CATE_theory}), and may have negative weights in finite samples, our extension replaces the global strong convexity in \citet{foster2023orthogonal} by a relaxed version. 
In addition, the extension avoids verifying the second order directional derivative of the population risk with respect to the nuisance parameters, which can be complex for cases such as ours. 
The extension in fact allows us to work with a general loss function that may not be Neyman orthogonal; when this is the case, it follows that the nuisance estimation errors have first order impact on the CATE estimation. When the loss function is Neyman orthogonal, the nuisance estimation error is shown to have only second order impact on CATE estimation, which recovers the results in \citet{foster2023orthogonal}. Additionally, if the loss function is also doubly robust, the nuisance estimation errors can be shown to only impact the CATE estimation error through the product and integral product errors. 

In addition to the extension above, we also use Rademacher complexity and critical radius 
\citep{bartlett2005local, wainwright2019high, foster2023orthogonal}
to instantiate oracle rates for the ltrcR-learner and ltrcDR-learner. 
Again with norms and integral products defined in Appendix \ref{app:integral_prod}, we have the following. 
\begin{assumption}\label{k0_converge_to_0}
\begin{align*}
    & \|\hat F-F_0\|_{\infty} \cdot\left\{ \|\hat G-G_0\|_{\infty} + \|\hat S_D - S_{D0}\|_{\infty} \right\} \nonumber \\
    &\qquad + \|\tilde K_1(\hat g,g_0;1)\|_{\infty} + \|\tilde K_2(\hat g,g_0;1)\|_{\infty} + \|\tilde K_3(\hat g,g_0;1)\|_{\infty} = o_p(1). 
\end{align*}
\end{assumption}
The above is a mild assumption, and is used to relax the  strong convexity assumption in \citet{foster2023orthogonal}. 

\begin{theorem}\label{thm:CATE_ltrcR}
    Under Assumptions \ref{ass:consistency}
    - \ref{ass:strict_positivity} and  \ref{k0_converge_to_0},
     suppose that  
      $\sup_{\tau\in\T, v} |\tau(v)| < \infty$, and that all the nuisance parameters in $\Gc$ satisfy Assumption \ref{ass:strict_positivity}.
 Using the ltrcR-loss \eqref{eq:ltrcR-loss}    $\hat\tau_{\RR}$ achieves the oracle error rate if 
\begin{align}
        &  \|\pi - \pi_0\|_{4}^2 + \|F-F_0\|_{\sup,4}\cdot\left\{ \|\pi - \pi_0\|_{4}+ \|G-G_0\|_{\sup,4} + \|S_D - S_{D0}\|_{\sup,4} \right\}  \nonumber \\
    &\quad + \|K_1(g,g_0)\|_2 + \|K_2(g,g_0)\|_2 + \|K_3(g,g_0)\|_2 = o_p(n^{-1/2}).  
    \label{eq:k_R}
\end{align}
\end{theorem}
From the above,  $\hat\tau_{\RR}$ achieves the oracle error rate if $\|\hat\pi - \pi_0\|_{4} = o_p(n^{-1/4})$ and the product and integral product rates of the nuisance estimation errors between $\hat F$ and $(\hat G,\hat\pi,\hat S_D)$ in \eqref{eq:k_R} are faster than $n^{-1/2}$.

\begin{theorem}\label{thm:CATE_ltrcDR}
    Under Assumptions \ref{ass:consistency}
    - \ref{ass:strict_positivity} and  \ref{k0_converge_to_0},
     suppose that 
     $\sup_{\tau\in\T, v} |\tau(v)| < \infty$, and that all the nuisance parameters in $\Gc$ satisfy Assumption \ref{ass:strict_positivity}.
Using the ltrcDR-loss \eqref{eq:ltrcDR-loss}     $\hat\tau_{\DR}$ achieves the oracle error rate if 
\begin{align}
        &   \|F-F_0\|_{\sup,4}\cdot\left\{ \|\pi - \pi_0\|_{4}+ \|G-G_0\|_{\sup,4} + \|S_D - S_{D0}\|_{\sup,4} \right\}  \nonumber \\
    &\quad + \|K_1(g,g_0)\|_2 + \|K_2(g,g_0)\|_2 + \|K_3(g,g_0)\|_2 = o_p(n^{-1/2}).  
    \label{eq:k_DR}
\end{align}
\end{theorem}
Compared to the ltrcR-learner, the ltrcDR-learner $\hat\tau_{\DR}$ does not require $\|\hat\pi - \pi_0\|_{4} = o_p(n^{-1/4})$ to achieve the oracle error rate. 
Despite the advantageous error rate of the ltrcDR-learner compared to the ltrcR-learner, however, our simulations in Section \ref{sec:simu_CATE} show that the ltrcR-learner tends to have smaller mean squared error in finite samples, especially with small sample sizes. This observation is consistent with \citet{morzywolek2023general}. 

\begin{remark}
Recall that the ATE is a special case of CATE with $V^* = \varnothing$, and $\partial \tilde\ell_{\DR}/\partial \theta$ gives $U$ in \eqref{eq:AIPW_3bias}.
The product rate condition required in Theorem \ref{thm:rdr}  is slightly weaker than that in Theorem \ref{thm:CATE_ltrcDR}, as the norms in Assumption \ref{ass:prodrate} are upper bounded by their counterparts in \eqref{eq:k_DR}. Nonetheless  both theorems imply that the nuisance estimation has only second order impact on the estimation of $\theta=\tau(\varnothing)$. 
\end{remark}

\section{Simulation} \label{sec:simu}

\subsection{ATE} \label{sec:simu_ATE}

The details of data generation are given in the Appendix. 
Our estimand 
$\theta=\PP\{T^*(1)>3\} - \PP\{T^*(0)>3\} = -0.1163$, computed from a simulated full data sample of size  $n=10^7$.
We then consider three models for estimating $F,G$ and $S_D$: 
1) true model `Cox1': $T, Q \sim A+ Z_1+Z_2$, $D\sim A+ Z_1+Z_2+Q$; 
2) false model `Cox2': $T, Q \sim AZ_1+Z^2_2$, $D\sim AZ_1+Z^2_2+Q$; 
3) ML model `pCox': regularized Cox regression 
with $L_2$ penalty and 
 natural spline basis functions, see Appendix for more details.
We also consider three models for estimating $\pi$:
1) true logistic model `lgs1':  $A\sim Z_1+Z_2$;
2) false logistic model `lgs2': $A\sim Z_1^2 + Z_2^2$;
3) ML model `gbm': generalized boosted regression models, see also the Appendix for more details.

We consider the estimator $\hat\theta$ with the above (semi)parametric models to estimate the nuisance parameters, 
as well as the estimator $\hat\theta_{cf}$ with 5-fold cross-fitting in Table \ref{tab:simu_n1000_dr_cf_simu2}. 
We use scheme (a) in Figure \ref{fig:DGM_nuisance_summary} for estimation of the nuisance parameters. It follows that when the wrong ``Cox2'' or ``lgs2'' (marked in red in Table \ref{tab:simu_n1000_dr_cf_simu2}) is used to estimate the weights for the downstream nuisance parameters, the estimates can be inconsistent even though their models are correctly specified; these models are marked in {orange}.
We also consider the inverse probability weighted (IPW) estimator that accounts for all three sources of bias: confounding, truncation and censoring.
The estimated probabilities involved in the denominators of $\hat\theta$, $\hat\theta_{cf}$, and the IPW estimator are {truncated at} 0.1 to improve the stability of the estimators.
Finally as a benchmark, we provide the  full data estimator which is the average of the transformed potential event times in the full data sample. 

In Table \ref{tab:simu_n1000_dr_cf_simu2} 
we report the bias, empirical standard deviation (SD),  mean of the model-based standard errors (SE) and the bootstrapped standard errors (boot SE), and coverage probability (CP) of the 95\% confidence intervals (CI)
from 500 simulation runs each with sample size 1000. 
Bootstraps are performed by resampling with replacement 200 times \citep
{efron1994introduction}.
The model-based SE for the IPW estimator is computed from the robust sandwich variance estimator assuming that the weights are known. 
Additional simulation results with sample size 500 are in the Appendix.

From Table \ref{tab:simu_n1000_dr_cf_simu2}, we see that the full data estimator has  small bias, SD and close to nominal coverage as expected.
The estimator $\hat\theta$ has small bias and close to nominal coverage when at least one set of models for $F$ or $(\pi,G,S_D)$ is correctly specified; that is, the first five rows of the table. It has larger bias otherwise, especially row 7. 
The estimator $\hat\theta_{cf}$ also has good coverage, though larger variance in comparison. 
The IPW estimator, on the other hand, can have large bias and poor coverage when there is misspecification of the models, and misspecification of the $G$ or $S_D$ distribution appears to have a more severe impact than misspecification of the $\pi$ distribution here. 
 \begin{table}[h]
	\centering
 \renewcommand{\arraystretch}{0.6}
	\caption{Simulation results for estimates of $\theta=\PP\{T^*(1)>3\} - \PP\{T^*(0)>3\} = -0.1163$  with sample size 1000. The models in {red} are misspecified, while the models in {orange} are correctly specified  with weights estimated from the misspecified models.
    }
	\label{tab:simu_n1000_dr_cf_simu2}
\begin{tabular}{llrrrr}
  \toprule
 Methods & $F$/$\pi$-$G$-$S_D$ & bias & SD & SE/bootSE & CP/bootCP \\ 
\midrule
   $\hat\theta$ & Cox1/lgs1-Cox1-Cox1 & 0.0008 & 0.0481 & 0.0494/0.0500 & 0.962/0.962 \\ 
   & \red{Cox2}/lgs1-Cox1-Cox1 & 0.0003 & 0.0476 & 0.0500/0.0493 & 0.966/0.956 \\
   & Cox1/\org{lgs1}-\red{Cox2}-Cox1 & 0.0020 & 0.0450 & 0.0450/0.0467 & 0.948/0.950 \\ 
   & Cox1/\red{lgs2}-Cox1-Cox1 & 0.0006 & 0.0456 & 0.0456/0.0472 & 0.962/0.968 \\
   & Cox1/\org{lgs1}-\org{Cox1}-\red{Cox2} & 0.0032 & 0.0493 & 0.0516/0.0514 & 0.968/0.960 \\ 
   & \red{Cox2}/\org{lgs1}-\red{Cox2}-Cox1 & 0.0043 & 0.0450 & 0.0474/0.0469 & 0.964/0.950 \\ 
   & \red{Cox2}/\red{lgs2}-Cox1-Cox1 & 0.0190 & 0.0452 & 0.0460/0.0463 & 0.940/0.936 \\ 
   & \red{Cox2}/\org{lgs1}-\org{Cox1}-\red{Cox2}  & -0.0021 & 0.0490 & 0.0522/0.0508 & 0.968/0.956 \\ 
   \midrule 
   $\hat\theta_{cf}$ & pCox/gbm-pCox-pCox & -0.0285 & 0.0764 & 0.0777/0.0886 & 0.938/0.967 \\
   \midrule
   IPW & - /lgs1-Cox1-Cox1 & -0.0015 & 0.0520 & 0.0511/0.0509 & 0.946/0.936 \\ 
   & - /\org{lgs1}-\red{Cox2}-Cox1 & 0.0751 & 0.0473 & 0.0492/0.0473 & 0.662/0.650 \\ 
   & - /\red{lgs2}-Cox1-Cox1  & 0.0068 & 0.0502 & 0.0471/0.0487 & 0.932/0.936 \\ 
   & - /\org{lgs1}-\org{Cox1}-\red{Cox2} & -0.0376 & 0.0514 & 0.0514/0.0513 & 0.876/0.868 \\ 
   & - /gbm-pCox-pCox & 0.0040 & 0.0631 & 0.0603/0.0635 & 0.950/0.959 \\ 
   \midrule
   full data &  & -0.0009 & 0.0185 & 0.0184/0.0184 & 0.938/0.944 \\ 
\bottomrule
 \end{tabular}
 \end{table}

\subsection{CATE} \label{sec:simu_CATE}

The details of data generation are again in the Appendix. 
We consider $\nu(t) = \log(t)$ and three functional forms of $\tau(z) $:  
(i)  $  0.2 -0.2z_1$; 
(ii)  $ 0.2 -0.2\{(z_1+z_2)/2\}^2$;  
(iii)  $ 0.2 - 0.2\sin(\pi z_1) + 0.2 \sqrt{|z_2|}$.
We consider the ltrcR- and ltrcDR-learners. 
Since sample splitting does not make efficient use of the data, 
we consider the   5-fold cross-fitted empirical risk minimization \citep{nie2021quasi}. 

The nuisance parameters $(\pi,F,G,S_D)$ are estimated using 
the same ML models `gbm' and `pCox' as in Section \ref{sec:simu_ATE}. The final second-stage estimator  is obtained using extreme gradient boosting implemented in the R package 
``\texttt{xgboost}'' with customized loss function,  with details given in Appendix \ref{app:tuning_selection}.
As a comparison we also consider the IPW.S-learner, which is the S-learner considered in \citet[Supplementary material]{foster2023orthogonal}  with inverse probability weights to handle left truncation and right censoring; 
its loss function is also given in the Appendix. 
Finally we compare with the oracle versions of the above learners, with the true nuisance parameters plugged in. 

We compute the mean squared error (MSE) $n^{-1}\sum_{i=1}^n \{\hat\tau(Z_i) - \tau_0(Z_i)\}^2$ of the CATE estimates  with sample sizes 500, 1000, 2000 and 3000 using 500 simulation runs each. 
The results are shown as boxplots in Figure   
\ref{fig:CATE_MSE_HTE4}.
As illustration we visualize using 3D-plots  the true CATE surface and the estimated CATE surfaces in Appendix \ref{app:CATE_simu_3Dplots} from one simulated data set with sample size 2000.
\begin{figure}[h]
        \centering
        \begin{subfigure}{0.49\textwidth}
        \includegraphics[width=1\textwidth]{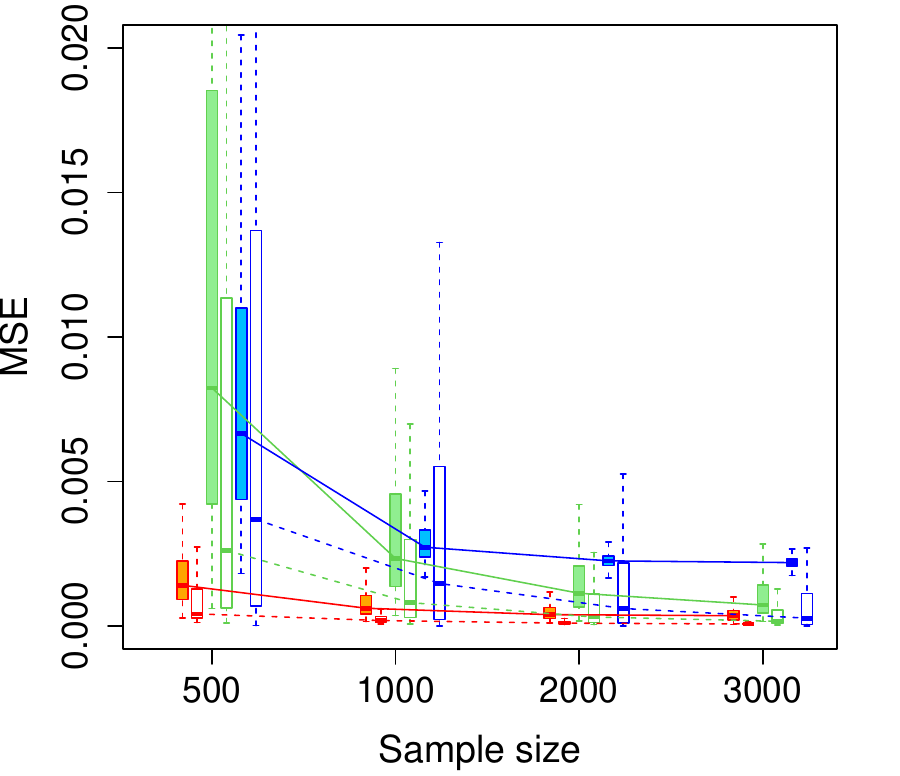}
        \caption*{Scenario (i)}
        \end{subfigure}
\hfill
\begin{subfigure}{0.49\textwidth}
         \includegraphics[width=1\textwidth]{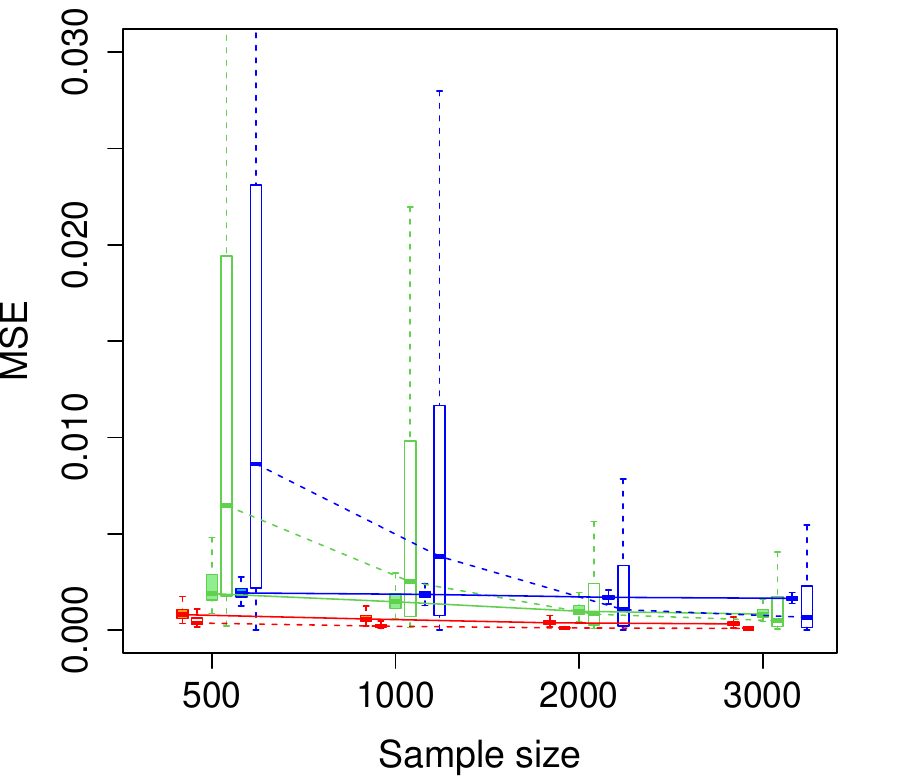}
         \caption*{Scenario (ii)}
         \end{subfigure}
         
         \vspace{0.7em}

         \begin{subfigure}{0.49\textwidth}
         \includegraphics[width=1\textwidth]{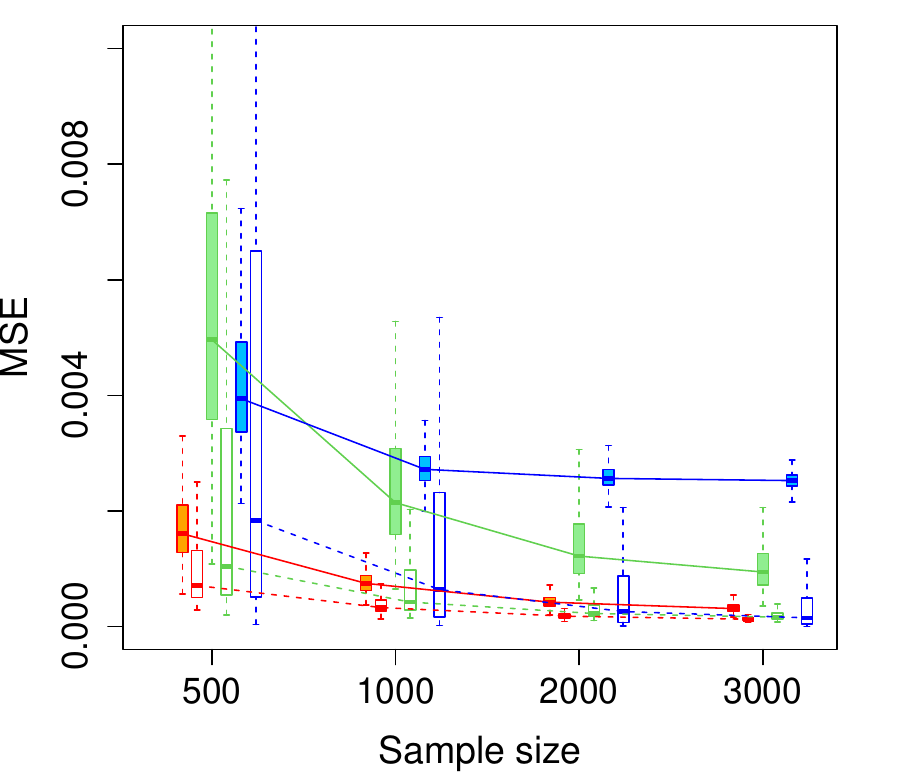}
         \caption*{Scenario (iii)}
         \end{subfigure}
         \hspace{0.7cm}
	\begin{subfigure}{0.32\textwidth}
         \includegraphics[width=0.5\textwidth]{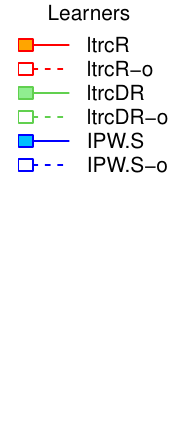}
         \end{subfigure}
        \caption{MSE for different learners under scenarios (i) - (iii) of simulation;  ``-o'' indicates the oracle learner with the true nuisance parameters.
        }\label{fig:CATE_MSE_HTE4}
\end{figure}

We see that in all three scenarios, the ltrcR-learner achieves the smallest median MSE for all sample sizes among the learners that use estimated nuisance parameters. 
The MSE's of the ltrcR- and ltrcDR-learners get closer to the oracle MSE's as the sample size increases, while the MSE's from the IPW.S-learner converges to its oracle counterpart much slower. 
Also both the ltrcR- and ltrcDR-learners have smaller median MSE than the IPW.S-learner when the sample size is at least 1000. 
We notice that scenario (ii) may be easy for the S-learner since $\mu(1,Z;F_0)$ and $\mu(0,Z;F_0)$ only involve linear and quadractic terms and are easier to estimate compared with the other two scenarios. We observe that in Scenario (ii) with sample sizes 500 and 1000, the oracle ltrcDR- and IPW.S-learners show larger median MSE's than their counterparts with estimated nuisance parameters; this phenomenon may be related to the knowledge in the literature that for  ATE, IPW estimators with the estimated propensity score can be more efficient than the one with the true propensity score \citep{robins1992estimating}.

Our simulation results aligned with the findings in \citet{morzywolek2023general} for the case without LTRC, that the R-learner usually has smaller MSE than the DR-learner in finite sample. 
Appendix \ref{app:simu_CATE_additional} 
shows the MSE's of the same learners under the same three scenarios as in Figure \ref{fig:CATE_MSE_HTE4}, except with the probabilities involved in the denominators of the loss function expressions truncated at 0.05 instead of  0.1; we see that the MSE of the ltrcDR-learner is much larger than that in Figure \ref{fig:CATE_MSE_HTE4}.

\section{Application}\label{sec:application}

We analyze the data collected between 1965 and 2012 from the Honolulu Heart Program (HHP, 1965-1990) and the subsequent Honolulu Asia Aging Study (HAAS, 1991-2012), which followed a cohort of men of Japanese ancestry born between 1900-1919 and living on the island of Oahu, Hawaii \citep{p2012honolulu, zhang2024marginal, rava2023doubly}.
In the {HHP}-HAAS study, mid-life alcohol exposure was assessed by self-reported consumption and was dichotomized into heavy versus non-heavy drinking during the earlier HHP period, while late life  cognitive impairment was assessed by the Cognitive Assessment and Screening Instrument (CASI) measured during the  HAAS period starting in 1992.

In this application, we focus on the effect of mid-life heavy drinking on late life cognitive impairment free survival on the age scale, 
 which may also be called disease free survival (DFS). 
It is clear that age at DFS is left truncated by age at HAAS study entry.
Following the prior analyses  \citep{rava2023doubly,  zhang2024marginal, luo2023doubly} we consider the following four baseline covariates: education ($\leq$ 12 years or otherwise), {\it APOE} genotype (presence or absence of an {\it APOE E4} risk allele),  systolic blood pressure {(mmHg)} and heart rate {(per minute)} assessed at the start of HHP. The covariates distributions are shown in Appendix \ref{app:HAAS_table1}. 
After removing subjects with missing covariates, the data consist of 1953 subjects who were alive and did not have cognitive impairment when they entered HAAS.
Among the 1953 subjects, 470 (24.1\%)  were heavy drinkers at some point in mid-life during the HHP period. 
Overall about 18\% of the subjects 
were right censored.  
Finally, subjects entered HAAS between 71.3 and 93.2 years old, 
with the minimum observed event time at  72.9 years old. 
We therefore consider the survival probabilities conditional on surviving to {71.3} years \citep{tsai1987note, wang1989semiparametric, wang1991nonparametric}. 

We assessed the potential violation of quasi-independence between age at HAAS study entry and age at DFS, as well as dependence between age at HAAS study entry 
and the covariates, and between the residual censoring time and the covariates. The results are given in the Appendix showing that it is appropriate to consider covariate dependent truncation and censoring for this data set.

\subsection{ATE}\label{sec:application_ATE}

We estimate the difference between the potential DFS probabilities for heavy drinking and non-heavy drinking at ages 80, 85, 90, and 95 years, respectively, i.e.~with $A=1$ representing heavy drinking and $\nu(t) = \ind(t>t_0)$ for $t_0 = 80, 85, 90, 95$, respectively.

We consider the same estimators as in the simulation section, with $\hat\theta_{cf}$ 
denoted by ``cf'' in Figure \ref{fig:HAAS_ATE_est_CIcfmodelSE_alpha0_lambdamin_nuisa}, and $\hat\theta$ denoted by `Param' with the following {(semi)parametric} nuisance models.   
For $F$ and $G$ (on the reversed time scale): a Cox model with alcohol exposure and all the covariates included as  regressors; 
for $S_D$ on the residual time scale, i.e.~since HAAS study entry: a Cox model with the left truncation time (age at HAAS study entry), alcohol exposure, and the covariates included as  regressors; for $\pi$: a logistic model with all the covariates included as regressors.  

\begin{figure}[h]
\centering
\includegraphics[width=0.7\textwidth]{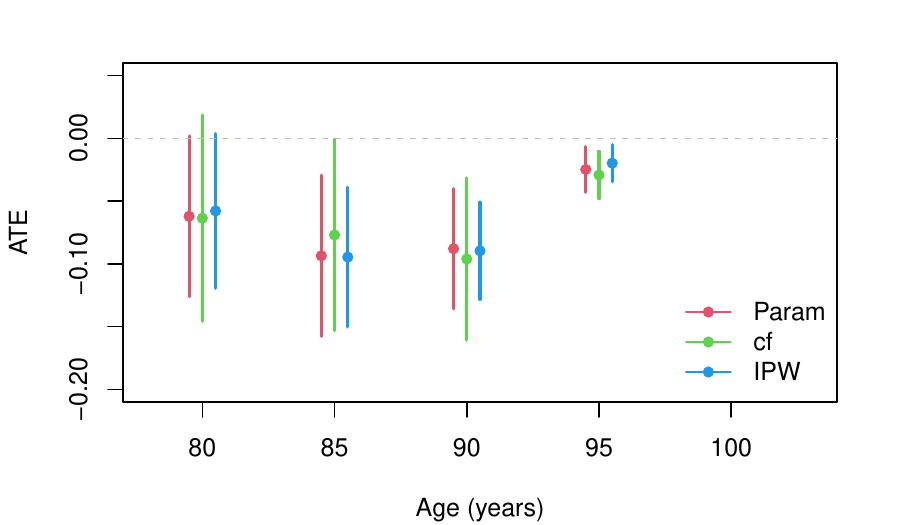}
	\caption{Estimates of ATE as difference of the potential DFS probabilities at various ages for heavy versus non-heavy drinkers in the HAAS data, with 95\% confidence intervals.
    } 
  \label{fig:HAAS_ATE_est_CIcfmodelSE_alpha0_lambdamin_nuisa}
\end{figure}

Figure \ref{fig:HAAS_ATE_est_CIcfmodelSE_alpha0_lambdamin_nuisa} shows the different ATE estimates and their 95\% confidence intervals. The confidence intervals for the ``cf'' estimator is based on model-based standard errors, which has been shown to be more accurate in simulation. 
The confidence intervals for other methods are based on bootstrapped standard errors from 200 replications. 
We see that all the point estimates are negative, indicating that mid-life heavy drinking can be a  cause for lower DFS probabilities at ages 80, 85, 90 and 95 years.

\subsection{CATE}

As an illustration we estimate the difference in the conditional DFS probabilities of heavy  versus non-heavy drinking at age 90 years, where  the estimated ATE appeared the most significant. 
We consider  all four baseline covariates as effect modifiers, 
and the ltrcR-, ltrcDR- and IPW.S-learners with the same way as  in Section \ref{sec:simu_CATE} to estimate the nuisance parameters and to solve the second stage  optimization problem. 

Figure \ref{fig:HAAS_CATE_hist} shows the histograms of the CATE estimates from the three learners. 
We see that the CATE estimates from the IPW.S-learner are much more concentrated 
than the other two, showing little variability of the treatment effects. 
We note that the ATE is the expectation of CATE under the full data covariate distribution, and is in general different from the expectation of CATE under the observed covariate distribution
under left truncation.  
As a result,  the estimated ATE may not be at the center of the histogram of the CATE estimates. 

\begin{figure}[h]
\centering
\includegraphics[width=1\textwidth]{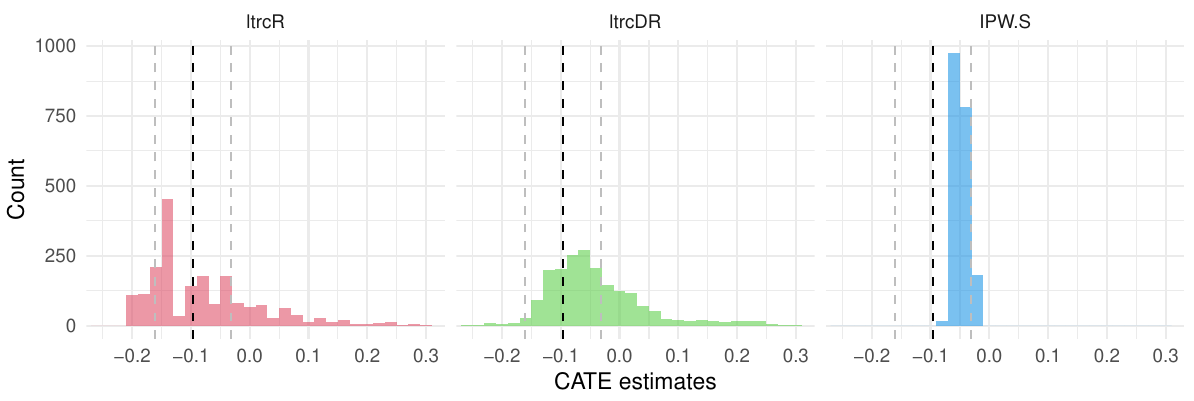}
	\caption{Histogram of the CATE estimate for DFS at age 90 using the HAAS data; the black and grey vertical dashed lines denote the ``cf'' estimate of the ATE and its 95\% confidence interval.} 
  \label{fig:HAAS_CATE_hist}
\end{figure}

Figure \ref{fig:HAAS_CATE_plot3D} shows the estimated CATE surfaces from the ltrcR-learner for each of the four combined categories of education and {\it APOE} genotype, with the x-axis and y-axis indicating different values of systolic blood pressure (SBP) and heart rate (HR). Table \ref{tab:A1} in 
Appendix \ref{app:HAAS_table1} shows the distribution of heavy versus non-heavy drinkers within each of the four combined education and {\it APOE}  genotype groups. 
From Figure \ref{fig:HAAS_CATE_plot3D} we see that for almost all values of systolic blood pressure and heart rate, the estimated CATE tends to be negative {and} the lowest for the green surface, 
which corresponds to education $>12$ years and absence of the {\it APOE E4} risk allele. We understand that education is a proxy for  socioeconomic status, so the green group can be seen as the lowest risk group for cognitive impairment or death, while the blue group, 
corresponding to 
education $\leq 12$ years and {\it APOE}  risk allele present, can be seen as at the highest risk. Apparently the harmful effect of heavy drinking during mid-life in this data set, is estimated to be the most substantial for the lowest risk green group. 
We  note  that the estimand for this data is conditional upon surviving beyond 71.3 years of age, 
so we speculate that the subjects in the highest risk blue group may be somewhat resilient against certain perceived harms during  their mid-lives (such as heavy drinking). 
A limitation of our current CATE estimate is the lack of confidence bounds, so that we cannot tell whether the effects are statistically significant. 

\begin{figure}[h]
\centering
\includegraphics[width=0.49\textwidth]{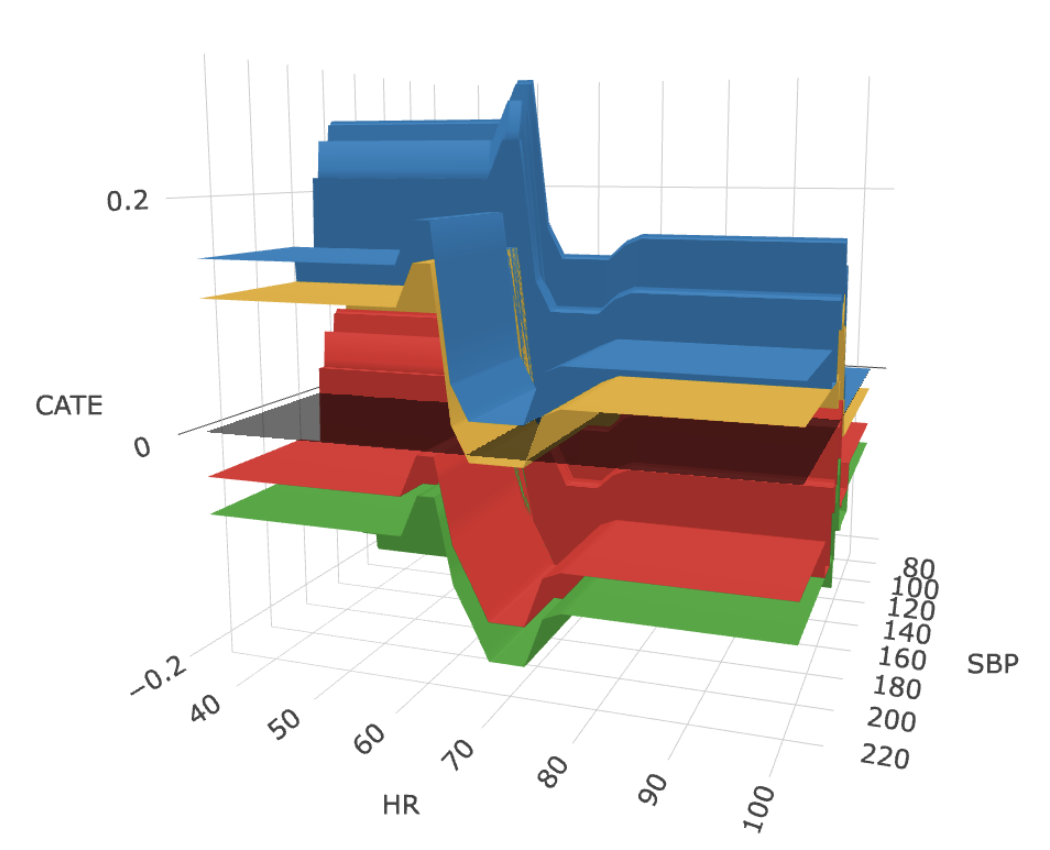}
\includegraphics[width=0.49\textwidth]{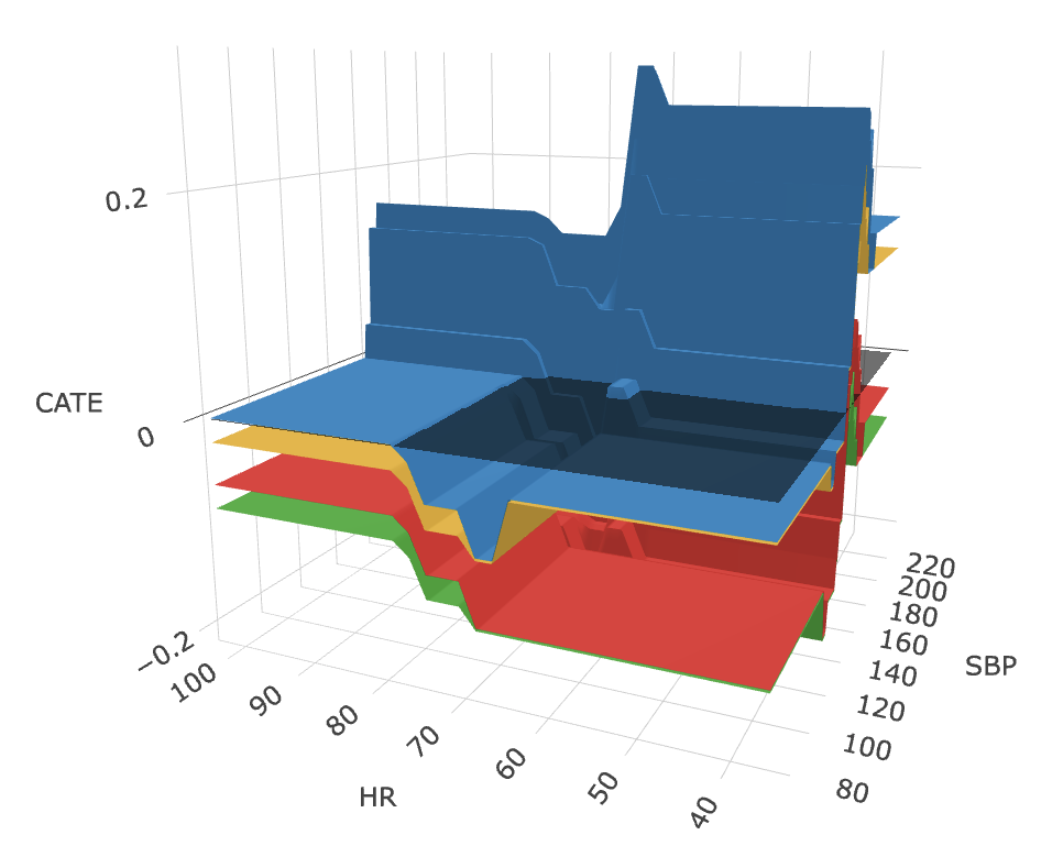}
\includegraphics[width=0.49\textwidth]{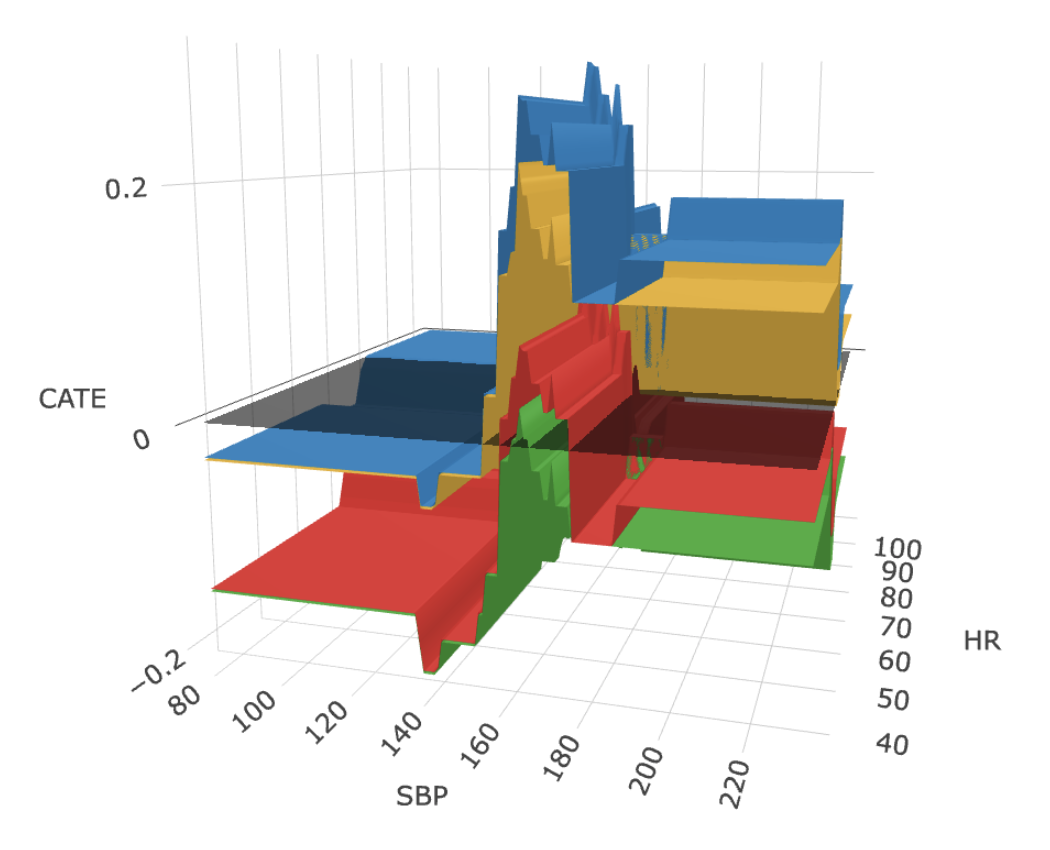}
\includegraphics[width=0.49\textwidth]{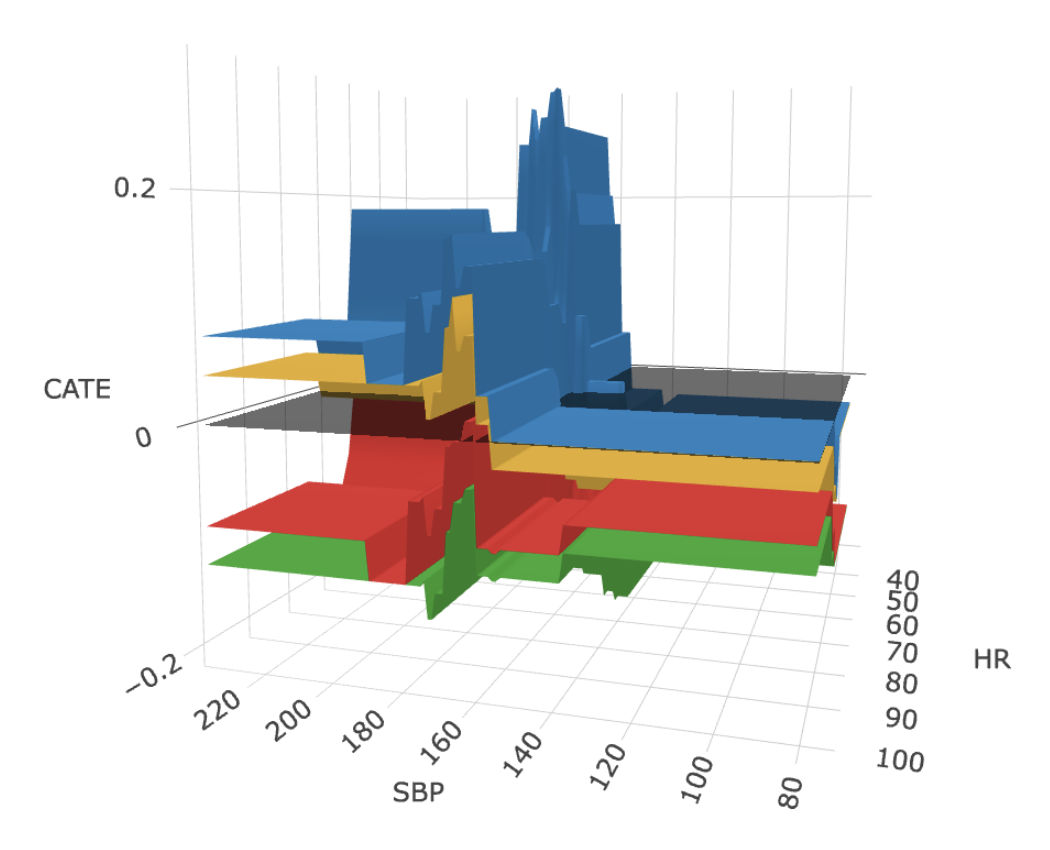}
\par 
\vspace{0.2em}
\hfill
\includegraphics[width=0.15\textwidth]{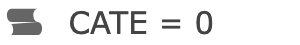}
\includegraphics[width=0.35\textwidth]{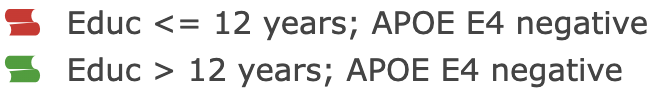}
\includegraphics[width=0.35\textwidth]{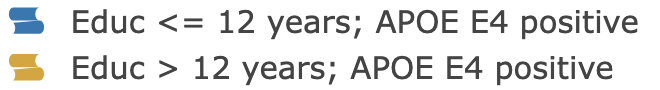}
\caption{Estimated CATE surfaces of DFS at age 90 from the ltrcR-learner for different education and {\it APOE}  genotype subgroups in the HAAS data (views from four different angles).
}
\label{fig:HAAS_CATE_plot3D}
\end{figure}

Other interesting features of Figure \ref{fig:HAAS_CATE_plot3D} are discussed in the Appendix, such as the dip for heart rate  between 64 - 76, or the spike for SBP between 158 - 180. These turn out to be roughly consistent with the DSF differences  in the corresponding subgroup product-limit plots.
Appendix \ref{app:HAAS_CATE_surv90} contains also a version of less smooth CATE surfaces from the ltrcR-learner, as well as a version of estimated CATE surfaces from the ltrcDR-learner.


\section{Discussion} \label{sec:discussion}

We have proposed a general orthogonal and  doubly robust framework for handling covariate dependent LTRC, which preserves the orthogonality and  double robustness of the original estimating or loss function without LTRC.  
We have applied it to the estimation of average treatment effect and conditional average treatment effect. The HAAS data example showed how such ATE and CATE estimates can be useful in practice; 
it is particularly encouraging that features discovered by the estimated CATE are confirmed by the direct subgroup PL plots which condition on the covariates. 
The  framework can be applied to a broad range of problems based on specific estimating functions or loss functions, including 
the survival integral considered in \citet{morenz2024debiased} 
and the weighted orthogonal learners developed in \citet{morzywolek2023general}. 

As discussed in \citet{qian2014assumptions}, in the presence of left truncation, there are two censoring scenarios: (1) censoring may occur before left truncation, and (2) censoring always after left truncation. 
In this manuscript, we focused on the second scenario which is the case for our application.  The proposed approaches can also be easily adapted to the first censoring scenario by addressing censoring first on the original time scale and then addressing left truncation.

For the operator $\V_Q$ in Section \ref{sec:AIPW_operators}, it can be shown that for any estimand that is identified by $\zeta$ nonparametrically in the LTRC-free data, \eqref{eq:AIPW_Q} is the efficient influence function (EIF) {in the truncated but censoring free data} up to a constant factor. 
For the ATE, after applying $\V_Q$ to $U_0$ in \eqref{eq:U0} we have an 
EIF of $\theta$ in the truncated but censoring-free data, up to the constant factor $\beta$. 
Following  \citet{tsiatis2006semiparametric}, 
$U$ in \eqref{eq:AIPW_3bias} is an influence function of $\theta$ (up to the factor $\beta$), and is efficient within the class 
given by equation (10.1) of \citet{tsiatis2006semiparametric}. 
However, this does not imply that $U$ (even when  scaled by  $\beta$) is the EIF in the observed data. This is because the influence function for $\theta$ in the truncated but censoring free data is not unique, 
 due to the conditional independent truncation restriction imposed by Assumption \ref{ass:trunc}. 
\citet[Lemma S2 in the Supplmenetary Material]{wang2024doubly} showed that
the tangent space for the semiparametric model implied by the conditional (quasi-)independent left truncation assumption is a proper subspace of the entire $L^2$ space of the truncated but censoring-free data; a similar result on  unsaturated tangent space was also  stated in \citet{bickel1993efficient} under the random left truncation assumption.

There are currently limited 
nonparametric or ML software that can handle data under LTRC; these include
R packages 
\texttt{glmnet}, 
\texttt{hal9001}, 
\texttt{LTRCforests},
and 
\texttt{survPen}. In addition, only the first two of these allow weights, which are needed in the first-stage  estimation of the nuisance functions. 
For the second stage both ltrcR- and ltrcDR-learners  may have negative weights resulting from augmentation, and   
many existing regression software 
including \texttt{glmnet}, 
\texttt{hal9001}, 
\texttt{KRLS}, \texttt{randomForest}, and 
\texttt{xgboost} with default loss functions,
do not support negative weights  (although the latest \texttt{glmnet} version 4.1-8 no longer return an error message). 
We used extreme gradient boosting implemented in the R package ``\texttt{xgboost}'' with customized loss function.
We expect that over time more software will be available to handle LTRC and to incorporate weights;
 another  useful future development would be efficient plug-in learners \citep{van2024combining} to mitigate challenges posed by negative weights. 

For the CATE estimation, following \cite{foster2023orthogonal} we provided in the Appendix the error bound and convergence rate results that apply to general second stage learning algorithms based on empirical risk minimization, which includes many widely used machine learning algorithms such as sparse linear regression, kernel regression, and neural networks. It does not include tree or forests based methods. 
\citet{cui2023estimating} extended causal forests for CATE to right censored data, 
by employing
orthogonal estimating functions,
and showed that the nuisance estimation errors only have second order impacts on the CATE estimator. 
Their approach assumes smoothness and estimates the hazard function of the censoring distribution, which can be of slower rate 
than the cumulative quantities that we use here. On the other hand, the 
 LTRC  operator we have developed here can be applied to both their splitting criteria as well as the estimating functions for each leaf, so that left truncation can be properly handled.

One limitation of this work is the lack of uncertainty quantification for the CATE estimates. 
With additional restrictions on the learning algorithms such as smoothness, asymptotic distributions of the CATE estimators may be characterized. Without LTRC, inference for CATE was {mentioned} in \citet{kennedy2023towards} for learners with oracle property, assuming that the inference results of the second-stage learning algorithm are known.   
In addition, asymptotic normality has been shown for CATE estimators from recursive partitioning and random forests \citep{athey2016recursive, wager2018estimation, cui2023estimating}.
Investigating the asymptotic distribution and inference results for CATE estimators under LTRC would be an important direction for future work.

The code for implementation of the proposed estimators and simulations in this paper is available at \href{https://github.com/wangyuyao98/truncAC}{https://github.com/wangyuyao98/truncAC}. 













\bibliographystyle{plainnat}
\bibliography{bib/left_trunc, bib/cox, bib/survivalHTE, bib/osg}

\clearpage

\begin{center}
    {\LARGE Appendix for ``A Liberating Framework from Truncation and Censoring, with Application to Learning Treatment Effect"}
\end{center}

\vspace{5em}

{\bf \Large Contents}
{\large
\begin{itemize}
    \item[] \hyperref[sec:crossprod_norm_properties]{\ref{sec:crossprod_norm_properties}~ Preliminaries} \hfill \pageref{sec:crossprod_norm_properties}

    \item[] \hyperref[app:DR_proof]{\ref{app:DR_proof}~ Proofs of  double robustness and Neyman orthogonality} \hfill \pageref{app:DR_proof}

    \item[] \hyperref[supp:nuis]{\ref{supp:nuis}~ Nuisance estimation}\hfill \pageref{supp:nuis}

    \item[] \hyperref[app:ATE_theory]{\ref{app:ATE_theory}~ ATE theory}\hfill \pageref{app:ATE_theory}

    \item[] \hyperref[app:CATE_theory]{\ref{app:CATE_theory}~ CATE theory}\hfill \pageref{app:CATE_theory}

    \item[] \hyperref[app:simulation]{\ref{app:simulation}~ Additional simulation details and results}\hfill \pageref{app:simulation}

    \item[]  \hyperref[app:HAAS_additional]{\ref{app:HAAS_additional}~ Additional HAAS data analysis results}\hfill \pageref{app:HAAS_additional}
\end{itemize}
}

\appendix

\begingroup
\setcounter{tocdepth}{1}
\endgroup

\renewcommand{\theequation}{A\arabic{equation}}  
\setcounter{equation}{0}                         


\setcounter{theorem}{0}
\setcounter{proposition}{0}
\renewcommand{\thetheorem}{A\arabic{theorem}}
\renewcommand{\theproposition}{A\arabic{proposition}}
\renewcommand{\thecorollary}{A\arabic{corollary}}
\setcounter{corollary}{0}

\renewcommand{\thefigure}{A\arabic{figure}}
\renewcommand{\thetable}{A\arabic{table}}
\setcounter{figure}{0}
\setcounter{table}{0}

\clearpage
\section{Preliminaries} \label{sec:crossprod_norm_properties}

We will use $\lesssim$ to denote less than or equal to after multiplying a {absolute} constant {(i.e.~not dependent on $n$)}.

\subsection{Expressions of $\V(\nu;F,G,S_D)$ and $\V(1;F,G,S_D)$
} \label{app:expression_V_nu} 

The expressions of the above two are in \eqref{eq:V(nu)} and \eqref{eq:V(1)}; {they are expressions when the LTRC operator $\V$ is applied to the function $\nu$ and constant function 1, respectively}. We derive them in the following.
For any real numbers $a$ and $b$, let $a\wedge b$ denote the minimum of $a$ and $b$. 
\begin{align}
    &\quad \V(\nu;F,G,S_D)  \nonumber \\
    & =  \frac{\Delta}{S_D(X-Q|Q,A,Z)} \left\{\frac{\nu(X)}{G(X|A, Z)} - \int_0^\infty \frac{\int_0^v \nu(t) dF(t|A,Z) }{1-F(v|A, Z)} \cdot \frac{d\bar M_Q(v;G)}{G(v|A, Z)} \right\} \nonumber \\
    &\quad + \int_0^\infty \left[\E\left\{\left. \frac{\nu(T)}{G(T|A,Z)} \right| T-Q\geq u,Q,A,Z\right\} \right.   \nonumber \\
    &\quad\quad\quad\quad\   \left.  -  \int_0^\infty \frac{\int_0^v \nu(t) dF(t|A,Z) }{1-F(v|A, Z)}  \cdot \frac{\E\{d\bar M_Q(v;G)|T-Q\geq u, Q,A,Z \}}{G(v|A,Z)} \right] \frac{dM_D(u;S_D)}{S_D(u|Q,A,Z)}. \label{eq:V(nu)_compute1}
\end{align}
In the following, we first compute the two conditional expectations involved in \eqref{eq:V(nu)_compute1}: 

\begin{align}
    &\quad \E\left\{\left. \frac{\nu(T)}{G(T|A,Z)} \right| T-Q\geq u,Q = q,A=a,Z=z\right\} \nonumber \\
    & = \E\left\{\left. \frac{\nu(T^*)}{G(T^*|A^*,Z^*)} \right| T^*\geq Q^*+u,Q^* = q,A^*=a,Z^*=z\right\} \nonumber \\
    & = \frac{\int_{Q+u}^\infty \{\nu(t)/G(t|a,z)\} dF(t|a,z)}{1-F(q+u|a,z)},  \nonumber
\end{align}
and
\begin{align}
    & \quad \E\{d\bar M_Q(v;G)|T-Q\geq u, Q,A,Z \} \nonumber \\
    & = d\bar N_Q(v) + \E\{\ind(Q\leq v<T)|T-Q\geq u, Q,A,Z \} \frac{dG(v|A,Z)}{G(v|A,Z)} \nonumber \\
    & = d\bar N_Q(v) + \E\{\ind(Q\leq v<T)|T\geq Q+u, Q,A,Z \} \frac{dG(v|A,Z)}{G(v|A,Z)} \nonumber \\
    & = d \bar N_Q(v) +  \left\{\frac{1-F(v|A,Z)}{1-F(Q+u|A,Z)}\ind(Q+u\leq v)  + \ind(Q+u>v) \right\} \ind(Q\leq v) \frac{dG(v|A,Z)}{G(v|A,Z)}.  \nonumber
\end{align}
Therefore, 
\begin{align}
    \V(\nu;F,G,S_D)  = \A_1 + \A_2 + \A_3 - \A_4 - \A_5, \nonumber
\end{align}
where 
\begin{align}
    \A_1 & = \frac{\Delta}{S_D(X-Q|Q,A,Z)} \left\{\frac{\nu(X)}{G(X|A, Z)} - \int_0^\infty \frac{\int_0^v \nu(t) dF(t|A,Z) }{1-F(v|A, Z)} \cdot \frac{d\bar M_Q(v;G)}{G(v|A, Z)} \right\},  \nonumber \\
    \A_2 & =  \int_0^\infty \frac{\int_{Q+u}^\infty \{\nu(t)/G(t|A,Z)\} dF(t|A,Z) }{1-F(Q+u|A,Z)} \cdot \frac{dM_D(u;S_D)}{S_D(u|Q,A,Z)}, \nonumber  \\
    \A_3 &  = \int_0^\infty  \frac{\int_0^Q \nu(t) dF(t|A,Z) }{\{1-F(Q|A, Z)\}G(Q|A,Z) }  \cdot \frac{dM_D(u;S_D)}{S_D(u|Q,A,Z)} \nonumber \\
    & =  \frac{\int_0^Q \nu(t) dF(t|A,Z) }{\{1-F(Q|A, Z)\}G(Q|A,Z)}  \int_0^\infty \frac{dM_D(u;S_D)}{S_D(u|Q,A,Z)}, \nonumber \\
    \A_4 & = \int_0^\infty \int_{Q}^\infty \frac{\ind(Q+u\leq v)  \int_0^v \nu(t) dF(t|A,Z) }{\{1-F(v|A, Z)\}G(v|A,Z)} \cdot \frac{1-F(v|A,Z)}{1-F(Q+u|A,Z)}\cdot \frac{dG(v|A,Z)}{G(v|A,Z)} \cdot \frac{dM_D(u;S_D)}{S_D(u|Q,A,Z)}  \nonumber \\
    & = \int_0^\infty \int_{Q}^\infty \frac{\ind(Q+u\leq v)  \int_0^v \nu(t) dF(t|A,Z) }{1-F(Q+u|A,Z)} \cdot \frac{dG(v|A,Z)}{G(v|A,Z)^2} \cdot \frac{dM_D(u;S_D)}{S_D(u|Q,A,Z)},  \nonumber \\
    \A_5 & = \int_0^\infty \int_Q^\infty \frac{\ind(Q+u>v) \int_0^v \nu(t) dF(t|A,Z) }{1-F(v|A, Z)}  \cdot \frac{dG(v|A,Z)}{G(v|A,Z)^2} \cdot  \frac{dM_D(u;S_D)}{S_D(u|Q,A,Z)}. \nonumber
\end{align}
Since 
\begin{align}
    \frac{\ind(Q+u\leq v)}{1-F(Q+u|A,Z)}  + \frac{\ind(Q+u>v)}{1-F(v|A,Z)}
    & = \frac{1}{1-F\{(Q+u)\wedge v|A,Z)\}}, \nonumber
\end{align}
we have 
\begin{align}
    \A_4 + \A_5 & = \int_0^\infty \left[ \int_Q^\infty \frac{\int_0^v \nu(t) dF(t|A,Z)}{1-F\{(Q+u)\wedge v|A,Z)\}}  \cdot \frac{dG(v|A,Z)}{G(v|A,Z)^2} \right] \frac{dM_D(u;S_D)}{S_D(u|Q,A,Z)}.  \nonumber
\end{align}
Moreover, 
\begin{align}
    \int_0^\infty \frac{dM_D(u;S_D)}{S_D(u|Q,A,Z)} 
    & = \int_0^\infty \left\{ \frac{dN_D(u)}{S_D(u|Q,A,Z)} + \frac{\ind(X-Q\geq u)dS_D(u|Q,A,Z)}{S_D(u|Q,A,Z)^2} \right\} \nonumber \\
    & = \frac{1-\Delta}{S_D(X-Q|Q,A,Z)} - \int_{0}^{X-Q} d\left\{\frac{1}{S_D(u|Q,A,Z)}\right\} \nonumber \\
    & = \frac{1-\Delta}{S_D(X-Q|Q,A,Z)} - \frac{1}{S_D(X-Q|Q,A,Z)} + \frac{1}{S_D(0|Q,A,Z)} \nonumber \\
    & = - \frac{\Delta}{S_D(X-Q|Q,A,Z)} +1. \nonumber
\end{align}
Therefore, 
\begin{align}
    \V(\nu;F,G,S_D)  
    & = \frac{\Delta}{S_D(X-Q|Q,A,Z)} \left\{\frac{\nu(X)}{G(X|A, Z)} - \int_0^\infty \frac{\int_0^v \nu(t) dF(t|A,Z) }{1-F(v|A, Z)} \cdot \frac{d\bar M_Q(v;G)}{G(v|A, Z)} \right\}  \nonumber \\
    &\quad + \int_0^\infty \frac{\int_{Q+u}^\infty \{\nu(t)/G(t|A,Z)\} dF(t|A,Z) }{1-F(Q+u|A,Z)} \cdot \frac{dM_D(u;S_D)}{S_D(u|Q,A,Z)}  \nonumber \\
    &\quad +  \left\{1 - \frac{\Delta}{S_D(X-Q|Q,A,Z)}\right\} \frac{\int_0^Q \nu(t) dF(t|A,Z) }{\{1-F(Q|A, Z)\}G(Q|A,Z)}   \nonumber \\
    &\quad - \int_0^\infty \left[\int_Q^\infty \frac{\int_0^v \nu(t) dF(t|A,Z)}{1-F\{(Q+u)\wedge v|A,Z)\}} \cdot \frac{d G(v|A,Z)}{G(v|A,Z)^2} \right] \frac{dM_D(u;S_D)}{S_D(u|Q,A,Z)}. \label{eq:V(nu)}
\end{align}
As a special case with $\nu\equiv 1$, we have 
\begin{align}
    \V(1;F,G,S_D)  
    & = \frac{\Delta}{S_D(X-Q|Q,A,Z)} \left\{\frac{1}{G(X|A,Z)} - \int_0^\infty \frac{F(v|A,Z) }{1-F(v|A,Z)} \cdot \frac{d\bar M_Q(v;G)}{G(v|A,Z)} \right\}   \nonumber \\
    &\quad + \int_0^\infty \frac{\int_{Q+u}^\infty \{1/G(t|A,Z)\} dF(t|A,Z) }{1-F(Q+u|A,Z)} \cdot \frac{dM_D(u;S_D)}{S_D(u|Q,A,Z)}  \nonumber \\
    &\quad +  \left\{1 - \frac{\Delta}{S_D(X-Q|Q,A,Z)}\right\} \frac{F(Q|A,Z)}{\{1-F(Q|A, Z)\}G(Q|A,Z)}   \nonumber \\
    &\quad - \int_0^\infty \left[\int_Q^\infty \frac{F(v|A,Z)}{1-F\{(Q+u)\wedge v|A,Z)\}} \cdot \frac{d G(v|A,Z)}{G(v|A,Z)^2} \right] \frac{dM_D(u;S_D)}{S_D(u|Q,A,Z)}. \label{eq:V(1)}
\end{align}

\newpage
\bigskip 
\subsection{Norms, integral products and related quantities} \label{app:integral_prod}

Let $P_Z$, $P_{A,Z}$ and $P_{Q,A,Z}$ denote the cumulative distribution function of $Z$, $(A,Z)$ and $(Q,A,Z)$, respectively. 
For $p\in [1,\infty)$, denote 
$\|\hat\pi - \pi\|_{p}^p = \int |\hat\pi(z) - \pi(z)|^p dP_{Z}(z)$,
$\|\hat F - F\|_{\sup,p}^p = \int \sup_{t}|\hat F(t|a,z) - F(t|a,z)|^p dP_{A,Z}(a,z)$,
$\|\hat G - G\|_{\sup,p}^p = \int \sup_{t}|\hat G(t|a,z) - G(t|a,z)|^p dP_{A,Z}(a,z)$,
and
$\|\hat S_D - S_D\|_{\sup,p}^p = \int \sup_{t}|\hat S_D(t|q,a,z) - S_D(t|q,a,z)|^p dP_{Q,A,Z}(q,a,z)$.  
Note that the above norms are random quantities with randomness from $\hat\pi$, $\hat F$, $\hat G$, and $\hat S_D$.
In addition  for $\tilde K_1(g,g_0)$, $\tilde K_2(g,g_0)$ and $\tilde K_2(g,g_0)$ in \eqref{eq:K_1} \eqref{eq:K_2} \eqref{eq:K_3} below, 
$\|\cdot\|_1$ denotes the $L_1$ norm with respect to the distribution of $O$ as in the above while the randomness from $\hat g$ is still kept.

Also let $\|F - F_0\|_{\infty} = \sup_{t,a,z}| F(t|a,z) - F_0(t|a,z)|$, $\|G - G_0\|_{\infty} = \sup_{t,a,z}|G(t|a,z) - G_0(t|a,z)|$, and $\|S_D - S_{D0}\|_{\infty} = \sup_{t,q,a,z}|S_D(t|q,a,z) - S_{D0}(t|q,a,z)|$. 
For $\tilde K_1(g,g_0;1)$, $\tilde K_2(g,g_0;1)$ and $\tilde K_2(g,g_0;1)$ in \eqref{eq:K_1_tilde} \eqref{eq:K_2_tilde} \eqref{eq:K_3_tilde} below, 
similarly let $\|\tilde K_1(g,g_0;1)\|_{\infty}$, $\|\tilde K_2(g,g_0;1)\|_{\infty}$, $\|\tilde K_3(g,g_0;1)\|_{\infty}$ denote the $L_\infty$ norm with respect to the random variables involved.

Denote the following integral product terms with respect to errors $F-F_0$, $G-G_0$, etc.; these are needed when multiple nuisance parameters involve functions of time \citep{wang2024doubly, luo2023doubly}: 
\begin{align}
    H_{1}(g,g_0;\nu) & = \int_Q^X \left\{ \frac{\int_0^v \nu(t) dF(t|A,Z)}{1-F(v|A,Z)} - \frac{\int_0^v \nu(t) dF_0(t|A,Z)}{1-F_0(v|A,Z)} \right\} \cdot d\left\{\frac{1}{G(v|A,Z)} - \frac{1}{G_0(v|A,Z)}\right\}, \label{eq:H1}\\
    H_2(g,g_0;\nu) & = \int_{X}^\infty \left\{\frac{\nu(t)}{G(t|A,Z)} - \frac{\nu(t)}{G_0(t|A,Z)} \right\} \cdot d \left\{F(t|A,Z) - F_0(t|A,Z) \right\}, \nonumber
\end{align}

\begin{align}
    H_{31}(g,g_0;\nu) & =  \int_0^{X-Q}  \int_{Q+u}^\infty \left\{\frac{\nu(t)}{G(t|A,Z)} - \frac{\nu(t)}{G_0(t|A,Z)} \right\}  \nonumber  \\
    &\qquad\qquad\qquad\quad   \cdot \left\{\frac{dF(t|A,Z)}{1-F(Q+u|A,Z)} - \frac{dF_0(t|A,Z)}{1-F_0(Q+u|A,Z)} \right\} \cdot d\left\{\frac{1}{S_{D0}(u|Q,A,Z)} \right\},  \nonumber \\
    H_{32}(g,g_0;\nu) & =  \int_0^{X-Q}  \int_{Q+u}^\infty \frac{\nu(t)}{G_0(t|A,Z)} \cdot \left\{\frac{dF(t|A,Z)}{1-F(Q+u|A,Z)} - \frac{dF_0(t|A,Z)}{1-F_0(Q+u|A,Z)} \right\}  \nonumber  \\
    &\qquad\qquad\qquad\quad  \cdot  d\left\{\frac{1}{S_{D}(u|Q,A,Z)}  - \frac{1}{S_{D0}(u|Q,A,Z)} \right\}, \nonumber \\
    H_{33}(g,g_0;\nu) & =  \int_0^{X-Q}  \int_{Q+u}^\infty  \left\{\frac{\nu(t)}{G(t|A,Z)} - \frac{\nu(t)}{G_0(t|A,Z)} \right\}  \cdot \left\{\frac{dF(t|A,Z)}{1-F(Q+u|A,Z)} - \frac{dF_0(t|A,Z)}{1-F_0(Q+u|A,Z)} \right\} \nonumber  \\
    &\qquad\qquad\qquad\quad  \cdot  d\left\{\frac{1}{S_{D}(u|Q,A,Z)}  - \frac{1}{S_{D0}(u|Q,A,Z)} \right\}. \nonumber
\end{align}

\begin{align}
    H_4(g,g_0;\nu) & = \int_Q^\infty \left\{ \frac{\int_0^v \nu(t) dF(t|A,Z)}{1-F(X\wedge v|A,Z)} - \frac{\int_0^v \nu(t) dF_0(t|A,Z)}{1-F_0(X\wedge v|A,Z)} \right\}  
    \cdot  d\left\{ \frac{1}{G(v|A,Z)} -  \frac{1}{G_0(v|A,Z)} \right\}. \nonumber
\end{align}

\begin{align}
    H_{51}(g,g_0;\nu) & = \int_0^{X-Q} \int_Q^\infty \left[ \frac{\int_0^v \nu(t) dF(t|A,Z)}{1-F\{(Q+u)\wedge v|A,Z)\}} - \frac{\int_0^v \nu(t) dF_0(t|A,Z)}{1-F_0\{(Q+u)\wedge v|A,Z)\}} \right] \nonumber  \\
    &\qquad\qquad\qquad\quad  \cdot d \left\{\frac{1}{G(v|A,Z)} - \frac{1}{G_0(v|A,Z)} \right\} \cdot d\left\{\frac{1}{S_{D0}(u|Q,A,Z)} \right\}, \nonumber \\
    H_{52}(g,g_0;\nu) & = \int_0^{X-Q} \int_Q^\infty \left[ \frac{\int_0^v \nu(t) dF(t|A,Z)}{1-F\{(Q+u)\wedge v|A,Z)\}} - \frac{\int_0^v \nu(t) dF_0(t|A,Z)}{1-F_0\{(Q+u)\wedge v|A,Z)\}} \right] \nonumber  \\
    &\qquad\qquad\qquad\quad \cdot d \left\{ \frac{1}{G_0(v|A,Z)} \right\} \cdot d\left\{\frac{1}{S_D(u|Q,A,Z)} - \frac{1}{S_{D0}(u|Q,A,Z)} \right\}, \nonumber \\
    H_{53}(g,g_0;\nu) & = \int_0^{X-Q} \int_Q^\infty \left[ \frac{\int_0^v \nu(t) dF(t|A,Z)}{1-F\{(Q+u)\wedge v|A,Z)\}} - \frac{\int_0^v \nu(t) dF_0(t|A,Z)}{1-F_0\{(Q+u)\wedge v|A,Z)\}} \right]  \nonumber \\
    &\qquad\qquad\qquad\quad \cdot d \left\{ \frac{1}{G(v|A,Z)} -  \frac{1}{G_0(v|A,Z)} \right\} \cdot d\left\{\frac{1}{S_D(u|Q,A,Z)} - \frac{1}{S_{D0}(u|Q,A,Z)} \right\}. \nonumber
\end{align}
In addition, denote
\begin{align}
    \tilde K_1(g,g_0;\nu) & = |H_1(g,g_0;\nu)| + |H_2(g,g_0;\nu)| + |H_{31}(g,g_0;\nu)| + |H_4(g,g_0;\nu)| + |H_{51}(g,g_0;\nu)|, \label{eq:K_1_tilde} \\
    \tilde K_2(g,g_0;\nu) & = |H_{32}(g,g_0;\nu)| + |H_{52}(g,g_0;\nu)|, \label{eq:K_2_tilde} \\
    \tilde K_3(g,g_0;\nu) & = |H_{33}(g,g_0;\nu)| + |H_{53}(g,g_0;\nu)|, \label{eq:K_3_tilde}
\end{align}
Let
\begin{align}
    K_1(g,g_0) & = \tilde K_1(g,g_0;\nu) + \tilde K_1(g,g_0;1), \label{eq:K_1} \\
    K_2(g,g_0) & = \tilde K_2(g,g_0;\nu) + \tilde K_2(g,g_0;1), \label{eq:K_2} \\
    K_3(g,g_0) & = \tilde K_3(g,g_0;\nu) + \tilde K_3(g,g_0;1), \label{eq:K_3}
\end{align}
where $\tilde K_j(g,g_0;1) $ is $ \tilde K_j(g,g_0;\nu)$ with $\nu\equiv 1$, $j=1,2,3$. 
We note that $K_1$, $K_2$ are the integral products for $(F-F_0, G-G_0)$ and $(F-F_0,S_D-S_{D0})$, respectively, and $K_3$ contains the higher order integral product terms.

Finally denote 
\begin{align}
    \D_\V(g,g_0;\nu) 
    & = \{\V(\nu;F,G,S_D) - \V(\nu;F_0,G,S_D)\}  \nonumber \\
    &\quad - \{\V(\nu;F,G_0,S_{D0}) - \V(\nu;F_0,G_0,S_{D0})\}. \label{eq:D_V}
\end{align}

\newpage
\subsection{Useful propositions}

We first introduce the following properties of $\V$, which will be used in the later proofs.

\begin{proposition}\label{prop:W_nuf(A,Z)}
    Under Assumptions \ref{ass:consistency} - \ref{ass:strict_positivity}, for any function $\zeta(A^*,Z^*)$ that is bounded almost surely,  we have 
    \begin{align}
    \E[\V(\nu;F,G,S_{D})\zeta(A,Z)] & = \beta^{-1} \E\{ \nu(T^*)\zeta(A^*,Z^*)\}, \nonumber \\
    \E[\V(1;F,G,S_{D})\zeta(A,Z)] & = \beta^{-1} \E\{\zeta(A^*,Z^*)\}  \nonumber
\end{align}
if either $F = F_0$ or $(G,S_D) = (G_0,S_{D0})$. 
\end{proposition}

\begin{proof}[Proof of Proposition \ref{prop:W_nuf(A,Z)}]
    For any function $\zeta(A,Z)$, by the expressions of $\V_Q$ in \eqref{eq:AIPW_Q} and $\V_C$ in \eqref{eq:AIPCW}, we can show that 
    \begin{align*}
        \V(\nu\cdot\zeta;F,G,S_D)
        & = \V_C\{ \V_Q(\nu\cdot\zeta;F,G); F,S_D\} \\
        & = \V_C\{ \V_Q(\nu;F,G)\cdot \zeta(A,Z); F,S_D\} \\
        & = \V_C\{\V_Q(\nu;F,G); F,S_D\} \cdot \zeta(A,Z)\\
        & = \V(\nu;F,G,S_D) \cdot \zeta(A,Z).
    \end{align*}
This together with Lemma \ref{lem:DR_V} implies that 
\begin{align}
    \E\{ \V(\nu;F,G,S_D) \zeta(A,Z)\} = \E\{ \V(\nu\cdot\zeta;F,G,S_D) \} = \beta^{-1} \E\{\nu(T^*)\zeta(A^*,Z^*)\} \nonumber
\end{align}
if either $F = F_0$ or $(G,S_D) = (G_0,S_{D0})$. 
As a special case of $\nu\equiv 1$, we have 
\begin{align}
    \E[\V(1;F,G,S_{D})\zeta(A,Z)] & = \beta^{-1} \E\{\zeta(A^*,Z^*)\} \nonumber
\end{align}
if either $F = F_0$ or $(G,S_D) = (G_0,S_{D0})$.
\end{proof}

\vspace{3em}

Recall from \eqref{eq:D_V} the definition of $\D_\V$.
We now introduce the following properties of $\D_\V$, which will be used in the later proofs.

\begin{proposition}\label{prop:D_W_bound}
    \begin{align*}
        |\D_\V(g,g_0;\nu)|
        & \lesssim \left\{ \sup_t |F(t|A,Z) - F_0(t|A,Z)| \right\} \left\{ \sup_t |G(t|A,Z) - G_0(t|A,Z)| \right\} . \\
        &\qquad + \left\{ \sup_t |F(t|A,Z) - F_0(t|A,Z)| \right\} \left\{ \sup_t |S_D(t|Q,A,Z) - S_{D0}(t|Q,A,Z)| \right\} \\
        &\qquad + \tilde K_1(g,g;\nu) + \tilde K_2(g,g;\nu) + \tilde K_3(g,g;\nu),
    \end{align*}
    \begin{align*}
        |\D_\V(g,g_0;1)|
        & \lesssim \left\{ \sup_t |F(t|A,Z) - F_0(t|A,Z)| \right\} \left\{ \sup_t |G(t|A,Z) - G_0(t|A,Z)| \right\} \\
        &\qquad + \left\{ \sup_t |F(t|A,Z) - F_0(t|A,Z)| \right\} \left\{ \sup_t |S_D(t|Q,A,Z) - S_{D0}(t|Q,A,Z)| \right\} \\
        &\qquad + \tilde K_1(g,g;1) + \tilde K_2(g,g;1) + \tilde K_3(g,g;1);
    \end{align*}
    and 
    \begin{align*}
        \|\D_\V(g,g_0;\nu)\|_2
        &\lesssim \|F-F_0\|_{\sup,4} \cdot \|G-G_0\|_{\sup,4} + \|F-F_0\|_{\sup,4} \cdot \|S_D - S_{D0}\|_{\sup,4}  \nonumber \\
        &\quad + \|\tilde K_1(g,g_0;\nu)\|_2 + \|\tilde K_2(g,g_0;\nu)\|_2 + \|\tilde K_3(g,g_0;\nu)\|_2, \\
        \|\D_\V(g,g_0;1)\|_2
        &\lesssim \|F-F_0\|_{\sup,4} \cdot \|G-G_0\|_{\sup,4} + \|F-F_0\|_{\sup,4} \cdot \|S_D - S_{D0}\|_{\sup,4}  \nonumber \\
        &\quad + \|\tilde K_1(g,g_0;1)\|_2 + \|\tilde K_2(g,g_0;1)\|_2 + \|\tilde K_3(g,g_0;1)\|_2.
\end{align*}
\end{proposition}

\begin{proof}[Proof of Proposition \ref{prop:D_W_bound}]
    Recall from \eqref{eq:V(nu)} the expression of $\V(\nu;F,G,S_D)$: 
    \begin{align}
    \V(\nu;F,G,S_D)  
    & = \frac{\Delta}{S_D(X-Q|Q,A,Z)} \left\{\frac{\nu(X)}{G(X|A, Z)} - \int_0^\infty \frac{\int_0^v \nu(t) dF(t|A,Z) }{1-F(v|A, Z)} \cdot \frac{d\bar M_Q(v;G)}{G(v|A, Z)} \right\}  \nonumber \\
    &\quad + \int_0^\infty \frac{\int_{Q+u}^\infty \{\nu(t)/G(t|A,Z)\} dF(t|A,Z) }{1-F(Q+u|A,Z)} \cdot \frac{dM_D(u;S_D)}{S_D(u|Q,A,Z)}  \nonumber \\
    &\quad +  \left\{1 - \frac{\Delta}{S_D(X-Q|Q,A,Z)}\right\} \frac{\int_0^Q \nu(t) dF(t|A,Z) }{\{1-F(Q|A, Z)\}G(Q|A,Z)}   \nonumber \\
    &\quad - \int_0^\infty \left[\int_Q^\infty \frac{\int_0^v \nu(t) dF(t|A,Z)}{1-F\{(Q+u)\wedge v|A,Z)\}} \cdot \frac{d G(v|A,Z)}{G(v|A,Z)^2} \right] \frac{dM_D(u;S_D)}{S_D(u|Q,A,Z)}. \nonumber
\end{align}
It can be verified that 
    \begin{align*}
        \D_\V(g,g_0;\nu) = - \D_1 + \D_2 + \D_3 - \D_4,
    \end{align*}
    where 
    \begin{align*}
        \D_1 & = \frac{\Delta}{S_D(X-Q|Q,A,Z)} \int_0^ \infty \left[ \frac{\int_0^v \nu(t) dF(t|A,Z) }{1-F(v|A, Z)} - \frac{\int_0^v \nu(t) dF_0(t|A,Z) }{1-F_0(v|A, Z)} \right] \cdot \frac{d\bar M_Q(v;G)}{G(v|A, Z)} \\
        & \quad - \frac{\Delta}{S_{D0}(X-Q|Q,A,Z)} \int_0^ \infty \left[ \frac{\int_0^v \nu(t) dF(t|A,Z) }{1-F(v|A, Z)} - \frac{\int_0^v \nu(t) dF_0(t|A,Z) }{1-F_0(v|A, Z)} \right] \cdot \frac{d\bar M_Q(v;G_0)}{G_0(v|A, Z)}, \\
        \D_2 & = \int_0^\infty \left\{\frac{\int_{Q+u}^\infty \{\nu(t)/G(t|A,Z)\} dF(t|A,Z) }{1-F(Q+u|A,Z)} - \frac{\int_{Q+u}^\infty \{\nu(t)/G(t|A,Z)\} dF_0(t|A,Z) }{1-F_0(Q+u|A,Z)}  \right\}\cdot \frac{dM_D(u;S_D)}{S_D(u|Q,A,Z)} \\
        & \quad - \int_0^\infty \left\{\frac{\int_{Q+u}^\infty \{\nu(t)/G_0(t|A,Z)\} dF(t|A,Z) }{1-F(Q+u|A,Z)} - \frac{\int_{Q+u}^\infty \{\nu(t)/G_0(t|A,Z)\} dF_0(t|A,Z)}{1-F_0(Q+u|A,Z)}  \right\}\cdot \frac{dM_D(u;S_{D0})}{S_{D0}(u|Q,A,Z)},\\
        \D_3 & = \left[\frac{\int_0^Q \nu(t) dF(t|A,Z) }{\{1-F(Q|A, Z)\}} - \frac{\int_0^Q \nu(t) dF_0(t|A,Z) }{\{1-F_0(Q|A, Z)\}} \right] \\
        &\qquad \cdot \left[  \left\{1 - \frac{\Delta}{S_D(X-Q|Q,A,Z)}\right\} \frac{1}{G(Q|A,Z)} -  \left\{1 - \frac{\Delta}{S_{D0}(X-Q|Q,A,Z)}\right\}\frac{1}{G_0(Q|A,Z)} \right]
    \end{align*}
    \begin{align*}
        \D_4 & = \int_0^\infty \left[\int_Q^\infty \frac{\int_0^v \nu(t) dF(t|A,Z)}{1-F\{(Q+u)\wedge v|A,Z)\}} - \frac{\int_0^v \nu(t) dF_0(t|A,Z)}{1-F_0\{(Q+u)\wedge v|A,Z)\}} \right] \cdot \frac{d G(v|A,Z)}{G(v|A,Z)^2} \cdot \frac{dM_D(u;S_D)}{S_D(u|Q,A,Z)} \\
        & \quad -  \int_0^\infty \left[\int_Q^\infty \frac{\int_0^v \nu(t) dF(t|A,Z)}{1-F\{(Q+u)\wedge v|A,Z)\}} - \frac{\int_0^v \nu(t) dF_0(t|A,Z)}{1-F_0\{(Q+u)\wedge v|A,Z)\}} \right] \cdot \frac{dG_0(v|A,Z)}{G_0(v|A,Z)^2} \cdot \frac{dM_D(u;S_{D0})}{S_{D0}(u|Q,A,Z)}. 
    \end{align*}

    We first bound $|\D_1|$. We have
    \begin{align*}
        \D_1 = \D_{11} + \D_{12},
    \end{align*}
    where
    \begin{align*}
        \D_{11} 
        & = \frac{\Delta}{S_D(X-Q|Q,A,Z)} \int_0^ \infty \left[ \frac{\int_0^v \nu(t) dF(t|A,Z) }{1-F(v|A, Z)} - \frac{\int_0^v \nu(t) dF_0(t|A,Z) }{1-F_0(v|A, Z)} \right] \cdot \left\{ \frac{d\bar M_Q(v;G)}{G(v|A, Z)} - \frac{d\bar M_Q(v;G_0)}{G_0(v|A, Z)}\right\}, \\
        \D_{12} & = \left\{\frac{\Delta}{S_D(X-Q|Q,A,Z)} - \frac{\Delta}{S_{D0}(X-Q|Q,A,Z)} \right\}  \\
        &\qquad \cdot\int_0^ \infty \left[ \frac{\int_0^v \nu(t) dF(t|A,Z) }{1-F(v|A, Z)} - \frac{\int_0^v \nu(t) dF_0(t|A,Z) }{1-F_0(v|A, Z)} \right] \cdot \frac{d\bar M_Q(v;G_0)}{G_0(v|A, Z)}. 
    \end{align*}
    By the definition of $\bar M_Q$ in \eqref{eq:def_MQ}, we have 
    \begin{align*}
        \D_{11} & = \D_{111} + \D_{112},
    \end{align*}
    where 
    \begin{align*}
        \D_{111} & = - \frac{\Delta}{S_D(X-Q|Q,A,Z)} \cdot\left[ \frac{\int_0^Q \nu(t) dF(t|A,Z) }{1-F(Q|A, Z)} - \frac{\int_0^Q \nu(t) dF_0(t|A,Z) }{1-F_0(Q|A, Z)} \right] \cdot \left\{ \frac{1}{G(Q|A, Z)} - \frac{1}{G_0(Q|A, Z)}\right\}, \\
        \D_{112} & = - \frac{\Delta}{S_D(X-Q|Q,A,Z)} \cdot H_1(g,g_0;\nu),
    \end{align*}
    and recall that $H_1(g,g_0;\nu)$ is defined in \eqref{eq:H1}.  
    By Assumption \ref{ass:strict_positivity}, we have 
    \begin{align*}
        |\D_{111}| 
        & \lesssim \left| \frac{\int_0^Q \nu(t) dF(t|A,Z) }{1-F(Q|A, Z)} - \frac{\int_0^Q \nu(t) dF_0(t|A,Z) }{1-F_0(Q|A, Z)} \right| \cdot \left| \frac{1}{G(Q|A, Z)} - \frac{1}{G_0(Q|A, Z)}\right| \\
        & \lesssim \left\{ \sup_t |F(t|A,Z) - F_0(t|A,Z)| \right\} \left\{ \sup_t |G(t|A,Z) - G_0(t|A,Z)| \right\}, \\
        |\D_{112}| 
        & \lesssim |H_1(g,g;\nu)|.
    \end{align*}
    Similarly, we can show that 
    \begin{align*}
        |\D_{12}| \lesssim \left\{ \sup_t |F(t|A,Z) - F_0(t|A,Z)| \right\} \left\{ \sup_t |S_D(t|Q,A,Z) - S_{D0}(t|Q,A,Z)| \right\}. 
    \end{align*}
    Combining the above, we have 
    \begin{align*}
        |\D_1| 
        & \leq |\D_{11}| + |\D_{12}| \\
        & \leq |\D_{111}| + |\D_{112}|  + |\D_{12}| \\
        & \lesssim \left\{ \sup_t |F(t|A,Z) - F_0(t|A,Z)| \right\} \left\{ \sup_t |G(t|A,Z) - G_0(t|A,Z)| \right\} \\
        &\quad + \left\{ \sup_t |F(t|A,Z) - F_0(t|A,Z)| \right\} \left\{ \sup_t |S_D(t|Q,A,Z) - S_{D0}(t|Q,A,Z)| \right\} \\
        &\quad + |H_1(g,g;\nu)|. 
    \end{align*}

    With similar proofs, we can bound $|\D_2|$, $|\D_3|$, $|\D_4|$, and we have 
    \begin{align*}
        |\D_\V(g,g_0;\nu)|
        & \leq |\D_1| + |\D_2| + |\D_3| + |\D_4| \\
        &\lesssim \left\{ \sup_t |F(t|A,Z) - F_0(t|A,Z)| \right\} \left\{ \sup_t |G(t|A,Z) - G_0(t|A,Z)| \right\} \\
        &\quad + \left\{ \sup_t |F(t|A,Z) - F_0(t|A,Z)| \right\} \left\{ \sup_t |S_D(t|Q,A,Z) - S_{D0}(t|Q,A,Z)| \right\} \\
        &\quad + \tilde K_1(g,g;\nu) + \tilde K_2(g,g;\nu) + \tilde K_3(g,g;\nu),
    \end{align*}
    where $\tilde K_1$, $\tilde K_2$, and $\tilde K_3$ are defined in \eqref{eq:K_1_tilde} \eqref{eq:K_2_tilde} and \eqref{eq:K_3_tilde}, respectively.
    
    Likewise, we can show the upper bound for $|\D_\V(g,g_0;1)|$.

    We now show the two inequalities for $\|\D_\V(g,g_0;\nu)\|_2$ and $\|\D_\V(g,g_0;1)\|_2$. 
    We have 
    \begin{align}
        &\quad \|\D_\V(g,g_0;\nu)\|_2^2 \\
        & = \E\{|\D_\V(g,g_0;\nu)|^2\} \nonumber \\
        & \lesssim \E\left[ \left\{ \sup_t |F(t|A,Z) - F_0(t|A,Z)| \right\}^2 \left\{ \sup_t |G(t|A,Z) - G_0(t|A,Z)| \right\}^2 \right. \nonumber \\
        &\qquad\qquad + \left\{ \sup_t |F(t|A,Z) - F_0(t|A,Z)| \right\}^2 \left\{ \sup_t |S_D(t|Q,A,Z) - S_{D0}(t|Q,A,Z)| \right\}^2 \nonumber \\
        &\qquad\qquad \left. + \tilde K_1(g,g;\nu)^2 + \tilde K_2(g,g;\nu)^2 + \tilde K_3(g,g;\nu)^2 \right] \label{eq:thm3_proof_DV_1} \\
        &\leq \|F-F_0\|_{\sup,4}^2 \cdot \|G-G_0\|_{\sup,4}^2 + \|F-F_0\|_{\sup,4}^2 \cdot \|S_D - S_{D0}\|_{\sup,4}  \nonumber \\
        &\quad + \|\tilde K_1(g,g_0;\nu)\|_2^2 + \|\tilde K_2(g,g_0;\nu)\|_2^2 + \|\tilde K_3(g,g_0;\nu)\|_2^2. \label{eq:thm3_proof_DV_2} 
\end{align}
where \eqref{eq:thm3_proof_DV_1} uses the fact that for any real numbers $a$ and $b$, $(a+b)^2 \lesssim a^2 + b^2$, and \eqref{eq:thm3_proof_DV_2} is by Cauchy-Schwartz inequality.  
So 
\begin{align*}
        \|\D_\V(g,g_0;\nu)\|_2
        &\lesssim \|F-F_0\|_{\sup,4} \cdot \|G-G_0\|_{\sup,4} + \|F-F_0\|_{\sup,4} \cdot \|S_D - S_{D0}\|_{\sup,4}  \nonumber \\
        &\quad + \|\tilde K_1(g,g_0;\nu)\|_2 + \|\tilde K_2(g,g_0;\nu)\|_2 + \|\tilde K_3(g,g_0;\nu)\|_2.
\end{align*}
Likewise, we can show that 
\begin{align*}
        \|\D_\V(g,g_0;1)\|_2
        &\lesssim \|F-F_0\|_{\sup,4} \cdot \|G-G_0\|_{\sup,4} + \|F-F_0\|_{\sup,4} \cdot \|S_D - S_{D0}\|_{\sup,4}  \nonumber \\
        &\quad + \|\tilde K_1(g,g_0;1)\|_2 + \|\tilde K_2(g,g_0;1)\|_2 + \|\tilde K_3(g,g_0;1)\|_2.
\end{align*}
   
\end{proof}

\begin{proposition}\label{prop:D_W}
    Under Assumptions \ref{ass:consistency} - \ref{ass:strict_positivity}, for any function $\zeta(A,Z)$ that is bounded almost surely,  we have 
    \begin{align}
     \E[\{\V(\nu;g) - \V(\nu;g_0)\}\zeta(A,Z)] = \E\{ \D_\V(g,g_0;\nu) \zeta(A,Z)\}. \nonumber
\end{align}
\end{proposition}

\begin{proof}[Proof of Proposition \ref{prop:D_W}]
Recall that $g = (F,G,S_D)$.
By Proposition \ref{prop:W_nuf(A,Z)}, for any function $\zeta(A,Z)$, 
\begin{align*}
    \E[\V(\nu;F_0,G,S_{D})\zeta(A,Z)]
    = \E[\V(\nu;F,G_0,S_{D0})\zeta(A,Z)]
    = \E[\V(\nu;F_0,G_0,S_{D0})\zeta(A,Z)].
\end{align*}
Therefore, 
\begin{align*}
    \E\{ \D_\V(g,g_0;\nu) \zeta(A,Z)\}
    & = \E[\V(\nu;F,G,S_{D})\zeta(A,Z)] - \E[\V(\nu;F_0,G,S_{D})\zeta(A,Z)] \\
    &\quad -  \E[\V(\nu;F,G_0,S_{D0})\zeta(A,Z)] + \E[\V(\nu;F_0,G_0,S_{D0})\zeta(A,Z)] \\
    & = \E[\V(\nu;F,G,S_{D})\zeta(A,Z)] - \E[\V(\nu;F_0,G_0,S_{D0})\zeta(A,Z)] \\
    & = \E[\{\V(\nu;F,G,S_D) - \V(\nu;F_0,G_0,S_{D0})\}\zeta(A,Z)].
\end{align*}
\end{proof}

The following proposition relates expectations in the truncation-free data and expectations in the truncated data, which will be useful in the later proofs. 

\begin{proposition}\label{prop:equiv_norms}
    Under Assumptions \ref{ass:consistency} -  \ref{ass:trunc} and \ref{ass:strict_positivity}, for any nonnegative function $\zeta(A^*,Z^*)$ that is bounded almost surely, we have 
    \begin{align*}
        & \beta \E\{\zeta(A,Z)\} \leq \E\{\zeta(A^*,Z^*)\} \leq \beta\delta^{-1} \E\{\zeta(A,Z)\}. 
    \end{align*}
\end{proposition}

\begin{proof}[Proof of Proposition \ref{prop:equiv_norms}]
    Recall that $p$ denotes the density functions (or probability mass functions in the case of discrete variables),
    and $G$ denotes the conditional cumulative distribution function (CDF) of $Q^*$ given $(A^*,Z^*)$.
    Under Under Assumptions \ref{ass:consistency} -  \ref{ass:trunc} and \ref{ass:strict_positivity}, we have
\begin{align}
    p_{Q,T,A,Z}(q,t,a,z) = \frac{1}{\beta} \cdot p_{Q^*|A^*,Z^*}(q|a,z) p_{T^*|A^*,Z^*}(t|a,z) p_{A^*,Z^*}(a,z) \ind(q<t). \nonumber
\end{align}
Integrating out $q$ from the above, we have 
\begin{align}
    p_{T,A,Z}(t,a,z) = \frac{G(t|a,z)}{\beta} \cdot p_{T^*|A^*,Z^*}(t|a,z) p_{A^*,Z^*}(a,z)
    =  \frac{G(t|a,z)}{\beta} \cdot p_{T^*,A^*,Z^*}(t,a,z). \nonumber
\end{align}
We have 
\begin{align}
    \E\{\zeta(A^*,Z^*)\}
    & = \int \zeta(a,z)  \cdot p_{T^*,A^*,Z^*}(t,a,z) dt\ da\ dz \nonumber \\
    & = \int \zeta(a,z) \cdot \frac{\beta}{G(t|a,z)} \cdot p_{T,A,Z}(t,a,z)\ dt\ da\ dz \nonumber\\
    & \leq \frac{\beta}{\delta}\int \zeta(a,z) \cdot p_{T,A,Z}(t,a,z)\ dt\ da\ dz
    = \frac{\beta}{\delta}\E\{\zeta(A,Z)\}, \nonumber
\end{align}
since $\zeta$ is  nonnegative and $ G(t|a,z) \geq \delta$ on the support of $p_{T,A,Z}(t,a,z) $ by Assumption \ref{ass:strict_positivity}.
In addition since $ G(t|a,z) \leq 1$, we have 
\begin{align}
    \E\{\zeta(A^*,Z^*)\}
    & \geq \beta\int \zeta(a,z) \cdot p_{T,A,Z}(t,a,z)\ dt\ da\ dz 
    = \beta \E\{\zeta(A,Z)\}. \nonumber
\end{align}
\end{proof}

\subsection{Derivations of $U$, $\tilde\ell_{\RR}$ and $\tilde\ell_{\DR}$}\label{app:derivations}


The derivation of $U$ in \eqref{eq:AIPW_3bias} uses the following properties of $\V$ (which can be directly verified from the expression of $\V$): 
\begin{enumerate}
    \item For any function $\zeta(A^*,Z^*)$, we have $\V(\zeta;F,G,S_D) = \V(1;F,G,S_D)\zeta(A,Z)$; in addition, considering the function $\zeta\cdot \nu = \zeta(A^*,Z^*)\nu(T^*)$, we have $\V(\zeta\cdot \nu;F,G,S_D) = \V(\nu;F,G,S_D)\zeta(A,Z); $
    \item For any functions $\zeta_1(A^*,Z^*)$ and $\zeta_2(A^*,Z^*)$
Recall $\V({\zeta_1 + \zeta_2};F,G,S_D) = \V(\zeta_1;F,G,S_D) + \V(\zeta_2;F,G,S_D)$.
\end{enumerate}
The expression of $U$ then follows by applying these properties of $\V$.

\vspace{1em}

In the following we show the derivation of the ltrcR-loss and the ltrcDR-loss. In the derivation, we will use ``\cconst" to denote a generic constant factor that does not involve $\tau$; its value may vary in different occurrences. 

We have 
\begin{align*}
    \ell_\RR(\tau;\pi,F) 
    & = \left[\{\nu(T^*) - \tilde\mu(Z^*;\pi,F)\} - \{A^*-\pi(Z^*)\} \tau(Z^*) \right]^2 \\
    & = \{A^*-\pi(Z^*)\}^2 \tau(Z^*)^2 - 2\{A^*-\pi(Z^*)\}\{\nu(T^*) - \tilde\mu(Z^*;\pi,F)\}\tau(Z^*) + \cconst,
\end{align*}
By the two properties of $\V$ that are stated above,
\begin{align*}
    \V(\ell_{\RR};F,G,S_D) 
    & = \{A-\pi(Z)\}^2 \tau(Z)^2 \V(1;F,G,S_D) \\
    &\quad - 2\{A-\pi(Z)\}\{\V(\nu;F,G,S_D) - \tilde\mu(Z;\pi,F)\V(1;F,G,S_D)\}\tau(Z) + \cconst \\
    & = \V(1;F,G,S_D) 
      \left[ \frac{\V(\nu;F,G,S_D)}{\V(1;F,G,S_D)} - \tilde\mu(Z;\pi,F) - \{A-\pi(Z)\}\tau(Z) \right]^2  + \cconst.
\end{align*}

Likewise, we can show that 
\begin{align*}
    \V(\ell_{\DR};F,G,S_D)  
     & = \V(1;F,G,S_D) 
      \left[ \frac{A - \pi(Z)}{\pi(Z)\{1-\pi(Z)\}} \left\{ \frac{\V(\nu;F,G,S_D)}{\V(1;F,G,S_D)} - \mu(A,Z;F) \right\} \right. \\
     & \qquad\qquad\qquad\qquad\quad  + \mu(1,Z;F) - \mu(0,Z;F) - \tau(V) \Bigg]^2 + \cconst.
\end{align*}

\clearpage
\section{Proofs of  double robustness and Neyman orthogonality}\label{app:DR_proof}


\subsection{Proof of Lemma \ref{lem:DR_V}}\label{app:DR_lemmas_proofs}

We first show the double robustness of $\V_Q$ and $\V_C$, which are formally stated in the following two lemmas below. Then the double robustness of $\V = \V_C\circ\V_Q$ follows from the double robustness of $\V_Q$ and $\V_C$. 

The operator $\V_Q$ is doubly robust in the sense that is formally stated in Lemma \ref{lem:DR_VQ} below, the proof of which can be shown in a similar way as \citet[Proof of Theorem 1]{wang2024doubly}. 

\begin{appxlemma}[Double robustness of $\V_Q$] \label{lem:DR_VQ}
Under Assumptions \ref{ass:consistency} - \ref{ass:trunc}, and \ref{ass:strict_positivity}(ii), for any function $\zeta(T^*,A^*,Z^*)$ that is bounded almost surely,  
\begin{align}
    \E\{\V_Q(\zeta; F,G)(Q,T,A,Z)\} = \beta^{-1}\E\{\zeta(T^*,A^*,Z^*)\} \nonumber
\end{align}
if either $F = F_0$ or $G = G_0$.
\end{appxlemma}

Like the augmented inverse probability of censoring weighting (AIPCW) in \citet{rotnitzky2005inverse}, the operator $\V_C$ is doubly robust in the sense that is formally stated in Lemma \ref{lem:VC} below. This property is implied from the the double robustness of the AIPCW in \citet{rotnitzky2005inverse} by considering the residual censoring time $D$ and the residual event time $T-Q$ together with the independence condition $D\bigCI (T-Q)\mid Q,A,Z$, which can be implied from Assumptions \ref{ass:consistency} - \ref{ass:cen}.

\begin{appxlemma}[Double robustness of $\V_C$] \label{lem:VC}
Under Assumptions \ref{ass:consistency} - \ref{ass:strict_positivity},
for any function $\xi(Q,T,A,Z)$ that is bounded almost surely,  
\begin{align}
    \E\{\V_C(\xi; F,S_D)(Q,X,\Delta,A,Z)\} = \E\{\xi(Q,T,A,Z)\} \nonumber
\end{align}
if either $F = F_0$ or $S_D = S_{D0}$.
\end{appxlemma}

\subsection{Proofs of Lemma \ref{lem:DR_3bias} and Lemma \ref{lem:orthogonal_DR_loss}} 

Lemma \ref{lem:DR_3bias} directly follows from Lemma \ref{lem:DR_V} and the double robustness of $U_0$, i.e., $\E\{U_0(\theta_0;\pi,F)\} = 0$ if either $\pi = \pi_0$ or $F = F_0$.

\begin{proof}[Proof of Lemma \ref{lem:orthogonal_DR_loss}]
    To simplify the notation, denote $L^*(\tau,\eta) = \E\{\ell(T^*,A^*,Z^*;\tau,\eta)\}$, and denote $L(\tau, \eta, F,G,S_D) = \E\{\tilde\ell(\tau;\eta,F,G,S_D)\}$. 

    (i) 
    In the following we will show that the relevant directional derivatives 
    are all zero. 

 1)   We first compute $D_\eta D_\tau L(\tau_0, \eta_0, F_0,G_0,S_{D0})[\tau-\tau_0, \eta - \eta_0]$.
    By Lemma \ref{lem:DR_V}, for any $(\tau, \eta, F, G, S_D)$,
    \begin{align}
        L(\tau, \eta, F_0,G_0,S_{D0})
        & = \beta^{-1} L^*(\tau,\eta), \label{eq:L1} \\
        L(\tau, \eta, F_0,G,S_D) 
        & = \beta^{-1} L^*(\tau,\eta), \label{eq:L2}\\
        L(\tau, \eta, F,G_0,S_{D0})
        & = \beta^{-1} L^*(\tau,\eta).  \label{eq:L3}
    \end{align}
    So for any $(\tau, \eta, F, G, S_D)$, 
    \begin{align}
        &\quad D_\eta D_\tau L(\tau_0, \eta_0, F_0,G_0,S_{D0})[\tau-\tau_0, \eta - \eta_0] 
        \nonumber \\
        & = \left.\frac{d}{ds} \left\{ \left. \frac{d}{dt} L(\tau_0+t(\tau-\tau_0), \eta_0 + s(\eta-\eta_0), F_0,G_0,S_{D0}) \right|_{t=0} \right\} \right|_{s=0} \nonumber\\
        & = \beta^{-1}\cdot \left.\frac{d}{ds} \left\{ \left. \frac{d}{dt} L^*(\tau_0+t(\tau-\tau_0), \eta_0 + s(\eta-\eta_0)) \right|_{t=0} \right\} \right|_{s=0} \label{eq:lem3_proof_1}\\
        &= \beta^{-1}\cdot  D_\eta D_\tau L^*(\tau_0,\eta_0)[\tau-\tau_0,\eta - \eta_0], \nonumber
    \end{align}
    where \eqref{eq:lem3_proof_1} holds by \eqref{eq:L1}.
    Since $\ell$ is a Neyman orthogonal loss in the LTRC-free data, we have 
    $D_\eta D_\tau L^*(\tau_0,\eta_0)[\tau-\tau_0,\eta - \eta_0] = 0$. 
    Therefore, 
    \begin{align*}
        D_\eta D_\tau L(\tau_0, \eta_0, F_0,G_0,S_{D0})[\tau-\tau_0, \eta - \eta_0] = 0.
    \end{align*}

 2)   We now compute $D_F D_\tau L(\tau_0, \eta_0, F_0,G_0,S_{D0})[\tau-\tau_0, F-F_0]$.
    \begin{align}
        &\quad D_F D_\tau L(\tau_0, \eta_0, F_0,G_0,S_{D0})[\tau-\tau_0, F-F_0] \nonumber \\
        & = \left.\frac{d}{ds} \left\{ \left. \frac{d}{dt} L(\tau_0+t(\tau-\tau_0), \eta_0, F_0 + s(F-F_0),G_0,S_{D0}) \right|_{t=0} \right\} \right|_{s=0} \nonumber \\
        & = \left.\frac{d}{dt} \left\{ \left. \frac{d}{ds} L(\tau_0+t(\tau-\tau_0), \eta_0, F_0 + s(F-F_0),G_0,S_{D0}) \right|_{s=0} \right\} \right|_{t=0}. \label{eq:lem3_proof_2}
    \end{align}
    By \eqref{eq:L3}, 
    \begin{align*}
        L(\tau_0+t(\tau-\tau_0), \eta_0, F_0 + s(F-F_0),G_0,S_{D0}) 
        = \beta^{-1} L^*(\tau_0+t(\tau-\tau_0), \eta_0)
    \end{align*}
    for all $s\in[0,1]$. So 
    \begin{align*}
        \left. \frac{d}{ds} L(\tau_0+t(\tau-\tau_0), \eta_0, F_0 + s(F-F_0),G_0,S_{D0}) \right|_{s=0} = 0. 
    \end{align*}
    This together with \eqref{eq:lem3_proof_2} imply 
    \begin{align*}
        D_F D_\tau L(\tau_0, \eta_0, F_0,G_0,S_{D0})[\tau-\tau_0, F-F_0] =0. 
    \end{align*}

3) and 4)    Similarly, by \eqref{eq:L2}, we can show that 
    \begin{align}
        & D_G D_\tau L(\tau_0, \eta_0, F_0,G_0,S_{D0})[\tau-\tau_0, G-G_0] 
        = 0, \nonumber \\
        & D_{S_D} D_\tau L(\tau_0, \eta_0, F_0,G_0,S_{D0})[\tau-\tau_0, S_D-S_{D0}] 
        = 0. \nonumber
    \end{align}
    
    Therefore, 
    $\tilde\ell(\tau;\eta,F,G,S_D)$
    is a Neyman orthogonal loss.

    \bigskip
    (ii)  The conclusion follows from Lemma \ref{lem:DR_V}.
\end{proof}

\clearpage
\section{Nuisance estimation}\label{supp:nuis}

\begin{figure}[h]
\centering
 \begin{subfigure}[H]{0.45\textwidth}
		\large{\begin{tikzpicture}[%
		->,
		>=stealth,
		node distance=1cm,
		pil/.style={
			->,
			thick,
			shorten =2pt,}
		]
        \node at (0,0) (0) { };
        \node at (4,0) (2) {$S_D$};
		\node [left=of 2] (1) {$F$};
        \node[below=of 2] (3) {$G$};
        \node[below=of 3] (4) {$\pi$};
        \draw [->] (2) to (3);
        \node at (3, -0.7) {\footnotesize\color{gray}(Eliminate censoring, };
        \node at (3, -1) {\footnotesize\color{gray} truncation still remains) };
        \draw [->] (3) to (4);
		\draw [->] (2) to [out=325, in=45] (4);
  \node at (4.2, -2.3) {\footnotesize\color{gray}(Eliminate truncation) };
	\end{tikzpicture}}
 \caption{ }\label{fig:nuisance_DR}
	\end{subfigure}
 \begin{subfigure}[H]{0.54\textwidth}
		\large{\begin{tikzpicture}[%
		->,
		>=stealth,
		node distance=1cm,
		pil/.style={
			->,
			thick,
			shorten =2pt,}
		]
        \node at (0,0) (0) { };
		\node at (4.5,-1) (F) {$F$};
		\node[right=of F] (Sd) {$S_D$};
        \node[below=of F] (G) {$G$};
        \node[left=of G] (pi) {$\pi$};
        \draw [->] (F) to (G);
        \node at (4.5, -1.8) {\footnotesize\color{gray}(Eliminate truncation) };
        \draw [->] (F) to (pi);
	\end{tikzpicture}}
\caption{ } \label{fig:nuisance_nonDR}
\end{subfigure}
\par\vspace{2em}
\begin{subfigure}[H]{0.45\textwidth}
		\large{\begin{tikzpicture}[%
		->,
		>=stealth,
		node distance=1cm,
		pil/.style={
			->,
			thick,
			shorten =2pt,}
		]
        \node at (0.7,0) (0) { };
        \node at (4,0) (2) {$S_D$};
		\node [left=of 2] (1) {$F$};
        \node[below=of 2] (3) {$G$};
        \node[below=of 3] (4) {$\pi$};
        \draw [->] (1) to (3);
        \node at (2.5, -0.8) {\footnotesize\color{gray}(Eliminate truncation)};
        \draw [->] (3) to (4);
		\draw [->] (2) to [out=325, in=45] (4);
  \node at (4.2, -2.3) {\footnotesize\color{gray}(Eliminate truncation) };
	\end{tikzpicture}}
 \caption{ }\label{fig:nuisance_c}
	\end{subfigure}
\begin{subfigure}[H]{0.45\textwidth}
		\large{\begin{tikzpicture}[%
		->,
		>=stealth,
		node distance=1cm,
		pil/.style={
			->,
			thick,
			shorten =2pt,}
		]
		\node at (4.5,-1) (F) {$F$};
		\node[right=of F] (Sd) {$S_D$};
        \node[below=of Sd] (G) {$G$};
        \node[below=of F] (pi) {$\pi$};
        \draw [->] (Sd) to (G);
        \draw [->] (F) to (pi);
        \node at (3.5, -1.8) {\footnotesize\color{gray}(Eliminate truncation) };
        \node at (7.2, -1.6) {\footnotesize\color{gray}(Eliminate censoring, };
        \node at (7.2, -1.9) {\footnotesize\color{gray} truncation still remains)};
	\end{tikzpicture}}
\caption{ } \label{fig:nuisance_d}
\end{subfigure}
 \caption{Schemes of nuisance parameter estimation: (a) preserves double robustness, while (b) - (d) break double robustness. 
 }
 \label{fig:DGM_nuisance_summary}
\end{figure}

For the estimation of $F$ and $S_D$, 
in the observed data $T$ is left truncated by $Q$ and right censored by $C$, and $D$ is right censored by $T-Q$, so $F$ can be estimated using existing  software for LTRC data, and $S_D$ can be estimated using  existing  software for right censored data.

Next for $G$ if there is no right censoring, $Q$ is right truncated by $T$ in the observed data, so estimating $G$ is a dual problem as estimating $1-F$ by considering the reversed time scale. However, this duality no longer exists in the presence of right censoring \citep{wang1991nonparametric}. 
One approach to estimating $G$ is to apply inverse probability of censoring weights (IPCW) $\Delta/\hat S_D(X-Q|Q,A,Z)$, 
and then use  existing  software for left truncated data on the reversed time scale. This  approach was used in \citet{vakulenko2022nonparametric} and \citet{wang2024doubly}, in the case with random censoring.  
Another approach is 
to use inverse probability of truncation weights $1/\{1-\hat F(Q|A,Z)\}$,  
and then use  existing regression software for any continuous outcomes to estimate $G$. This approach 
coincides with  \citet{wang1991nonparametric} for $G$ estimation  without treatment or covariates. 

Finally for $\pi$, directly fitting a regression model of $A$ on $Z$ is subject to bias due to left truncation \citep{cheng2012estimating}. 
\citet{cheng2012estimating} considered estimating $\pi$ under a logistic model with an offset to account for left truncation. 
More generally, if there were no right censoring, one may consider using inverse probability of truncation weights $1/\hat G(T|A,Z)$ or $1/\{1-\hat F(Q|A,Z)\}$ and  regressing $A$ on $Z$. For LTRC data, $T$ is not always observed, so  $1/\hat G(T|A,Z)$ needs to be further adjusted, for example by incorporating additional IPCW and resulting in weights $\Delta/\{\hat G(X|A,Z) \hat S_D(X-Q|Q,A,Z)\}$.

\clearpage
\section{ATE theory}\label{app:ATE_theory}

\subsection{Expressions of the estimators} \label{supp:expression}

The proposed estimators $\hat\theta$ and $\hat\theta_{cf}$ have closed-form expressions. 
In particular, we have 
\begin{align}
    \hat\theta 
    & =  \left\{\sum_{i=1}^n \V_i(1;\hat F, \hat G, \hat S_D) \right\}^{-1}  \nonumber \\
    &  \qquad \cdot \sum_{i=1}^n  \Bigg[  \frac{A_i - \pi(Z_i)}{\hat\pi(Z_i)\{1-\hat\pi(Z_i)\}} \left\{ \V_i(\nu;\hat F, \hat G, \hat S_D) - \V_i(1;\hat F, \hat G, \hat S_D)\mu(A_i,Z_i;\hat F) \right\} \ \nonumber \\
    & \qquad\qquad\qquad  +  \V_i(1;\hat F, \hat G, \hat S_D) \{\mu(1,Z_i;\hat F) - \mu(0,Z_i;\hat F)\} \Bigg]. \label{eq:dr_estimator_expression} 
\end{align}
The estimator $\hat\theta_{cf}$ has a similar expression as $\hat\theta$, with the nuisance estimators in \eqref{eq:dr_estimator_expression} are replaced by the corresponding out-of-fold estimates.


\subsection{Asymptotics for ATE estimation}\label{app:ATE_asymptotics}

The following two assumptions are  similar to those in \cite{wang2024doubly}; see \cite{wang2024doubly} for more discussion.
We first define some norms. 
For a random function $X(t,q,a,z)$ with $t\in[\tau_1,\tau_2]$ and $(q,a,z)$ in the support of $(Q,A,Z)$, define $\|X\| = \E\{\sup_{t\in[\tau_1,\tau_2]}|X(t,Q,A,Z)|^2\}^{1/2}$ and \\
$\|X\|_{\TV} = \E[\TV\{X(\cdot,Q,A,Z)\}^2]^{1/2}$,
where the expectations are taking with respect to both the randomness of $X$ and $(Q,A,Z)$,   
and $\TV\{X(\cdot,q,a,z)\}$ is the total variation of $X(\cdot, q,a,z)$ on the interval $[\tau_1,\tau_2]$. 
As a special case of a random function $X(z)$ that does not involve $(t,q,a)$, $\|X\| = \E\{|X(Z)|^2\}^{1/2}$.
Let $\Z$ denote the support of $Z^*$.

\begin{appxassumption}[Uniform Convergence]\label{ass:uniformcons1}
	There exist cumulative distribution functions $F^\divideontimes$,  $G^\divideontimes$, $S_D^\divideontimes$ and propensity function $\pi^\divideontimes(Z)$ such that
	\begin{align}
	\|\hat\pi - \pi^\divideontimes\| = o(1), \quad  
    \|\hat F - F^\divideontimes\|= o(1), 
	\quad \|\hat G  - G^\divideontimes \| = o(1), 
    \quad \|\hat S_D- S_D^\divideontimes\| = o(1). \nonumber
	\end{align}
\end{appxassumption}

\begin{appxassumption}[Asymptotic Linearity]\label{ass:AL}
	For fixed $(t,a,z)\in[0,\infty) \times \{0,1\}\times \Z$, $\hat\pi(z)$, $\hat F(t|a,z)$, $\hat G(t|a,z)$, and $\hat S_D(t|q,a,z)$are regular and asymptotically linear estimators for $\pi(z)$, $F(t|a,z)$, $G(t|a,z)$, and $S_D(t|q,a,z)$ with influence functions  
 $\xi_1(z,O)$, $\xi_2(t,a,z,O)$, $\xi_3(t,a,z,O)$, and $\xi_4(t,q,a,z,O)$, respectively. Denote the residual terms: 
    \begin{align}
        R_1(z) & = \hat\pi(z) - \pi^\divideontimes(z) - \frac{1}{n}\sum_{i = 1}^n \xi_{1}(z, O_i), \nonumber \\
        R_2(t,a,z) & = \hat F(t|a,z) - F^\divideontimes(t|a,z) - \frac{1}{n}\sum_{i = 1}^n \xi_{2}(t, a, z, O_i), \nonumber\\
        R_3(t,a,z) & = \hat G(t|a,z) - G^\divideontimes(t|a,z) - \frac{1}{n}\sum_{i = 1}^n \xi_{3}(t, a, z, O_i), \nonumber \\
        R_4(t,q,a,z) & = \hat S_D(t|q,a,z) - S_D^\divideontimes(t|q,a,z) - \frac{1}{n}\sum_{i = 1}^n \xi_{4}(t, q,a, z, O_i). \nonumber
    \end{align}
    Suppose that the residual terms satisfy 
    $\left\|R_1\right\| = o(n^{-1/2})$,
    $\left\|R_2\right\| = o(n^{-1/2})$, 
    $\left\|R_3\right\| = o(n^{-1/2})$, 
    $\left\|R_4\right\| = o(n^{-1/2})$;
    and either (i) $\left\|R_2\right\|_{\TV} = o(1)$ or (ii) $\left\|R_3\right\|_{\TV} = o(1)$ and $\left\|R_4\right\|_{\TV} = o(1)$.
    Furthermore, suppose that $\sup_t|\xi_{2}(t, A_1, Z_1, O_2)|$, $\sup_t|\xi_{3}(t, A_1, Z_1, O_2)|$, and $\sup_t|\xi_{4}(t, Q_1, A_1, Z_1, O_2)|$ are all bounded almost surely.
\end{appxassumption}

Theorem \ref{thm:mdr} can be shown using the same techniques as \citet[Proof of Theorem S2 in the Supplementary Material]{wang2024doubly}.

\bigskip
\begin{proof}[Proof of Theorem \ref{thm:rdr}]

    (1) The consistency of $\hat\theta_{cf}$ can be shown similarly as the consistency proof of Theorem \ref{thm:mdr}.
    
    (2) We now prove asymptotic normality. 
    Recall that $K$ is the total number of folds in the cross-fitting procedure, which is fixed. Without loss of generality, assume $n = Km$. Then the number of subjects in the $k$-th fold $|\I_k| = m$ for all $k = 1, ..., K$. 
    We have
	\begin{align}
	n^{1/2}(\hat\theta_{cf} - \theta_{0} )
	& = \left. n^{1/2}\left[\frac{1}{K}\sum_{k=1}^K\frac{1}{m}\sum_{i\in\I_k} U_{i}\{\theta_{0};\hat\pi^{(-k)}, \hat F^{(-k)},\hat G^{(-k)}, \hat S_{D}^{(-k)} \} \right] \right. \nonumber \\
    &\qquad\qquad 
    \left/\left[\frac{1}{K}\sum_{k=1}^K\frac{1}{m}\sum_{i\in\I_k} \V_i\{1;\hat F^{(-k)},\hat G^{(-k)},\hat S_D^{(-k)} \} \right] \right. . \label{eq:cf_est_err}
	\end{align}
 Using similar techniques as \citet[Proof of Theorem S2 in the Supplementary Material]{wang2024doubly}, we can show that the denominator of the right hand side (R.H.S.) of \eqref{eq:cf_est_err} converges to $\beta^{-1}$ in probability. 

 In the following,  we will show that for any $k\in\{1,...,K\}$,
	\begin{align}
	\frac{1}{m}\sum_{i\in\I_k} U_{i}\{\theta_{0};\hat\pi^{(-k)}, \hat F^{(-k)},\hat G^{(-k)}, \hat S_{D}^{(-k)} \} 
	&=  \frac{1}{m}\sum_{i\in\I_k} U_{i}(\theta_{0};\pi_0,F_0,G_0,S_{D0}) + o_p(n^{-1/2}). \label{eq:ANproof_4}  
	\end{align} 
	Equation \eqref{eq:ANproof_4} implies that the numerator of the R.H.S of \eqref{eq:cf_est_err} can be written as \\$n^{-1/2}\sum_{i=1}^n U_{i}(\theta_{0};\pi_0,F_0,G_0,S_{D0}) + o_p(1)$, which converges in distribution to \\ 
    $N\left(0, \E\left\{U(\theta_{0};\pi_0,F_0,G_0,S_{D0})^2\right\}\right)$ by the central limit theorem. 
	Therefore, by Slutsky's Theorem, 
	$n^{1/2}(\hat\theta_{cf} - \theta_0) \convd N \left(0, \beta^2 \E\{U(\theta_{0};\pi_0,F_0,G_0,S_{D0})^2\}\right)$. 

 We now show \eqref{eq:ANproof_4}. Without loss of generality, consider $k=1$ and suppose that $\I_1 = \{1,...,m\}$.
	For simplification of the notation, we will denote $\hat\pi' = \hat\pi^{(-1)}$, $\hat F' = \hat F^{(-1)}$, $\hat G' = \hat G^{(-1)}$, and $\hat S_D' = \hat S_D^{(-1)}$ the estimates of $\pi$, $F$, $G$, and $S_D$ using the out-of-1-fold sample $\Oc'$, i.e., data indexed by $\{m+1,...,n\}$.
    Recall that $g = (F,G,S_D)$, and denote $g_0 = (F_0,G_0,S_{D0})$, $\hat g' = (\hat F', \hat G', \hat S_D')$.
    To simplify the notation, let 
    \begin{align*}
        \UU_1(\pi) 
        & = \frac{A - \pi(Z)}{\pi(Z)\{1-\pi(Z)\}}, \\
        \UU_2(\pi,F) 
        & =  - \frac{A - \pi(Z)}{\pi(Z)\{1-\pi(Z)\}} \cdot\mu(A,Z;F) + \left\{\mu(1,Z;F) - \mu(0,Z;F)\right\} - \theta_0.
    \end{align*}
    In addition, let $\UU_{1,i}$ and $\UU_{2,i}$ denote the corresponding expressions for $\UU_1$ and $\UU_2$ with $(A,Z)$ replaced by $(A_i,Z_i)$; 
    and $\UU^*_{1}$ and $\UU^*_{2}$ denote the corresponding expressions for $\UU_1$ and $\UU_2$ with $(A,Z)$ replaced by $(A^*,Z^*)$.
    Recall from \eqref{eq:U0} the estimating function $U_0$ in the LTRC-free data, and recall from \eqref{eq:AIPW_3bias} the estimating function $U$ in the observed data. 
    Then we have 
    \begin{align*}
        U_0(\theta_0;\pi,F) 
        & = \UU_1^*(\pi) \nu(T^*) + \UU_2^*(\pi,F), \\
        U(\theta_0; \pi,g)
        & = \UU_1(\pi) \V(\nu; g) + \UU_2(\pi,F)\V(1; g). 
    \end{align*}
    
	We consider the decomposition:
	\begin{align}
	\frac{1}{m}\sum_{i=1}^m U_i(\theta_0; \hat\pi', \hat g')
	&= \B_1 + \B_2 + \B_3, \nonumber 
	\end{align}
    where
	\begin{align}
	\B_1 & = \frac{1}{m}\sum_{i=1}^m U_{i}(\theta_{0}; \pi_0, g_0), \nonumber \\
	\B_2 & = \frac{1}{m}\sum_{i=1}^m \left[ \left\{\UU_{1,i}(\hat\pi')\V_i(\nu; g_0) + \UU_{2,i}(\hat\pi',\hat F')\V_i(1; g_0)\right\}  -  \left\{ \UU_{1,i}(\pi_0) \V_i(\nu; g_0) - \UU_{2,i}(\pi_0,F_0) \V_i(1; g_0) \right\} \right], \nonumber \\
    \B_3  & =  \frac{1}{m}\sum_{i=1}^m \left[ \left\{\UU_{1,i}(\hat\pi')\V_i(\nu; \hat g') + \UU_{2,i}(\hat\pi',\hat F')\V_i(1; \hat g')\right\}  -  \left\{ \UU_{1,i}(\hat\pi') \V_i(\nu; g_0) - \UU_{2,i}(\hat\pi',\hat F') \V_i(1; g_0) \right\} \right] \nonumber \\
    & = \frac{1}{m}\sum_{i=1}^m \left[ \UU_{1,i}(\hat\pi') \left\{\V_i(\nu; \hat g') - \V_i(\nu; g_0) \right\} + \UU_{2,i}(\hat\pi',\hat F') \left\{\V_i(1;\hat g') - \V_i(1;g_0) \right\} \right]. \nonumber
	\end{align}
    In the following, we will bound $\B_2$ and $\B_3$, respectively. 

    We first bound $\B_2$. 
    We have
    \begin{align*}
        \B_2 = \B_{21} + \B_{22} + \B_{23},
    \end{align*}
    where 
    \begin{align*}
        \B_{21} & = \frac{1}{m}\sum_{i=1}^m \left[ \left\{\UU_{1,i}(\hat\pi') - \UU_{1,i}(\pi_0)\right\} \V_i(\nu; g_0) + \left\{\UU_{2,i}(\hat\pi',F_0) - \UU_{2,i}(\pi_0,F_0)\right\}\V_i(1; g_0)\right], \\
        \B_{22} & = \frac{1}{m}\sum_{i=1}^m \left\{\UU_{2,i}(\pi_0,\hat F') - \UU_{2,i}(\pi_0,F_0)\right\}\V_i(1; g_0) \\
        \B_{23} & =  \frac{1}{m}\sum_{i=1}^m \left[ \left\{\UU_{2,i}(\hat\pi',\hat F') - \UU_{2,i}(\hat\pi',F_0)\right\} - \left\{\UU_{2,i}(\pi_0,\hat F') - \UU_{2,i}(\pi_0,F_0)\right\}  \right] \V_i(1; g_0)
    \end{align*}
    We first consider $\B_{21}$. 
    By Proposition \ref{prop:W_nuf(A,Z)}, 
    \begin{align}
        &\quad \E(\B_{21} | \Oc') \nonumber \\
        & = \beta^{-1} \E\left[ \left\{\UU_{1}^*(\hat\pi') - \UU_{1}^*(\pi_0)\right\} \nu(T^*) + \left\{\UU_{2}^*(\hat\pi',F_0) - \UU_{2}^*(\pi_0,F_0)\right\}\right] \nonumber \\
        & = \beta^{-1} \E\left[ \left. \left\{\frac{A^* - \hat\pi'(Z^*)}{\hat\pi'(Z^*)\{1-\hat\pi'(Z^*)\}} - \frac{A^* - \pi_0(Z^*)}{\pi_0(Z^*)\{1-\pi_0(Z^*)\}}\right\} \left\{ \nu(T^*) - \mu(A^*,Z^*;F_0)\right\} \right| \Oc' \right] \nonumber \\
        & = \beta^{-1} \E\left[\left\{ \left. \frac{A^* - \hat\pi'(Z^*)}{\hat\pi'(Z^*)\{1-\hat\pi'(Z^*)\}} - \frac{A^* - \pi_0(Z^*)}{\pi_0(Z^*)\{1-\pi_0(Z^*)\}}\right\} \E\left\{ \left. \nu(T^*) - \mu(A^*,Z^*;F_0) \right| A^*, Z^*, \Oc' \right\}  \right| \Oc' \right] \label{eq:proof_thm2_1} \\
        & = 0, \label{eq:proof_thm2_2}
    \end{align}
    where \eqref{eq:proof_thm2_1} is by the tower property of expectations, and \eqref{eq:proof_thm2_2} holds because $\E\{\nu(T^*) | A^*, Z^*, \Oc' \} = \E\{ \nu(T^*) | A^*, Z^* \} = \mu(A^*,Z^*;F_0)$. 
    In addition, 
    \begin{align*}
        &\quad \var(\B_{21}|\Oc') \\
        & = \frac{1}{m} \cdot  \var\left[ \left\{\frac{A - \hat\pi'(Z)}{\hat\pi'(Z)\{1-\hat\pi'(Z)\}} - \frac{A - \pi_0(Z)}{\pi_0(Z)\{1-\pi_0(Z)\}}\right\} \left\{\V(\nu; g_0) - \mu(A,Z;F_0)\V(1; g_0) \right\}  \Bigg| \Oc'\right] \\
        & = \frac{1}{m} \cdot  \E\left[ \left\{\frac{A - \hat\pi'(Z)}{\hat\pi'(Z)\{1-\hat\pi'(Z)\}} - \frac{A - \pi_0(Z)}{\pi_0(Z)\{1-\pi_0(Z)\}}\right\}^2 \left\{\V(\nu; g_0) - \mu(A,Z;F_0)\V(1; g_0) \right\}^2  \Bigg| \Oc'\right] \\
        &\lesssim \frac{1}{m}\|\hat\pi' - \pi_0\|_{2}^2  \\
        & = o_p(m^{-1}), 
    \end{align*}
    where the last equation holds by Assumption \ref{ass:uniformcons2}. 
    By Assumption \ref{ass:strict_positivity}, 
    $\B_{21}$ is bounded a.s., 
    so $\E\{\var(\B_{21}|\Oc')\} = o(m^{-1})$.  
Therefore, 
\begin{align*}
    \var(\B_{21}) = \E\{\var(\B_{21}|\Oc')\} + \var\{\E(\B_{21}|\Oc')\} = o(m^{-1}). 
\end{align*}
In addition, $\E(\B_{21}) = \E\{\E(\B_{21}|\Oc')\} = 0$.
So by Chebyshev's inequality, $\B_{21} = o_p(m^{-1/2}) = o_p(n^{-1/2})$. 

Likewise, we can show that $\B_{22} = o_p(m^{-1/2}) = o_p(n^{-1/2})$. 

We now consider $\B_{23}$. 
We have 
\begin{align*}
    \B_{23} & =  - \frac{1}{m}\sum_{i=1}^m \left\{\frac{A_i - \hat\pi'(Z_i)}{\hat\pi'(Z_i)\{1-\hat\pi'(Z_i)\}} - \frac{A_i - \pi_0(Z_i)}{\pi_0(Z_i)\{1-\pi_0(Z_i)\}} \right\}  \left\{\mu(A_i,Z_i;\hat F') - \mu(A_i,Z_i;F_0) \right\} \V_i(1; g_0). 
\end{align*}
So 
\begin{align}
    &\quad \E(|\B_{23}||\Oc') \nonumber \\
    &\lesssim \E\left[ \left. \left| \frac{A - \hat\pi'(Z)}{\hat\pi'(Z)\{1-\hat\pi'(Z)\}} - \frac{A - \pi_0(Z)}{\pi_0(Z)\{1-\pi_0(Z)\}} \right| \cdot \left|\mu(A,Z;\hat F') - \mu(A,Z;F_0)\right| \ \right| \Oc' \right] \label{eq:rDR_proof_2}\\
    &\lesssim \E\left[\left. |\hat\pi'(Z)-\pi_0(Z)| \cdot \sup_t \left|\hat F'(t|A,Z) - F_0(t|A,Z)\right|\ \right| \Oc' \right]  \nonumber \\
    &\lesssim \|\hat\pi'-\pi_0\|_{2} \cdot \|\hat F' - F_0\|_{\sup, 2}, \label{eq:rDR_proof_1} \\
    & = o_p(n^{-1/2}),  \label{eq:rDR_proof_3}
\end{align}
where \eqref{eq:rDR_proof_2} holds because under Assumption \ref{ass:strict_positivity}, $\V(1; F_0,G_0,S_{D0})$ is bounded a.s.; \eqref{eq:rDR_proof_1} holds by Cauchy-Schwarz inequality; and \eqref{eq:rDR_proof_3} holds by Assumption \ref{ass:prodrate}.
Since $\B_{23}$ is bounded a.s., we have
$\E(|\B_{23}|) = \E\{\E(|\B_{23}||\Oc')\} = o(n^{-1/2})$.
Therefore, by Markov's inequality,
$\B_{23} = o_p(n^{-1/2})$. 

Combining the above, we have $\B_{2} = o_p(n^{-1/2})$. 

We now bound $\B_3$. 
Recall from \eqref{eq:D_V} the definition of $\D_\V$. 
We consider the decomposition
\begin{align*}
    \B_3 & = \B_{31} + \B_{32} + \B_{33},
\end{align*}
where
\begin{align*}
    \B_{31} & = \frac{1}{m}\sum_{i=1}^m \left[ \UU_{1,i}(\hat\pi') \left\{\V_i(\nu; \hat F',G_0,S_{D0}) - \V_i(\nu; F_0,G_0,S_{D0}) \right\}  \right] \\
    &\qquad\qquad\ \left.  + \UU_{2,i}(\hat\pi',\hat F') \left\{\V_i(1;\hat F',G_0,S_{D0}) - \V_i(1;F_0,G_0,S_{D0}) \right\} \right], \\
     \B_{32} & = \frac{1}{m}\sum_{i=1}^m \left[ \UU_{1,i}(\hat\pi') \left\{\V_i(\nu; F_0,\hat G', \hat S_{D}') - \V_i(\nu; F_0,G_0,S_{D0}) \right\}  \right] \\
    &\qquad\qquad\ \left.  + \UU_{2,i}(\hat\pi',\hat F') \left\{\V_i(1;F_0,\hat G',\hat S_D') - \V_i(1;F_0,G_0,S_{D0}) \right\} \right], \\
    \B_{33} & = \frac{1}{m}\sum_{i=1}^m \left\{ \UU_{1,i}(\hat\pi') \D_\V(\hat g', g_0;\nu) + \UU_{2,i}(\hat\pi',\hat F')\D_\V(\hat g', g_0;1) \right\}. 
\end{align*}
We first consider $\B_{31}$. 
By Proposition \ref{prop:W_nuf(A,Z)}, we have $\E(\B_{31}|\Oc') = 0$. 
Then by a similar proof as the above proof for showing $\B_{21} = o_p(n^{-1/2})$, we can show that $\var(\B_{31}) = o(m^{-1}) = o(n^{-1})$ and $\E(\B_{31}) = 0$. So by Chebyshev's inequality, $\B_{31} = o_p(m^{-1/2}) = o_p(n^{-1/2})$. 

Similarly, we can show that $\B_{32} = o_p(m^{-1/2}) = o_p(n^{-1/2})$.

We now consider $\B_{33}$. 
We have 
    \begin{align*}
        \E\left(|\B_{33}| \left| \Oc' \right.\right)
        &\leq \E\left[\left. \left|\UU_{1}(\hat\pi') \right| \cdot |\D_\V(\hat g',g_0;\nu)| \right| \Oc' \right] 
        + \E\left[\left. \left| \UU_{2,i}(\hat\pi',\hat F') \right| \cdot |\D_\V(\hat g',g_0;1)| \right| \Oc' \right] \\
        &\lesssim \E\left[ | \D_\V(\hat g',g_0;\nu) | \left| \Oc'\right. \right] + \E\left[ | \D_\V(\hat g',g_0;1) | \left| \Oc'\right. \right]. 
    \end{align*}
Recall from \eqref{eq:K_1_tilde}, \eqref{eq:K_2_tilde}, \eqref{eq:K_3_tilde} the definitions of $\tilde K_1$, $\tilde K_2$ and $\tilde K_3$;
and recall from \eqref{eq:K_1} \eqref{eq:K_2} \eqref{eq:K_3} the definitions of $K_1$, $K_2$, and $K_3$. 
By Proposition \ref{prop:D_W_bound} and Cauchy-Schwarz inequality, 
\begin{align*}
    \E\left(|\B_{33}| \left| \Oc' \right.\right)
    & \lesssim \|\hat F'-F_0\|_{\sup,2} \cdot\left\{ \|\hat G'-G_0\|_{\sup,2} + \|\hat S_D' - S_{D0}\|_{\sup,2} \right\}  \\
    &\qquad + \|\tilde K_1(\hat g',g_0;\nu)\|_1 + \|\tilde K_2(\hat g',g_0;\nu)\|_1 + \|\tilde K_3(\hat g',g_0;\nu)\|_1 \\
     &\qquad + \|\tilde K_1(\hat g',g_0;1)\|_1 + \|\tilde K_2(\hat g',g_0;1)\|_1 + \|\tilde K_3(\hat g',g_0;1)\|_1 \\
     & \lesssim \|\hat F'-F_0\|_{\sup,2} \cdot\left\{ \|\hat G'-G_0\|_{\sup,2} + \|\hat S_D' - S_{D0}\|_{\sup,2} \right\}  \\
    &\qquad + \|K_1(\hat g',g_0;\nu)\|_1 + \|K_2(\hat g',g_0;\nu)\|_1 + \|K_3(\hat g',g_0;\nu)\|_1 \\
    & = o_p(n^{-1/2}), 
\end{align*}
where the last equation holds by Assumption \ref{ass:prodrate}. 
By Assumption \ref{ass:strict_positivity}, $\B_{33}$ is bounded almost surely, so we have $\E(|\B_{33}|) = \E\{\E(|\B_{33}|\mid \Oc')\} = o(n^{-1/2})$. Therefore, by Markov's inequality, $\B_{33} = o_p(n^{-1/2})$. 

 Combining the above, we have \eqref{eq:ANproof_4} holds, and the asymptotic normality of $\hat\theta_{cf}$ follows with $\sigma^2 = \beta^2 \E\{U(\theta_{0};\pi_0,F_0,G_0,S_{D0})^2\}$.

For the variance estimator, we can show that it converges in probability to $\sigma^2$ in a similar way as  \citet[part (3) in the proof of Theorem S2 in the Supplementary Material]{wang2024doubly}.

\end{proof}

\clearpage

\section{CATE theory}\label{app:CATE_theory}

In {Section \ref{app:general_CATE_err_bound}}, we extend Theorem 1 in \citet{foster2023orthogonal} to a general error bound result for  a general loss function in the observed data, without the orthogonality or smoothness assumptions. 
Section \ref{sec:CATE_theory} contains the error bound and the oracle result for CATE estimation with the observed data loss function obtained by applying $\V$ to the LTRC-free data loss function \eqref{eq:weighted_square_loss}; the proofs are in Section \ref{app:CATE_proof}. 
Section \ref{app:instantiate_oracle_results} contains the proofs of the theorems for the ltrcR- and ltrcDR-learners in the main paper, which utilizes the error bound and convergence rate results in Section \ref{sec:CATE_theory}.

\subsection{Extension of \citet{foster2023orthogonal} Theorem 1} \label{app:general_CATE_err_bound}

With a little abuse of notation, in this subsection we use $\tilde\ell(\tau; h)$ to denote a generic loss function for $\tau$ {the parameter of interest (which does not have to be CATE)}, 
with nuisance parameter $h$. Let $h_0$ denote the true nuisance parameter,  
and $\Hc$  the parameter class for $h$.
{We consider  empirical risk minimization using $\tilde\ell(\tau,h)$ with sampling splitting, and denote $\hat\tau$ the resulting estimate}.


Let $\tilde k(h,h_0)$ and $\tilde \ksup(h,h_0)$ denote two pseudo-distance functions between  $h$ and $h_0$ satisfying $\tilde k(h,h_0)\geq 0$ and $\tilde \ksup(h,h_0)\geq 0$ for all $h,h_0\in \Hc$.
Let $\|\cdot\|_\T$ denote a prenorm in $\T$, which satisfies nonnegativity and $\|0\|_\T = 0$ but not necessarily the triangle inequality nor absolute homogeneity ($\|a\tau\|_\T = |a|\cdot \|\tau\|_{\T}$ for any $a\in\R$ and $\tau\in\T$).
Let $L(\tau,h) = \E\{\tilde\ell(\tau; h)\}$ denote the population risk. 

A function class $\mathcal{J}$ is {\it star-shaped} if for any $f\in\mathcal{J}$, $\alpha f\in\mathcal{J}$ for all $\alpha\in[0,1]$. It follows immediately  that a convex function class containing the origin is star-shaped. In addition, for any function class $\mathcal{J}$ and any $f \in\mathcal{J}$, the {\it stall hull} is 
\begin{align*}
    \sstar(\mathcal{J},f) = \{t\cdot f + (1-t)\cdot \tilde f: \text{ for all } \tilde f\in\mathcal{J}, \text{ and all } t\in[0,1]\}.
\end{align*}
By definition,  $\sstar(\mathcal{J},0)$ is star-shaped if $0\in\mathcal{J}$, and we will use the shorthand $\sstar(\mathcal{J}) = \sstar(\mathcal{J},0)$.

We consider the following assumptions.

\begin{appxassumption}\label{ass:obs_1-4}
      
    (i) (First order optimality) The truth $\tau_0$ satisfies the first order optimality condition: 
    $D_\tau L(\tau_0,h_0)[\tau-\tau_0] \geq 0$, for all $\tau\in\T$. 
    
    (ii) {(Higher order smoothness)} There exist constants $\betatwo>0$ and $r\in[0,1)$ such that for all $\tau\in\T$, $h\in\Hc$, 
    \begin{align}
        |D_\tau L(\tau_0, h)[\tau-\tau_0] - D_\tau L(\tau_0,h_0)[\tau-\tau_0]| \leq \betatwo \cdot \|\tau-\tau_0\|_\T^{1-r} \cdot \tilde k(h,h_0). \nonumber
    \end{align}
    
    (iii) (Relaxed strong convexity) There exist constants $\lambda_1,\kappa>0$ and $\lambda_2 \geq 0$ such that for all $\tau\in\T$, $\bar\tau\in\sstar(\T,\tau_0)$, 
    $h\in\Hc$,  
    \begin{align}
        D_\tau^2 L(\bar\tau,h)[\tau-\tau_0,\tau-\tau_0] \geq \{\lambda_1 - \lambda_2 \tilde \ksup(h,h_0)\} \|\tau-\tau_0\|_\T^2 - \kappa\cdot \tilde k(h,h_0)^{2/(1+r)}. \nonumber
    \end{align}
\end{appxassumption}

Assumptions \ref{ass:obs_1-4}~(i) is the same as its counterpart in Assumption {2} of \citet{foster2023orthogonal}.
Assumption \ref{ass:obs_1-4} (ii) does not require computing the second order directional derivative of the population risk with respect to the nuisance parameters as in \citet{foster2023orthogonal}, which can be complex for cases like our loss functions. 
The strong convexity assumption in \citet{foster2023orthogonal} is a special case of Assumption \ref{ass:obs_1-4}~(iii) with $\lambda_2 = 0$. 
Assumption \ref{ass:obs_1-4}~(iii) only requires strong convexity when $h$ is close enough to $h_0$ in terms of $\tilde \ksup(h,h_0)$.

Let $\tilde r_1(\Hc,\delta)$ denote the first-stage error bound for $\hat h$ such that $\tilde k(\hat h,h_0) \leq \tilde r_1(\Hc,\delta)$ with probability at least $1-\delta$. 
Let $\hat\tau_h$ denote the estimator from the empirical risk minimization with fixed nuisance parameter $h$.
Let $\tilde r_2(\T,\delta;h)$ denote the excess risk bound for the second-stage learning algorithm given any fixed nuisance parameter $h\in\Hc$ in the sense that $L(\hat\tau_h,h) - L(\tau_0,h) \leq \tilde r_2(\T,\delta;h)$ with probability at least $1-\delta$. 

Below is an extension of \citet{foster2023orthogonal} Theorem 1. 
\begin{appxlemma} \label{lem:CATE_rate_general}
    Under  Assumption \ref{ass:obs_1-4}, 
    the sample splitting algorithm produces an estimator $\hat\tau$ such that, with probability at least $1-\delta$, 
    \begin{align}
        \|\hat\tau - \tau_0\|_{\T}^2 
        & \leq \tilde\cc_1(n) \cdot \tilde r_1(\Hc, \delta/2)^{2/(1+r)} + \tilde\cc_2(n) \cdot \tilde r_2(\T,\delta/2;\hat h). \nonumber
    \end{align}
    where 
    $$
    \tilde\cc_1(n) \leq 2\left[\left\{\frac{2\betatwo}{\lambda_1 - \lambda_2 \tilde \ksup(\hat h,h_0)}\right\}^{2/(1+r)} + \frac{\kappa}{\lambda_1 - \lambda_2 \tilde \ksup(\hat h,h_0)}\right], 
    \quad \tilde\cc_2(n)\leq \frac{4}{\lambda_1 - \lambda_2\tilde \ksup(\hat h,h_0)}
    $$
    if $\lambda_1 - \lambda_2 \tilde\ksup(h,h_0)>0$. 
\end{appxlemma}

     The proof of Lemma \ref{lem:CATE_rate_general} is similar to the proof of Theorem 1 in \citet{foster2023orthogonal}, with the steps 
     that use Neyman orthogonality and higher order smoothness in their proof replaced by the condition in Assumption \ref{ass:obs_1-4} (ii), and with the steps that use strong convexity replaced by the condition in Assumption \ref{ass:obs_1-4} (iii). The details of the proof can be found in Section 2.9.4 of \citet{wang2025towards}. 

    We note that Assumption \ref{ass:obs_1-4}~(ii) can be shown to hold generally with $\tilde k(h,h_0) = \|h-h_0\|_{\Hc}$, so the nuisance estimation error by default has first order impact on the estimation error bound for $\tau$ according to Lemma \ref{lem:CATE_rate_general}.  
    When $\tilde\ell$ is Neyman orthogonal, 
    Assumption \ref{ass:obs_1-4}~(ii) holds with $\tilde k(h,h_0) = \|h-h_0\|_{\Hc}^2$ if $\tilde\ell$ also satisfies higher order smoothness. 
    Therefore in this case, the nuisance estimation error only has second order impact on the estimation error bound for $\tau$ {according to} Lemma \ref{lem:CATE_rate_general}. 
    Furthermore, if the loss function $\tilde\ell$ is also doubly robust, Assumption \ref{ass:obs_1-4}~(ii) may hold with $\tilde k(h,h_0)$ being {(bounded by)} the product errors of the nuisance parameters, like \eqref{eq:k_DR} for the ltrcDR-loss  in Section \ref{app:instantiate_oracle_results}, indicating that the nuisance estimation error only impact the estimation error bound for $\tau$ in Lemma \ref{lem:CATE_rate_general} through the product errors.

         Assumption \ref{ass:obs_1-4}~(ii) holds with $\tilde k(h,h_0) = \|h-h_0\|_{\Hc}^2$ when $\tilde\ell$ is Neyman orthogonal and satisfying the higher order smoothness in Assumption 3b) of \citet{foster2023orthogonal}. 
        If in addition $\tilde\ell$ is strongly convex at all $h\in\Hc$ (i.e., Assumption \ref{ass:obs_1-4} (iii) holds with $\lambda_2 = 0$), which is Assumption 4 of \citet{foster2023orthogonal},  Lemma \ref{lem:CATE_rate_general} implies the error bound in Theorem 1 of \citet{foster2023orthogonal}.

\subsection{Weighted squared loss for CATE}\label{sec:CATE_theory}

In the following we consider a class of weighted squared losses for CATE in the LTRC-free data:  
\begin{align}
    \ell(T^*,A^*,Z^*;\tau,\eta) = \omega(A^*,Z^*;\eta)\{\gammaone(A^*,Z^*;\eta)\nu(T^*) + \gammatwo(A^*,Z^*;\eta) - \tau(V^*)\}^2, \label{eq:weighted_square_loss}
\end{align}
and recall that $\eta$ denote the nuisance parameter involved.
It is straightforward to see that \eqref{eq:weighted_square_loss} encompasses both the R-loss and the DR-loss, in which $\eta = (\pi, F)$.
For the rest of this subsection,  we consider a generic $\eta$.

We now state an error bound and an oracle result for CATE estimation when the observed data loss function $\tilde\ell$ is obtained from $\ell$ in \eqref{eq:weighted_square_loss}
  by applying the LTRC-operator $\V$   
(after dropping a function of {observed data} $O$ that does not involve $\tau$): 
\begin{align}
  \tilde\ell(\tau; \eta,F,G,S_D) &= \V(\ell;F,G,S_D) \nonumber \\
    &= \V(1;F,G,S_D)\omega(A,Z;\eta) \left\{ \gammaone(A,Z;\eta) \cdot \frac{\V(\nu;F,G,S_D)}{\V(1;F,G,S_D)} +\gammatwo(A,Z;\eta) - \tau(V) \right\}^2. \label{eq:weighted_square_loss_obs}
\end{align}

\subsubsection{An error bound}\label{sec:CATE_oracle_result}

 Assumption \ref{ass:1-4} below can be verified for both the R-loss and the DR-loss (Appendix \ref{app:instantiate_oracle_results}), and ensures that Assumption \ref{ass:obs_1-4} holds for the observed data loss function $\tilde\ell$ in \eqref{eq:weighted_square_loss_obs}.
 
\begin{appxassumption}\label{ass:1-4}
    (i) (First order optimality) The truth $\tau_0$ satisfies the first order optimality condition in the LTRC-free data: $D_\tau \E\{\ell(T^*,A^*,Z^*;\tau_0,\eta_0)\}[\tau-\tau_0] \geq 0$, for all $\tau\in\T$.\\
    (ii) (Boundedness of $\omega$) There exist constants $\consta, \constb  >0$ such that $\consta \leq \omega(A^*,Z^*;\eta) \leq \constb $ a.s. for all $\eta\in\mathcal{N}$.\\
    (iii) (Boundedness of $\gammaone$, $\gammatwo$) There exist $\constc,\constd >0$ such that $|\gammaone(A^*,Z^*;\eta)| \leq \constc$ a.s. and $|\gammatwo(A^*,Z^*;\eta)| \leq \constd$ a.s. for all $\eta\in\mathcal{N}$.
\end{appxassumption}

Before stating the  error bound  we need the following {\it pseudo-distance} between $\eta$ and $\eta_0$ in the  LTRC-free data:
\begin{align}
  \kstar(\eta,\eta_0) =  \E\{|\Dlstarz(\eta) - \Dlstarz(\eta_0)|^2\}^{1/2}, 
  \label{eq:bdd_by_k}
\end{align}
where 
\begin{align}
    \Dlstarz(\eta) & = \E[\omega(A^*,Z^*;\eta)\{\gammaone(A^*,Z^*;\eta)\nu(T^*) + \gammatwo(A^*,Z^*;\eta) - \tau_0(V^*)\} \mid Z^*]. \label{eq:Dlstar_z} 
\end{align}
Note that $\Dlstarz(\eta) $ is a quantity related to the directional derivative of the population risk in the LTRC-free data: 
\begin{align*}
    &\quad D_\tau \E\{ \ell(T^*,A^*,Z^*;\tau_0,\eta)\}[\tau- \tau_0] \\
    & = -2 \E[\omega(A^*,Z^*;\eta)\{\gammaone(A^*,Z^*;\eta)\nu(T^*) + \gammatwo(A^*,Z^*;\eta) - \tau_0(V^*)\}  \cdot \{\tau(V^*) - \tau_0(V^*)\}]\\
    & = -2 \E[ \Dlstarz(\eta) \cdot \{\tau(V^*) - \tau_0(V^*)\}],
\end{align*}
where the last equation is
 by the tower property of expectations (first conditioning on $Z^*$) and the fact that $V^*\subseteq Z^*$. 
We will instantiate $\kstar(\eta,\eta_0)$ for the R-loss and DR-loss in Section \ref{app:instantiate_oracle_results} and provide its upper bound. 

Recall that  $g = (F,G,S_D)$, and
$K_1(g,g_0)$, $K_2(g,g_0)$ and $K_3(g,g_0)$ are the same as in Section \ref{sec:ATE}. 
We consider the following measure of  difference between the nuisance parameters and their truth:
\begin{align}
    k((\eta,g),(\eta_0,g_0)) 
    & = \kstar(\eta,\eta_0) + \|F-F_0\|_{\sup,4} \cdot\left\{ \|G-G_0\|_{\sup,4} + \|S_D - S_{D0}\|_{\sup,4} \right\} \nonumber \\
    &\qquad + \|K_1(g,g_0)\|_2 + \|K_2(g,g_0)\|_2 + \| K_3(g,g_0)\|_2. \label{eq:k_expression}
\end{align} 
Note that \eqref{eq:k_expression} shares the same product and integral product terms  as that in Assumption \ref{ass:prodrate}, albeit with different norms. 

Let $r_1(\Gc, \delta)$ denote the first-stage nuisance error bound  such that 
with probability at least $1-\delta$,
\begin{align}
    k((\hat\eta,\hat g), (\eta_0,g_0)) \leq r_1(\Gc, \delta).
    \label{eq:k}
\end{align}
Note that \eqref{eq:k} is specific to causal LTRC problems in that it concerns the product and integral product nuisance error bound 
\citep{ying2023cautionary, wang2024doubly, luo2023doubly}. 

Let $L(\tau,\eta,g) =  \E\{\tilde\ell(\tau;\eta,g)\}$ denote the population risk, and let 
$r_2(\T,\delta;\eta,g)$ denote the excess risk bound of the second-stage learning algorithm for fixed  $(\eta,g)$ 
 such that with probability at least $1-\delta$, 
\begin{align}
    L(\hat\tau_{\eta,g},\eta,g) - L(\tau_0,\eta,g) \leq r_2(\T,\delta;\eta,g). \label{eq:second_stage_rate}
\end{align}
The excess risk bound is commonly considered in the regression literature for measuring the quality of the estimated regression function \citep{wainwright2019high}.

\begin{proposition}\label{thm:CATE_error_bdd}
    Under Assumptions \ref{ass:consistency}
    - \ref{ass:strict_positivity}, \ref{k0_converge_to_0} and  
     \ref{ass:1-4},
     suppose that 
     all the nuisance parameters in $\Gc$ satisfy Assumption \ref{ass:strict_positivity} (strict positivity). 
If $\hat\tau$ is obtained from empirical risk minimization using \eqref{eq:weighted_square_loss_obs} under sample splitting,
     then      with probability at least $1-\delta$, 
    \begin{align*}
        \|\hat\tau - \tau_0\|_2^2 \leq \cc_1(n) \cdot r_1(\Gc,\delta/2)^2 + \cc_2(n) \cdot r_2(\T,\delta/2;\hat\eta,\hat g),
    \end{align*}
    where $\cc_1(n), \cc_2(n)$ are bounded with probability arbitrarily close to 1 for large enough $n$.
\end{proposition}
Proposition \ref{thm:CATE_error_bdd} gives the error bound for $\hat\tau$ in terms of  the excess risk bound for the second-stage learning algorithm. 
The following corollary gives it in terms of the second-stage estimation error. 
Let $r_3(\T,\delta;\eta,g)$ denote the estimation error bound of the second stage learning algorithm for CATE in the sense that for any given $(\eta,g)\in\Gc$, the second stage learning algorithm outputs an estimator $\hat\tau_{\eta,g}$ such that with probability at least $1-\delta$, 
\begin{align*}
    \|\hat\tau_{\eta,g} - \tau_0\|_2 \leq r_3(\T,\delta;\eta,g).
\end{align*}

\begin{corollary}\label{cor:CATE_err_bound}
    Under Assumptions \ref{ass:consistency}
    - \ref{ass:strict_positivity}, \ref{k0_converge_to_0} and 
     \ref{ass:1-4},
     suppose in addition that the condition in Assumption \ref{ass:1-4} (i) holds with ``='' (which is commonly the case for weighted squared loss),  
     and that all the nuisance parameters in $\Gc$ satisfy Assumption \ref{ass:strict_positivity} (strict positivity).
If $\hat\tau$ is obtained from empirical risk minimization using \eqref{eq:weighted_square_loss_obs} under sample splitting,
     then    with probability at least $1-\delta$, 
    \begin{align*}
        \|\hat\tau - \tau_0\|_2^2 \leq \cc_3(n) \cdot r_1(\Gc,\delta/4)^2 + \cc_4(n) \cdot r_3(\T,\delta/4;\hat\eta,\hat g)^2,
    \end{align*}
    where $\cc_3(n), \cc_4(n)$ are bounded with probability arbitrarily close to 1 for large enough $n$.
\end{corollary}

\subsubsection{An oracle result}\label{sec:oracle_bound}

We note that the error bound $r_2(\T,\cdot;\hat\eta,\hat g)$  involved in the above Proposition \ref{thm:CATE_error_bdd} (and $r_3(\T,\cdot;\hat\eta,\hat g)$ in Corollary \ref{cor:CATE_err_bound}) is different from and the corresponding oracle bound $r_2(\T,\cdot;\eta_0,g_0)$. 
As discussed in \citet{foster2023orthogonal}, for many learning algorithms, the difference $r_2(\T,\cdot;\hat\eta,\hat g) - r_2(\T,\cdot;\eta_0,g_0)$ can be shown to be absorbed into $ r_1(\Gc,\cdot)^2$.
In particular,  \citet{foster2023orthogonal}  further provided the error bound results in terms of localized Rademacher complexity of the second-stage CATE function class when empirical risk minimization is used. As mentioned before, our observed data loss function is not strongly convex with respect to $\tau$ for all $(\eta,g)\in \Gc$, so once again the results in \citet{foster2023orthogonal} do not apply immediately. 
In the following, Proposition \ref{thm:CATE_error_rate_RC}  extends \citet{foster2023orthogonal} to our case.

Before stating the theorem, we first review the localized Rademacher complexity
\citep{bartlett2005local,koltchinskii2000rademacher,wainwright2019high}
and the critical radius of a function class. For a class $\mathcal{J}$ of one-dimensional functions of the covariates $V$, the localized Rademacher complexity is defined as
\begin{align*}
    \Rc_n(\mathcal{J},\crds) = \E\left\{ \sup_{f\in\mathcal{J}:\|f\|_2\leq\crds} \left|\frac{1}{n}\sum_{i=1}^n \epsilon_if(V_i) \right| \right\},
\end{align*}
where $\|f\|_2 = \E\{|f(V)|^2\}^{1/2}$ denote the $L_2$ norm, $\epsilon_1, ... ,\epsilon_n$ are independent Rademacher random variables, i.e., taking values 1 and -1 with probabilities 1/2 each, and the expectation is taken with respect to both the $\epsilon_i$'s and the $V_i$'s.
{$ \Rc_n(\mathcal{J},\crds) $ is the expected maximum correlation between $ f(V_1), ..., f(V_n)$ and the noise vector $\epsilon_1,...,\epsilon_n $; for a large function class it is more likely to find a function $f$ with a high correlation with a random noise vector \citep{wainwright2019high}. }

The critical radius concerns star-shaped function classes, which is defined in Appendix \ref{app:general_CATE_err_bound}.

Denote $\T - \tau_0 = \{\tau-\tau_0:\tau\in\T\}$.
The smallest $ \rho>0$ such that
$\Rc_n(\sstar(\T-\tau_0),\crds) \leq \crds^2$
is called the {\it critical radius} of $\sstar(\T-\tau_0)$; it is where the growth of $\Rc_n(\sstar(\T-\tau_0),\crds)$ as a function of $\rho$ becomes slower than quadratic. 
 The existence of such a $\rho$ is guaranteed because $\sstar(\T-\tau_0)$ is star-shaped; see \citet[Lemma 13.6]{wainwright2019high} for  detailed reasoning. In fact, if $\crds_0 > 0$ is a solution to the above inequality, so is any $\crds > \crds_0$; and $\rho^2$ on the right hand side of the inequality may be scaled by a constant $R$, as in the theorem below.
The critical radius reflects the complexity of the function class, 
with smaller function classes having smaller critical radii.
For example, parametric classes typically have critical radius of order $n^{-1/2}$,  
while the reproducing kernel Hilbert spaces with Gaussian kernels has critical radius $\lesssim \{\log(n+1)/n\}^{1/2}$ \citep[Chapters 13 and 14]{wainwright2019high}.

\begin{proposition}\label{thm:CATE_error_rate_RC}
    Under Assumptions \ref{ass:consistency}
    - \ref{ass:strict_positivity}, \ref{k0_converge_to_0}  and 
     \ref{ass:1-4},
     suppose that 
     $R = \max\{\sup_{\tau\in\T, v} |\tau(v)|, ~1\} < \infty$, and that all the nuisance parameters in $\Gc$ satisfy Assumption \ref{ass:strict_positivity} (strict positivity). 
    Let $\crds_n^2\gtrsim R^2 \log(\log(n))/n$
    be any solution to the inequality 
    \begin{align}
        \Rc_n(\sstar(\T-\tau_0),\crds) \leq \crds^2/R, \label{eq:ineq_delta_n}
    \end{align}
    where ``$\gtrsim$'' denotes ``$\geq$'' up to a universal constant factor.
If $\hat\tau$ is obtained from empirical risk minimization using  \eqref{eq:weighted_square_loss_obs} under sample splitting,
     then   with probability at least $1-\delta$, 
    \begin{align*}
        \|\hat\tau - \tau_0\|_2^2 \leq \cc_{5}(n) \cdot k((\hat\eta,\hat g), (\eta_0,g_0))^2 + \cc_{6}(n)\cdot \left\{\frac{\crds_n^2}{R^2} + \frac{\log(1/\delta)}{n} \right\},
    \end{align*}
     where $\cc_5(n), \cc_6(n)$ are bounded with probability arbitrarily close to 1 for large enough $n$.
\end{proposition}

According to \citet{wainwright2019high}, 
the oracle rate of $\|\hat\tau - \tau_0\|_2^2$ is $\crds_n^2$ for $\hat\tau$ obtained from empirical process minimization as in here.
From Proposition \ref{thm:CATE_error_rate_RC} we see that 
 the error bound for $\|\hat\tau - \tau_0\|_2^2$ is dominated by the oracle rate $\crds_n^2$ as long as $k((\hat\eta,\hat g), (\eta_0,g_0))$ converges to zero fast enough. 
Note that $\crds_n^2$ is no faster than $O_p(n^{-1})$,  so it suffice to have the product error $k((\hat\eta,\hat g), (\eta_0,g_0)) = o_p(n^{-1/2})$ to achieve the oracle error rate. In fact, we have the following Corollary \ref{cor:CATE_Op_rate} for the convergence rate of $\hat\tau$, with its proof in the Appendix.

\begin{corollary}\label{cor:CATE_Op_rate}
    Under the same conditions as in Proposition \ref{thm:CATE_error_rate_RC}, 
    suppose that $k((\hat\eta,\hat g), (\eta_0,g_0)) = O_p(a_n)$.  Let $\crds_n^2\gtrsim R^2 \log(\log(n))/n$ be any solution to the inequality \eqref{eq:ineq_delta_n}. Then 
    \begin{align*}
    \|\hat\tau - \tau_0\|_2 = O_p\left( a_n + \crds_n \right). 
    \end{align*}
\end{corollary}

\subsection{Proofs of the error bound and oracle result for weighted squared loss} \label{app:CATE_proof}

\subsubsection{Proofs of Proposition \ref{thm:CATE_error_bdd} and Corollary \ref{cor:CATE_err_bound}  }\label{app:proof_thm3}

\begin{proof}[Proof of Proposition \ref{thm:CATE_error_bdd}]

Recall that $g = (F,G,S_D)$, and $\D_\V$ from \eqref{eq:D_V}.
Let $\|\D_\V(g,g_0;1)\|_{\infty}$ denote the $L_\infty$ norm with respect to the random variables in $O$.
We will verify that the conditions in Assumption \ref{ass:obs_1-4} are satisfied with $h = (\eta,g)$,
$\|\cdot\|_{\T} = \|\cdot\|_2$, $r = 0$, 
$\tilde k(h,h_0) = k((\eta,g),(\eta_0,g_0))$ as in \eqref{eq:k_expression}, and 
$\tilde \ksup(h,h_0) = \|\D_\V(\hat g,g_0;1)\|_{\infty}$.
Then the conclusion follows by Lemma \ref{lem:CATE_rate_general}.

Recall that $L(\tau,\eta, g) =  \E\{\tilde\ell(\tau; \eta, g)\}$ denote the population risk where $\tilde\ell$ is given in \eqref{eq:weighted_square_loss_obs}.

By Lemma \ref{lem:DR_V} and Assumption \ref{ass:1-4} (i), we have that Assumption \ref{ass:obs_1-4} (i) holds. 

In the following, we will verify that Assumption \ref{ass:obs_1-4} (ii) and (iii) hold.

1)
We first verify Assumption \ref{ass:obs_1-4} (ii). 
We consider the decomposition 
\begin{align}
     D_\tau L(\tau_0, \eta,g)[\tau-\tau_0] - D_\tau L(\tau_0,\eta_0,g_0)[\tau-\tau_0] 
     & = \A_1 + \A_2, \label{eq:thm3_proof_8}
\end{align}
where 
\begin{align*}
    \A_1 & =  D_\tau L(\tau_0, \eta,g_0)[\tau-\tau_0] - D_\tau L(\tau_0,\eta_0,g_0)[\tau-\tau_0], \\
    \A_2 & =  D_\tau L(\tau_0, \eta,g)[\tau-\tau_0] - D_\tau L(\tau_0,\eta,g_0)[\tau-\tau_0]. 
\end{align*}
In the following, we will bound $|\A_1|$ and $|\A_2|$, respectively. 

We first show below the connection \eqref{eq:proof_thm3_Dequal} between $D_\tau \E\{\tilde\ell \}$  in the observed data  and $D_\tau \E\{\ell \}$  in the LTRC-free data. 
From \eqref{eq:weighted_square_loss_obs}, we have 
\begin{align}
    &\quad D_\tau L(\tau_0,\eta,g)[\tau-\tau_0] \nonumber \\
    & = D_\tau \E\{\tilde\ell(\tau_0; \eta,g) \}[\tau-\tau_0] \nonumber \\
    & = -2 \E\left[\V(1;g)\omega(A,Z;\eta) \left\{ \gammaone(A,Z;\eta) \cdot \frac{\V(\nu;g)}{\V(1;g)} +\gammatwo(A,Z;\eta) - \tau_0(V) \right\}  \{\tau(V)-\tau_0(V)\} \right] \nonumber \\
    & = - 2 \E\left[ \BB(\eta,g) \{\tau(V)-\tau_0(V)\} \right], \label{eq:thm3_proof_1}
\end{align}
where 
\begin{align}
    \BB(\eta,g) & = \omega(A,Z;\eta) \gammaone(A,Z;\eta)  \V(\nu; g) + \V(1; g) \omega(A,Z;\eta)\{\gammatwo(A,Z;\eta)- \tau_0(V)\}.  \label{eq:thm3_proof_B(eta,g)}
\end{align}
From \eqref{eq:weighted_square_loss}, we have 
\begin{align*}
    D_\tau\E\{\ell(T^*,A^*,Z^*;\tau,\eta)\}[\tau - \tau_0] & = -2 \E\left[ \Dlstar(\eta)\{\tau(V)-\tau_0(V)\} \right]
\end{align*}
where 
\begin{align}
    \Dlstar(\eta) = \omega(A^*,Z^*;\eta)\{\gammaone(A^*,Z^*;\eta)\nu(T^*) + \gammatwo(A^*,Z^*;\eta) - \tau_0(V^*)\}. \label{eq:Dlstar(eta)}
\end{align}
By Proposition \ref{prop:W_nuf(A,Z)}, 
\begin{align*}
    \E\{B(\eta,g_0)\varphi(Z)\} 
    & = \beta^{-1} \E\{\Dlstar(\eta)\varphi(Z)\},
\end{align*}
so
\begin{align}
    D_\tau L(\tau_0, \eta,g_0)[\tau-\tau_0]
    & = \beta^{-1} \cdot D_\tau\E\{\ell(T^*,A^*,Z^*;\tau,\eta)\}[\tau - \tau_0]. \label{eq:proof_thm3_Dequal}
\end{align}

We now bound $|\A_1|$. 
Recall Recall $\Dlstarz(\eta)$ from \eqref{eq:Dlstar_z} and $\Dlstar(\eta)$ from \eqref{eq:Dlstar(eta)}. 
We have $\Dlstarz(\eta) = \E\{\Dlstar(\eta) \mid Z^*\}$. 
By tower property of expectations, we have 
\begin{align*}
    D_\tau\E\{\ell(T^*,A^*,Z^*;\tau,\eta)\}[\tau - \tau_0] & = -2 \E\left[ \Dlstar(\eta)\{\tau(V)-\tau_0(V)\} \right] \\
    & = -2 \E\left[ \E\{\Dlstar(\eta)\mid Z^*\} \{\tau(V)-\tau_0(V)\} \right] \\
    & = -2 \E\left[ \Dlstarz(\eta)\{\tau(V)-\tau_0(V)\} \right].
\end{align*} 
This together with \eqref{eq:proof_thm3_Dequal} implies that 
\begin{align}
    |\A_1| 
    & = 2 \beta^{-1} \cdot \left| \E[ \{\Dlstarz(\eta) - \Dlstarz(\eta_0)\} \{\tau(V^*) - \tau_0(V^*)\} ]  \right| \nonumber \\
    & \leq  2\beta^{-1} \E\left[ |\Dlstarz(\eta) - \Dlstarz(\eta_0)| \cdot |\tau(V^*) - \tau_0(V^*)| \right] \nonumber \\
     &\leq 2\beta^{-1} \E\left\{|\Dlstarz(\eta) - \Dlstarz(\eta_0)|^2\right\}^{1/2} \cdot \E\{|\tau(V^*) - \tau_0(V^*)|^2\}^{1/2} \label{eq:CATE_rate_proof_DL*_1} \\
    &\lesssim \kstar(\eta,\eta_0) \cdot \|\tau - \tau_0\|_2, \label{eq:CATE_rate_proof_DL*_2}
\end{align}
where \eqref{eq:CATE_rate_proof_DL*_1} is by Cauchy-Schwarz inequality, and \eqref{eq:CATE_rate_proof_DL*_2} is by \eqref{eq:bdd_by_k} and Proposition \ref{prop:equiv_norms}.

We now bound $|\A_2|$. 
We consider the decomposition: 
\begin{align}
    \A_2 = \A_{21} + \A_{22}, \label{eq:thm3_proof_3}
\end{align}
where 
\begin{align*}
    \A_{21} & = \E[ \{\V(\nu;g) - \V(\nu;g_0)\} \BBB_1(\eta) ], \\
    \A_{22} & = \E[ \{\V(1; g) - \V(1; g_0)\} \BBB_2(\eta) ], \\
    \BBB_1(\eta) & = \omega(A,Z;\eta) \gammaone(A,Z;\eta)\{\tau(V)-\tau_0(V)\}, \\
    \BBB_2(\eta) & = \omega(A,Z;\eta)\{\gammatwo(A,Z;\eta)- \tau_0(V)\}\{\tau(V)-\tau_0(V)\}.
\end{align*}
So we have $|\A_2| \leq |\A_{21}| + |\A_{22}|$.
In the following, we will bound $|\A_{21}|$ and $|\A_{22}|$, respectively.
Recall $\D_\V(g,g_0;\nu)$ from \eqref{eq:D_V}.
By Proposition \ref{prop:D_W}, we have 
\begin{align*}
    \A_{21} & = \E[ \D_\V(g,g_0;\nu) \BBB_1(\eta)] .
\end{align*}
By Assumption \ref{ass:1-4} (ii) (iii), 
\begin{align*}
    |\BBB_1(\eta)| &\leq \constb \constc \cdot |\tau(V)-\tau_0(V)|.
\end{align*}
So 
\begin{align*}
    |\A_{21}|
    & \leq \constb \constc \cdot  \E\{ |\D_\V(g,g_0;\nu)| \cdot |\tau(V)-\tau_0(V)| \} \nonumber \\
    &\leq \constb \constc \cdot \|\D_\V(g,g_0;\nu)\|_2 \cdot \|\tau-\tau_0\|_2,\\
\end{align*}
where the last inequality holds by Cauchy–Schwarz inequality. 
From Proposition \ref{prop:D_W_bound},
\begin{align}
        \|\D_\V(g,g_0;\nu)\|_2
        &\lesssim \|F-F_0\|_{\sup,4} \cdot \|G-G_0\|_{\sup,4} + \|F-F_0\|_{\sup,4} \cdot \|S_D - S_{D0}\|_{\sup,4}  \nonumber \\
        &\quad + \|\tilde K_1(g,g_0;\nu)\|_2 + \|\tilde K_2(g,g_0;\nu)\|_2 + \|\tilde K_3(g,g_0;\nu)\|_2.
\end{align}
Therefore,
\begin{align}
    |\A_{21}| 
    & 
    \lesssim 
    \Big[ \|F-F_0\|_{\sup,4} \cdot\left\{ \|G-G_0\|_{\sup,4} + \|S_D - S_{D0}\|_{\sup,4} \right\}    \nonumber \\
        &\qquad  \left. + \|\tilde K_1(g,g_0;\nu)\|_2 + \|\tilde K_2(g,g_0;\nu)\|_2 + \|\tilde K_3(g,g_0;\nu)\|_2 \right] \cdot \|\tau-\tau_0\|_2 . \label{eq:thm3_proof_4}
\end{align}
Likewise, we can show that  
\begin{align}
    |\A_{22}| & \lesssim 
    \Big[ \|F-F_0\|_{\sup,4} \cdot\left\{ \|G-G_0\|_{\sup,4} + \|S_D - S_{D0}\|_{\sup,4} \right\}  \nonumber \\
    &\qquad \left. + \|\tilde K_1(g,g_0;1)\|_2 + \|\tilde K_2(g,g_0;1)\|_2 + \|\tilde K_3(g,g_0;1)\|_2 \right] \cdot \|\tau-\tau_0\|_2 . \label{eq:thm3_proof_5}
\end{align}
Recall the definitions of $K_1$, $K_2$, and $K_3$ from \eqref{eq:K_1} \eqref{eq:K_2} \eqref{eq:K_3}. Combining 
\eqref{eq:thm3_proof_4} \eqref{eq:thm3_proof_5}, we have 
\begin{align}
    |\A_{2}| 
    &\leq |\A_{21}| + |\A_{22}| \nonumber \\
    & \lesssim \Big[ \|F-F_0\|_{\sup,4} \cdot\left\{ \|G-G_0\|_{\sup,4} + \|S_D - S_{D0}\|_{\sup,4} \right\}  \nonumber \\
    &\qquad  + \|K_1(g,g_0)\|_2 + \|K_2(g,g_0)\|_2 + \|K_3(g,g_0)\|_2 \Big] \cdot \|\tau-\tau_0\|_2. \label{eq:thm3_proof_6}
\end{align}

Recall $k((\eta,g), (\eta_0,g_0))$ from \eqref{eq:k_expression}. 
Combining  \eqref{eq:thm3_proof_8} \eqref{eq:CATE_rate_proof_DL*_2} \eqref{eq:thm3_proof_6}, we have 
\begin{align}
    \big|D_\tau L(\tau_0, \eta,g)[\tau-\tau_0] - D_\tau L(\tau_0,\eta_0,g_0)[\tau-\tau_0] \big| 
    & \leq |\A_{1}| + |\A_{2}|  \nonumber \\
    &\lesssim  k((\eta,g),(\eta_0,g_0)) \cdot \|\tau-\tau_0\|_2.  \label{eq:thm3_proof_Btau}
\end{align}
Therefore, Assumption \ref{ass:obs_1-4} (ii) holds with $r = 0$.

2)
Again recall $\D_\V$ from \eqref{eq:D_V}, and 
 that $\|\D_\V(g,g_0;1)\|_{\infty}$ denote the $L_\infty$ norm with respect to the random variables in $O$. Note that $\|\D_\V(\hat g,g_0;1)\|_{\infty}$ is random due to $\hat g$.

We now verify that Assumption \ref{ass:obs_1-4} (iii) holds with 
$\kappa = 0$, 
$\tilde \ksup(h,h_0) = \|\D_\V(\hat g,g_0;1)\|_{\infty}$, $\|\cdot\|_{\T} = \|\cdot\|_2$, and $r = 0$. 

We have 
\begin{align}
    D_\tau^2 L(\bar\tau,\eta,g)[\tau-\tau_0,\tau-\tau_0] 
    & = D_\tau^2 \E\{\tilde\ell(\bar\tau; \eta,g)\}[\tau-\tau_0,\tau-\tau_0] \\
    &= 2\E\left[ \V(1;g)\omega(A,Z;\eta) \{\tau(V)- \tau_0(V)\}^2 \right] \nonumber \\
    & \geq 2\Ec_1 - 2|\Ec_2|, \label{eq:CATE_rate_proof_1}
\end{align}
where 
\begin{align}
    \Ec_1 &= \E\left[ \V(1;g_0)\omega(A,Z;\eta) \{\tau(V)- \tau_0(V)\}^2 \right], \nonumber \\
    \Ec_2 & = \E\left[ \{\V(1;g) - \V(1;g_0)\} \omega(A,Z;\eta) \{\tau(V)- \tau_0(V)\}^2 \right]. \nonumber
\end{align}
We first consider $\Ec_1$. 
By Proposition \ref{prop:W_nuf(A,Z)}, 
\begin{align}
    \Ec_1 
    & = \beta^{-1} \E\left[ \omega(A^*,Z^*;\eta) \{\tau(V^*)- \tau_0(V^*)\}^2 \right] 
    \geq \beta^{-1}\consta \E\left[|\tau(V^*)- \tau_0(V^*)|^2 \right],  \nonumber
\end{align}
where the last inequality holds by Assumption \ref{ass:1-4} (ii).
By Proposition \ref{prop:equiv_norms}, 
$\E[\{\tau(V^*) - \tau_0(V^*)\}^2] \geq \beta\E[\{\tau(V) - \tau_0(V)\}^2] = \beta\|\tau-\tau_0\|_2^2$, so
\begin{align}
    \Ec_1 
    &\geq \consta \|\tau-\tau_0\|_2^2. \label{eq:CATE_rate_proof_2}
\end{align}
We now consider $\Ec_2$. Recall from \eqref{eq:D_V} the definition of $\D_\V$. By Proposition \ref{prop:D_W},
\begin{align}
    |\Ec_2| 
    & = \left| \E\left[ \D_\V(g,g_0;1) \omega(A,Z;\eta) \{\tau(V)- \tau_0(V)\}^2 \right] \right| \nonumber \\
    & \leq \constb \cdot \E\left[ |\D_\V(g,g_0;1)|\cdot |\tau(V)- \tau_0(V)|^2 \right] \label{eq:CATE_rate_proof_5} \\
    &\leq \constb \cdot \|\D_\V(g,g_0;1)\|_{\infty} \cdot \|\tau-\tau_0\|_2^2 \label{eq:CATE_rate_proof_8}
\end{align}
where \eqref{eq:CATE_rate_proof_5} holds by Assumption \ref{ass:1-4} (ii), and \eqref{eq:CATE_rate_proof_8} holds by the definitions of $\|\D_\V(g,g_0;1)\|_{\infty}$ and $\|\tau-\tau_0\|_2$. 
Combining the above, Assumption \ref{ass:obs_1-4} (iii) is satisfied with $\lambda_1 = 2\consta$, $\lambda_2 = 2\constb$, $\kappa = 0$, and $\tilde \ksup(h,h_0) = \|\D_\V(\hat g,g_0;1)\|_{\infty}$.

Therefore, by Lemma \ref{lem:CATE_rate_general}, the error bound in Proposition \ref{thm:CATE_error_bdd} follows with $\cc_1(n)$ and $\cc_2(n)$ satisfying 
\begin{align}
    \cc_1(n) \lesssim \frac{1}{\{\consta - \constb \cdot\|\D_\V(\hat g,g_0;1)\|_{\infty} \}^2}, 
    \quad 
    \cc_2(n) \lesssim \frac{1}{\consta - \constb \cdot \|\D_\V(\hat g,g_0;1)\|_{\infty}} \label{eq:c1(n)_c2(n)}
\end{align}
if $\|\D_\V(\hat g,g_0;1)\|_{\infty} < \consta \constb^{-1}$; and $\cc_1(n) = \cc_2(n) = \infty$ otherwise. 

Finally, we show that $\cc_1(n)$ and $\cc_2(n)$ are bounded with probability arbitrarily close to 1 for large enough $n$ if Assumption \ref{k0_converge_to_0} holds.
By Proposition \ref{prop:D_W_bound}, 
\begin{align}
    \|\D_\V(\hat g,g_0;1)\|_{\infty} & \lesssim \|\hat F-F_0\|_{\infty} \cdot\left\{ \|\hat G-G_0\|_{\infty} + \|\hat S_D - S_{D0}\|_{\infty} \right\} \nonumber \\
    &\qquad + \|\tilde K_1(\hat g,g_0;1)\|_{\infty} + \|\tilde K_2(\hat g,g_0;1)\|_{\infty} + \|\tilde K_3(\hat g,g_0;1)\|_{\infty}. \label{eq:\ksup}
\end{align}
This together with Assumption \ref{k0_converge_to_0} implies that $\|\D_\V(\hat g,g_0;1)\|_{\infty} \convp 0$ as $n\to\infty$. Therefore, $\cc_1(n)$ and $\cc_2(n)$ are bounded with probability arbitrarily close to 1 for large enough $n$.

\end{proof}

\vspace{3em}
To prove Corollary \ref{cor:CATE_err_bound}, we first show the following Proposition \ref{prop:CATE_error_bdd}, which connects the second stage estimation error  with  the excess risk. 
Corollary \ref{cor:CATE_err_bound} then immediately follows from Proposition \ref{thm:CATE_error_bdd} with $\eta = (\pi,F)$.
Recall the definition of $k((\eta,g),(\eta_0,g_0))$ in \eqref{eq:k_expression}. 

\begin{proposition}
    \label{prop:CATE_error_bdd}
    Under Assumptions \ref{ass:consistency}
      - \ref{ass:strict_positivity} and
     \ref{ass:1-4},
     in addition suppose that Assumption \ref{ass:1-4} (i) holds with ``=".
    Suppose 
    that all $(\eta,g)\in \Gc$ satisfy Assumption \ref{ass:strict_positivity} (strict positivity). Then 
    \begin{align*}
        L(\hat\tau_{\eta,g},\eta,g) - L(\tau_0,\eta,g) 
        \lesssim \|\hat\tau_{\eta,g} - \tau_0\|_2^2 + k((\eta,g),(\eta_0,g_0))^2
    \end{align*}
    for all $(\eta,g) \in \Gc$.
\end{proposition}

\begin{proof}[Proof of Proposition \ref{prop:CATE_error_bdd}]

    In the following we will show that for any $\tau\in\T$, 
    \begin{align}
        \left|L(\tau,\eta,g) - L(\tau_0,\eta,g) \right|
        \lesssim \|\tau - \tau_0\|_2^2 + k((\eta,g),(\eta_0,g_0))^2.  \label{eq:cor2_proof_1}
    \end{align}
    Then the conclusion follows by replacing $\tau$ by $\hat\tau_{\eta,g}$.
    
Recall that $g = (F,G,S_D)$.
The observed data loss function constructed from the weighted squared loss in the LTRC-free data is:
\begin{align}
    \tilde\ell(\tau; \eta,g) 
     & = \V(1;g)\omega(A,Z;\eta) \left\{ \gammaone(A,Z;\eta) \cdot \frac{\V(\nu;g)}{\V(1;g)} +\gammatwo(A,Z;\eta) - \tau(V) \right\}^2. \nonumber
\end{align}
Recall from \eqref{eq:thm3_proof_B(eta,g)} that we denote 
\begin{align}
    \BB(\eta,g) & = \omega(A,Z;\eta) \gammaone(A,Z;\eta)  \V(\nu; g) + \V(1; g) \omega(A,Z;\eta)\{\gammatwo(A,Z;\eta)- \tau_0(V)\}.  \nonumber 
\end{align}
With some algebra, we have 
\begin{align*}
    \tilde\ell(\tau; \eta, g) - \tilde\ell(\tau_0;\eta, g) 
    & = \V(1;g)\omega(A,Z;\eta)\left\{\tau(V) - \tau_0(V) \right\}^2 
    -2 \BB(\eta,g)\left\{\tau(V) - \tau_0(V) \right\}. 
\end{align*}
So 
\begin{align}
    L(\tau,\pi, g) - L(\tau_0,\pi, g) 
    &=\E\{\tilde\ell(\tau; \pi, g)\} - \E\{\tilde\ell(\tau_0; \pi, g)\} 
    = \B_{1} - 2\B_2 - 2\B_3,  \label{eq:prop5_proof_1}
\end{align}
where
\begin{align*}
    \B_1 & = \E\left[ \V(1;g)\omega(A,Z;\eta)\left\{\tau(V) - \tau_0(V) \right\}^2 \right], \\
    \B_2 & = \E\left[ \{\BB(\eta,g) - \BB(\eta_0,g_0)\} \left\{\tau(V) - \tau_0(V) \right\} \right], \\
    \B_3 & = \E\left[ \{\BB(\eta_0,g_0)\} \left\{\tau(V) - \tau_0(V) \right\} \right]. 
\end{align*}
In the following, we will bound $|\B_1|$ and $|\B_2|$, and show that $\B_3 = 0$.

We first consider $|\B_1|$.
For any $(\eta,g)\in\Gc$, since $(\eta,g)$ satisfies Assumption \ref{ass:strict_positivity}, $\V(1;g)$ is bounded almost surely. 
This together with Assumption \ref{ass:1-4} (ii) implies 
\begin{align}
    |\B_1| & \lesssim \E\left[ \left\{\tau(V) - \tau_0(V) \right\}^2 \right] 
    = \|\tau - \tau_0\|_2^2.  \label{eq:prop5_proof_B1}
\end{align}

We now bound $|\B_2|$. 
Recall the bound for $\big| \E\left[ \{\BB(\eta,g) - \BB(\eta_0,g_0)\} \left\{\tau(V) - \tau_0(V) \right\} \right] \big|$ from \eqref{eq:thm3_proof_Btau}, so we have 
\begin{align}
    |\B_2| 
    & \lesssim k((\eta,g),(\eta_0,g_0)) \cdot \|\tau-\tau_0\|_2 \nonumber \\
    &\leq \frac{1}{2} k((\eta,g),(\eta_0,g_0))^2 + \frac{1}{2} \|\tau - \tau_0\|_2^2,  \label{eq:prop5_proof_B2}
\end{align}
where the last inequality holds by the fact that for any real number $a,b$, $ab \leq a^2/2 + b^2/2$.

We now consider $\B_3$. Recall $\ell$ from \eqref{eq:weighted_square_loss}. 
By Proposition \ref{prop:W_nuf(A,Z)}, 
\begin{align}
    \B_3 & = \beta^{-1} \E\left[ \omega(A^*,Z^*;\eta_0)\{\gammaone(A^*,Z^*;\eta_0)\nu(T^*) + \gammatwo(A^*,Z^*;\eta_0) - \tau_0(V^*)\}\{\tau(V^*) - \tau_0(V^*)\} \right] \nonumber\\
    & = -\frac{1}{2} D_\tau \E\{ \ell(T^*,A^*,Z^*;\tau_0,\eta)\}[\tau- \tau_0] \nonumber \\
    &= 0, \label{eq:prop5_proof_B3}
\end{align}
where the last equation holds because Assumption \ref{ass:1-4} (i) holds with ``=". 

Combining \eqref{eq:prop5_proof_1} \eqref{eq:prop5_proof_B1} \eqref{eq:prop5_proof_B2} \eqref{eq:prop5_proof_B3}, we have 
\begin{align}
    \left| L(\tau,\eta, g) - L(\tau_0,\eta, g) \right| 
    \leq |\B_1| + 2|\B_2| + 2|\B_3| 
    \lesssim \|\tau - \tau_0\|_2^2 + k((\eta,g),(\eta_0,g_0))^2. 
\end{align}

\end{proof}

\begin{proof}[Proof of Corollary \ref{cor:CATE_err_bound}]
By the definitions of $r_1, r_2, r_3$, 
Proposition \ref{prop:CATE_error_bdd} implies that we can take $r_2$ such that
\begin{align*}
    r_2(\T,\delta/2;\eta,g) \lesssim r_3(\T,\delta/4;\eta,g) + r_1(\Gc,\delta/4)^2. 
\end{align*}
Recall from Proposition \ref{thm:CATE_error_bdd} the notation of $\cc_1(n)$ and $\cc_2(n)$, and recall from \eqref{eq:c1(n)_c2(n)} the upper bound of $\cc_1(n)$ and $\cc_2(n)$.
Without loss of generality, suppose $r_1(\Gc,\delta/2) \leq r_1(\Gc,\delta/4)$. Then Proposition \ref{thm:CATE_error_bdd} implies Corollary \ref{cor:CATE_err_bound}
with constants $\cc_3(n)$, $\cc_4(n)$ satisfying 
\begin{align*}
        \cc_3(n) \lesssim \frac{1}{\{\consta - \constb  \cdot \|\D_\V(\hat g,g_0;1)\|_{\infty} \}^2}, 
        \quad 
        \cc_4(n) \lesssim \frac{1}{\consta - \constb  \cdot \|\D_\V(\hat g,g_0;1)\|_{\infty}},
\end{align*}
Since Assumption \ref{k0_converge_to_0} holds, by \eqref{eq:\ksup}, $\|\D_\V(\hat g,g_0;1)\|_{\infty} \convp 0$ as $n\to\infty$, so we have that $\cc_3(n)$ and $\cc_4(n)$ are bounded with probability arbitrarily close to 1 for large enough $n$.

\end{proof}

\subsubsection{Proofs of Proposition \ref{thm:CATE_error_rate_RC}  and Corollary \ref{cor:CATE_Op_rate} } \label{sec:CATE_oracle}

The proof of Proposition \ref{thm:CATE_error_rate_RC} is very similar to that of Theorem 3 in \citet{foster2023orthogonal},  and can be found in Section 2.9.4 of \citet{wang2025towards}.

\begin{proof}[Proof of Corollary \ref{cor:CATE_Op_rate}]
    For any solution $\crds_n$ to the inequality \eqref{eq:ineq_delta_n} that satisfies  $\crds_n^2\gtrsim R^2 \log(\log(n))/n$, 
     by Proposition \ref{thm:CATE_error_rate_RC}, we have that or any $\varepsilon>0$, with probability at least $1-\varepsilon/3$, 
    \begin{align}
        \|\hat\tau - \tau_0\|_2^2 \leq \cc_{5}(n) \cdot k((\hat\eta,\hat g), (\eta_0,g_0))^2 + \cc_{6}(n)\cdot \left\{\frac{\crds_n^2}{R^2} + \frac{\log(3/\varepsilon)}{n} \right\}, \nonumber 
    \end{align}
    where
    \begin{align*}
        \cc_{5}(n) &\leq \frac{\constg}{\{\consta - \constb \cdot \|\D_\V(\hat g,g_0;1)\|_{\infty} \}^2}, \\
        \cc_{6}(n) &\leq \frac{\consth}{\{\consta -\constb \cdot \|\D_\V(\hat g,g_0;1)\|_{\infty} \}^2} + \frac{\consth}{\consta -\constb \cdot \|\D_\V(\hat g,g_0;1)\|_{\infty} }, 
    \end{align*}
    for some universal constants $\constg, \consth>0$, as shown in the proof of Proposition \ref{thm:CATE_error_rate_RC}; and recall from \eqref{eq:D_V} the definition of $\D_\V$ and that  $\|\D_\V(g,g_0;1)\|_{\infty}$ denote the $L_\infty$ norm with respect to the random variable $O$. 
    
    By Assumption \ref{k0_converge_to_0} and \eqref{eq:\ksup}, $\|\D_\V(\hat g,g_0;1)\|_{\infty} \convp 0$ as $n\to\infty$. Therefore, for any $\varepsilon>0$, there exist an integer $N_1$ such that for any $n\geq N_1$, 
    \begin{align*}
    \PP\{ \|\D_\V(\hat g,g_0;1)\|_{\infty}  > \consta \constb^{-1}/2\} < \varepsilon/3.
    \end{align*}
    Under the event $\{ \|\D_\V(\hat g,g_0;1)\|_{\infty} \leq  \consta \constb^{-1}/2 \}$, we have 
    \begin{align*}
    \cc_{5}(n) \leq 4\consta^{-2}\constg, \quad 
    \cc_{6}(n) \leq (4\consta^{-2}+2\consta^{-1})\consth . 
   \end{align*}
    Moreover, $k((\hat\eta,\hat g), (\eta_0,g_0)) = O_p(a_n)$, so there exists $\consti>0$ and integer $N_2$ such that for all $n\geq N_2$, 
    \begin{align*}
    \PP\{k((\hat\eta,\hat g), (\eta_0,g_0)) > \consti a_n\} < \varepsilon/3.
    \end{align*}

    Combining the above, 
    there exist $C_\varepsilon>0$ such that for any $n\geq\max(N_1,N_2)$,  
    \begin{align*}
        \PP\left\{\|\hat\tau - \tau_0\|_2^2 \leq C_\varepsilon \cdot  \left(a_n^2 + \crds_n^2+n^{-1}\right) \right\} \geq 1-\varepsilon, 
    \end{align*}
    that is,
    \begin{align*}
        \PP\left( \frac{\|\hat\tau - \tau_0\|_2^2}{a_n^2 + \crds_n^2+n^{-1}} \geq C_\varepsilon \right)
        & \leq \varepsilon. 
    \end{align*}
    Therefore, 
    \begin{align*}
        \|\hat\tau - \tau_0\|_2^2 
        & = O_p\left(a_n^2 + \crds_n^2 + n^{-1}\right)
         = O_p\left(a_n^2 + \crds_n^2 \right).
    \end{align*}
    where the last equality holds because $\crds_n^2  \gtrsim R^2 \log(\log(n))^{1/2}/n$.
    So 
    \begin{align*}
        \|\hat\tau - \tau_0\|_2 
        & = O_p\left( a_n + \crds_n \right).
    \end{align*}
\end{proof}

\subsection{Instantiate the oracle results 
}\label{app:instantiate_oracle_results}

We now instantiate the oracle results for the ltrcR- and ltrcDR-learners and show the proof of the Theorems in Section \ref{sec:CATE}. 
Recall that $g = (F,G,S_D)$. Since $\eta = (\pi,F)$ for both the R-loss and the DR-loss, we have $(\eta,g) = (\pi,F,G,S_D) = (\pi,g)$.



\begin{proof}[Proof of Theorem \ref{thm:CATE_ltrcR}]

We will verify that Assumption \ref{ass:1-4} is satisfied. 
We will also show that 
$\kstar((\pi,F),(\pi_0,F_0)) \lesssim \|\pi - \pi_0\|_{4}^2 + \|\pi - \pi_0\|_{4} \cdot \|F-F_0\|_{\sup,4}$, which implies that 
\begin{align}
    k((\pi,g),(\pi_0,g_0)) 
    & \lesssim \|\pi - \pi_0\|_{4}^2 + \|F-F_0\|_{\sup,4}\cdot\left\{ \|\pi - \pi_0\|_{4}+ \|G-G_0\|_{\sup,4} + \|S_D - S_{D0}\|_{\sup,4} \right\}  \nonumber \\
    &\quad + \|K_1(g,g_0)\|_2 + \|K_2(g,g_0)\|_2 + \|K_3(g,g_0)\|_2. 
\end{align}
Therefore, 
 Propositions \ref{thm:CATE_error_bdd} and \ref{thm:CATE_error_rate_RC}  imply that the estimation error of the nuisance parameters $(\pi,F,G,S_D)$ only has second-order impacts on the estimation error rate for CATE from the ltrcR-learner through the above $k((\hat \pi,\hat g),(\pi_0,g_0))$; 
Corollary \ref{cor:CATE_Op_rate} implies Theorem \ref{thm:CATE_ltrcR}.

Note that the R-loss is a special case of \eqref{eq:weighted_square_loss} with $\eta = (\pi,F)$, $\omega(A^*,Z^*;\eta) = \{A^*-\pi(Z^*)\}^2$, $\gammaone(A^*,Z^*;\eta) = \{A^*-\pi(Z^*)\}^{-1}$, and $\gammatwo(A^*,Z^*;\eta) = -\tilde\mu(Z^*;\pi,F)/\{A^*-\pi(Z^*)\}$. By Assumption \ref{ass:strict_positivity} (i), we have that Assumption \ref{ass:1-4} (ii) (iii) are satisfied.
\citet{foster2023orthogonal} showed that the R-loss satisfies Assumption \ref{ass:1-4} (i) {with ``="} if there were no LTRC. 

We now show an upper bound for $k^*((\pi,F),(\pi_0,F_0))$. Recall $\eta = (\pi,F)$. 
Recall from \eqref{eq:bdd_by_k} \eqref{eq:Dlstar_z}: 
\begin{align}
  \kstar(\eta,\eta_0) =  \E\{|\Dlstarz(\eta) - \Dlstarz(\eta_0)|^2\}^{1/2} = \|\Dlstarz(\eta) - \Dlstarz(\eta_0)\|_2.
\end{align}
We will show that for the R-loss, 
\begin{align}
        \kstar(\eta,\eta_0)
        \lesssim \|\pi - \pi_0\|_{4}^2 + \|\pi - \pi_0\|_{4} \cdot \|F-F_0\|_{\sup,4}. \label{eq:instantiate_R_LHS}
\end{align}
From  \eqref{eq:Dlstar_z},
\begin{align*}
    \Dlstarz(\eta) 
    &= \E[\{A^* - \pi(Z^*)\}\{\nu(T^*) - \tilde\mu(Z^*;\pi,F)\} \mid Z^*]  - \E\left[\left. \{A^*-\pi(Z^*)\}^2 \right| Z^* \right] \tau_0(Z^*).
\end{align*}
Since $\nu(T^*) = A^*\nu\{T^*(1)\} + (1-A^*)\nu\{T^*(0)\}$, $\tilde\mu(Z^*;\pi,F) = \pi(Z^*) \mu(1,Z^*;F) + \{1-\pi(Z^*)\}\mu(0,Z^*;F)$, $\E(A^*|Z^*) = \pi_0(Z^*)$, and $\E[\nu\{T^*(a)\}|Z^*] = \mu(a,Z^*,F_0)$ for $a = 1,0$, we have after some algebra,
\begin{align*}
    \Dlstarz(\eta) - \Dlstarz(\eta_0)
    & = \{\pi(Z^*)-\pi_0(Z^*)\}\pi_0(Z^*) \{\mu(1,Z^*;F) - \mu(1,Z^*;F_0)\} \\
    &\quad + \{\pi(Z^*)-\pi_0(Z^*)\}\{1-\pi_0(Z^*)\}\{\mu(0,Z^*;F) - \mu(0,Z^*;F_0)\} \\
    &\quad + \{\pi(Z^*)-\pi_0(Z^*)\}^2 \mu(1,Z^*;F) - \{\pi(Z^*)-\pi_0(Z^*)\}^2\mu(0,Z^*;F) \\
    &\quad - \{\pi(Z^*) - \pi_0(Z^*)\}^2\tau_0(Z^*).
\end{align*}
By Cauchy-Schwartz inequality, 
\begin{align}
     \|\Dlstarz(\eta) - \Dlstarz(\eta_0)\|_2
    &\lesssim \|\pi - \pi_0\|_4 \cdot \|\mu(1,Z^*;F) - \mu(1,Z^*;F_0)\|_4 \nonumber \\
    &\quad + \|\pi - \pi_0\|_4 \cdot \|\mu(0,Z^*;F) - \mu(0,Z^*;F_0)\|_4 \nonumber \\
    &\quad + \|\pi - \pi_0\|_4^2. \label{eq:instantiation_R_proof1}
\end{align}
In addition, by Proposition \ref{prop:equiv_norms} and after some algebra, we have
\begin{align*}
    \|\mu(1,Z^*;F) - \mu(1,Z^*;F_0)\|_4 & \lesssim \|F - F_0\|_{\sup, 4}, \\
    \|\mu(0,Z^*;F) - \mu(0,Z^*;F_0)\|_4 & \lesssim \|F - F_0\|_{\sup, 4}, 
\end{align*}
which together with \eqref{eq:instantiation_R_proof1} concludes the proof. 

\end{proof}


\begin{proof}[Proof of Theorem \ref{thm:CATE_ltrcDR}]

In the following, we will verify that Assumption \ref{ass:obs_1-4} is satisfied and show an upper bound for $k^*((\pi,F),(\pi_0,F_0))$, which  implies Theorem \ref{thm:CATE_ltrcDR}.

Note that the DR-loss is a special case of \eqref{eq:weighted_square_loss} with $\eta = (\pi,F)$, $\omega(A^*,Z^*;\eta) = 1$, $\gammaone(A^*,Z^*;\eta) = \{A^*-\pi(Z^*)\}/[\pi(Z^*)\{1-\pi(Z^*)\}]$, and $\gammatwo(A^*,Z^*;\eta) = - \mu(A^*,Z^*;F)\{A^*-\pi(Z^*)\}/[\pi(Z^*)\{1-\pi(Z^*)\}] + \mu(1,Z^*;F) - \mu(0,Z^*;F)$. By Assumption \ref{ass:strict_positivity} (i), we immediately see that Assumption \ref{ass:1-4} (ii) (iii) hold.
\citet{foster2023orthogonal} showed that the DR-loss satisfies Assumption \ref{ass:1-4} (i) {with ``="} if there were no LTRC. 
With a similar as the proof of Theorem \ref{thm:CATE_ltrcR}, we can show that the DR-loss satisfies \eqref{eq:bdd_by_k} with 
$\kstar((\pi,F),(\pi_0,F_0)) \lesssim \|\pi - \pi_0\|_{4}\cdot \|F-F_0\|_{\sup,4}$.

Similar to the proof of Theorem \ref{thm:CATE_ltrcR}
for the ltrcDR-learner, 
we can show that 
$\kstar((\pi,F),(\pi_0,F_0)) \lesssim \|\pi - \pi_0\|_{4} \cdot \|F-F_0\|_{\sup,4}$, which implies 
\begin{align}
    k((\pi,g),(\pi_0,g_0)) 
    & \lesssim \|F-F_0\|_{\sup,4}\cdot\left\{ \|\pi - \pi_0\|_{4} + \|G-G_0\|_{\sup,4} + \|S_D - S_{D0}\|_{\sup,4} \right\}  \nonumber \\
    &\quad + \|K_1(g,g_0)\|_2 + \|K_2(g,g_0)\|_2 + \|K_3(g,g_0)\|_2. 
\end{align}
Therefore, 
Propositions \ref{thm:CATE_error_bdd} and \ref{thm:CATE_error_rate_RC} imply that the estimation error of the nuisance parameters $(\pi,F,G,S_D)$ only impacts the estimation error of $\hat\tau$ in terms of the above product and integral product errors; 
Corollary \ref{cor:CATE_Op_rate} implies the Theorem for the ltrcDR-learner in the main paper.

\end{proof}

\clearpage
\section{Additional simulation details and results}\label{app:simulation}

\subsection{ATE}\label{app:simu_ATE_n500}

We generate 
the covariates $Z^* = (Z_1^*,Z_2^*)$ with 
$Z_1^*, Z_2^*$ from independent {Uniform}(-1,1). 
For $a=1,0$, we generate the potential outcome variables and treatment assignment as follows. 
\begin{itemize}
\item $T^*(a) \sim$  Weibull distributions with shape parameter 1.5 and scale parameter $\exp\{-1.5^{-1}(-2 + 0.4a + 0.2Z_1^* + 0.3Z_2^*)\}$, {plus} a location shift $\tau_1$; 
\item $Q^*(a)$  is generated such that $\tau_2-Q^*(a)$ follows the following proportional hazards (PH) model:
$\lambda(t|Z^*) = \lambda_{0}(t) \exp(-0.6 + 0.6 a + 0.4 Z_1^* + 0.2Z_2^*),$
where 
$\lambda_{0}(t)$ is the hazard function associated with the $\text{Uniform}(0,\tau_2)$ distribution;
\item $D^*(a) \sim$  Weibull distribution with shape parameter 1 and scale parameter $\exp(1.5 - 0.3a - 0.1Z_1^* - 0.2Z_2^*)$; 
\item  $A^*$ is generated from the logistic model with
 $   \pi(Z^*) = {\exp(Z_1^* - Z_2^*)}/\{1+ \exp(Z_1^* - Z_2^*)\}$.
\end{itemize}
We then take $T^* = A^*T^*(1) + (1-A^*)T^*(0)$, $Q^* = A^*Q^*(1) + (1-A^*)Q^*(0)$, and $D^* = A^*D^*(1) + (1-A^*)D^*(0)$. 
Only subjects with $Q^*<T^*$ are observed, 
and we take $\tau_1 = 1$ and $\tau_2 = 5$ in the above resulting in about 22.8\% truncation. 
The rate of $A=1$ is about 50\% in the full data, and about 43.9\% in the truncated data. The censoring rate is about 47.6\% in the observed data.
Under the above data generating mechanism, the  distributions of $T^*$ given $(A^*, Z^*)$, $\tau_2 -Q^*$ given $(A^*, Z^*)$, and $D$ given $(Q,A,Z)$ all follow  the PH models. 

In estimating $F,G$ and $S_D$, `pCox' is regularized Cox regression using the R package ``\texttt{glmnet}'' with $L_2$ penalty and 
\texttt{s = lambda.min} to select the best tuning parameter.
For continuous variables 
we {generate} natural spline basis functions with 7 degrees of freedom {using} the \texttt{ns()} function in the R package `\texttt{splines}'. 
The {regressors} in  ``\texttt{glmnet}'' include the following: squared terms and interaction terms between the spline basis functions, interaction terms between the spline basis functions and binary variables, and interaction terms between binary variables (if applicable). For our simulation, there are two continuous covariates $(Z_1,Z_2)$ and one binary variable $A$ for estimation of $F$ and $G$, resulting in 134 regressors in the model; and there are three continuous covariates $(Z_1,Z_2,Q)$ and one binary covariate $A$ for estimation of $S_D$, resulting in 274 regressors in the model.

For estimation of  $G$, the  models are fitted on the reversed time scale. In particular, we take $t_1 = \max_i T_i +1$ when computing the reversed-time variables $t_1-X$ and $t_1-Q$.

For estimating $\pi$, 
`gbm' is generalized boosted regression models implemented in the R package ``\texttt{gbm}'' with \texttt{ntree = 2000}.

\begin{table}[h]
	\centering
 \renewcommand{\arraystretch}{0.6}
	\caption{Simulation results for estimates of $\theta=\PP\{T^*(1)>3\} - \PP\{T^*(0)>3\} = -0.1163$  with sample size 500. The models in \mis{red} are misspecified, while the models in \smis{orange} are correctly specified  with weights estimated from the misspecified models.
 } 
	\label{tab:simu_n500_dr_cf_simu2}
\begin{tabular}{llrrrr}
  \toprule
 Methods & $F$/$\pi$-$G$-$S_D$ & bias & SD & SE/bootSE & CP/bootCP \\ 
\midrule
   $\hat\theta$ & Cox1/lgs1-Cox1-Cox1 & -0.0001 & 0.0689 & 0.0710/0.0735 & 0.956/0.962 \\
   & \red{Cox2}/lgs1-Cox1-Cox1 & -0.0002 & 0.0675 & 0.0715/0.0721 & 0.964/0.964 \\
   & Cox1/\org{lgs1}-\red{Cox2}-Cox1 & 0.0017 & 0.0640 & 0.0644/0.0678 & 0.956/0.962 \\ 
   & Cox1/\red{lgs2}-Cox1-Cox1 & -0.0002 & 0.0640 & 0.0653/0.0691 & 0.950/0.964 \\ 
   & Cox1/\org{lgs1}-\org{Cox1}-\red{Cox2} & 0.0033 & 0.0711 & 0.0741/0.0751 & 0.956/0.962 \\  
   & \red{Cox2}/\org{lgs1}-\red{Cox2}-Cox1 & 0.0046 & 0.0642 & 0.0679/0.0686 & 0.964/0.968 \\
   & \red{Cox2}/\red{lgs2}-Cox1-Cox1 & 0.0177 & 0.0631 & 0.0657/0.0676 & 0.960/0.960 \\ 
   & \red{Cox2}/\org{lgs1}-\org{Cox1}-\red{Cox2} & -0.0015 & 0.0696 & 0.0745/0.0739 & 0.966/0.966 \\ 
   \midrule 
   $\hat\theta_{cf}$ & pCox/gbm-pCox-pCox & -0.0568 & 0.1325 & 0.1359/0.1894 & 0.932/0.985 \\ 
   \midrule
   IPW & - /lgs1-Cox1-Cox1 & -0.0008 & 0.0698 & 0.0725/0.0721 & 0.954/0.942 \\  
   & - /\org{lgs1}-\red{Cox2}-Cox1 & 0.0738 & 0.0658 & 0.0699/0.0676 & 0.818/0.804 \\  
   & - /\red{lgs2}-Cox1-Cox1 & 0.0070 & 0.0669 & 0.0669/0.0687 & 0.950/0.948 \\ 
   & - /\org{lgs1}-\org{Cox1}-\red{Cox2} & -0.0341 & 0.0709 & 0.0729/0.0727 & 0.930/0.916 \\ 
   & - /gbm-pCox-pCox & 0.0226 & 0.0887 & 0.0886/0.1035 & 0.948/0.972 \\  
   \midrule
   full data &  & 0.0008 & 0.0247 & 0.0260/0.0260 & 0.964/0.962 \\   
   
\bottomrule
 \end{tabular}
 \end{table}

\subsection{CATE}

We generate  $A^*$ and $Z^*$  the same  as for the ATE. 
For $T^*(a)$ 
we consider the following three scenarios with $\nu(t) = \log(t)$, 
where $\epsilon(1)$ and $\epsilon(0)$ are independent and centered, i.e.~mean zero, Weibull distribution with shape parameter 2 and scale parameter 0.04: 
\begin{itemize}
\item[(i)] $\log\{T^*(a)\} = 0.8+0.2a-0.2aZ_1^* + 0.2 \sqrt{|Z_2^*|} + \epsilon(a)$, giving $\tau(z) =  0.2 -0.2z_1$; 
\item[(ii)] $  \log\{T^*(a)\} = 0.8+0.2a-0.2 a\cdot\{(Z_1^*+Z_2^*)/2\}^2 + 0.2 Z_2^* + \epsilon(a)$, giving $\tau(z) = 0.2 -0.2\{(z_1+z_2)/2\}^2$;  
\item[(iii)] $  \log\{T^*(a)\} = 0.8+0.2a-0.2 a\cdot\sin(\pi Z_1^*) + 0.2 a\sqrt{|Z_2^*|} + \epsilon(a)$, giving $\tau(z) = 0.2 - 0.2\sin(\pi z_1) + 0.2 \sqrt{|z_2|}$.
\end{itemize}
$Q^*(a)$  is generated such that $\tau_2-Q^*(a)$ follows the PH model:
$\lambda_a(t|Z^*) = \lambda_{0}(t) \exp(-0.6 + 0.4 Z_1^* + 0.5 aZ_2^*)$, 
where $\lambda_{0}(t)$ is the hazard function associated with the $\text{Uniform}(0,\tau_2)$ distribution. 
Finally $D^*(a)$ 
is generated in the same way as Section \ref{sec:simu_ATE}.
Only subjects with $Q^*<T^*$ are observed, and  
we take $\tau_2 = 5$ in the above resulting in around 28\% truncation. 
The rate of $A=1$ is around 50\% in the full data, and around 44\% in the truncated data. The censoring rate is around 50\% in the observed data. 

The IPW.S-learner has loss function 
    \begin{align}
        \tilde\ell_S(\tau; F,G,S_D) = 
        \frac{\Delta}{S_D(X-Q|Q,A,Z)G(X|A,Z)} \  \left\{\mu(1,Z; {F}) - \mu(0,Z;{F}) - \tau(Z)\right\}^2. \nonumber
    \end{align}

\subsubsection{Details for tuning parameter selection for extreme gradient boosting}\label{app:tuning_selection}

The tuning parameters involved in the extreme gradient boosting implemented in the R package \texttt{xgboost} are chosen from the following candidate set:
\begin{itemize}
    \item \texttt{subsample}: 0.5, 0.7, 0.9;
    \item \texttt{colsample\_bytree}: 0.6, 0.8, 1;
    \item \texttt{eta}: 0.001, 0.005, 0.01, 0.02, 0.05, 0.08, 0.1;
    \item \texttt{max\_depth}: 2, 3, 4, 5, 6;
    \item \texttt{gamma}: 0, 0.5, 1, 2, 3, 5;
    \item \texttt{min\_child\_weight}: 10, 15, 20, ... , 50;
    \item \texttt{max\_delta\_step}: 0, 2, 4, ... , 10.
\end{itemize}

There are in total $10^5$ 
combinations from the above candidate set, making an exhaustive search for selecting the best tuning parameters computationally expensive. We therefore adopt the random search algorithm, which is also used in \citet[code on GitHub]{nie2021quasi}. The algorithm with $N_s \  (= 200$ in simulations and 500 for HAAS data application) random searches is stated below: 
\begin{enumerate}
    \item For each random search,  a set of tuning parameters {is randomly selected} from the above candidate sets.
    \item With each set of tuning parameters, 10 fold cross validation (CV) {is used} to compute the mean squared error of the model using extreme gradient boosting with the number of trees ranging between 1 to 1000; the number of trees with the smallest mean squared error is selected as the best number of trees for this set of tuning parameter.
    \item {By comparing} the mean squared errors from the fitted model with the best number of trees across the $N_s$ random searches, the set of tuning parameter and the corresponding best number of trees with the smallest mean squared error is selected to fit the final model.
\end{enumerate}

\subsubsection{Visualization of the estimated CATE surface from one simulated data set}\label{app:CATE_simu_3Dplots}

The following plots visualize the true CATE surface and the estimated CATE surface from the learners in Section \ref{sec:simu_CATE} for one simulated data set under Scenarios (i) (ii) and (iii), respectively. 

\begin{figure}[h]
\centering
\begin{subfigure}[H]{0.48\textwidth}
		\includegraphics[width=1\textwidth]{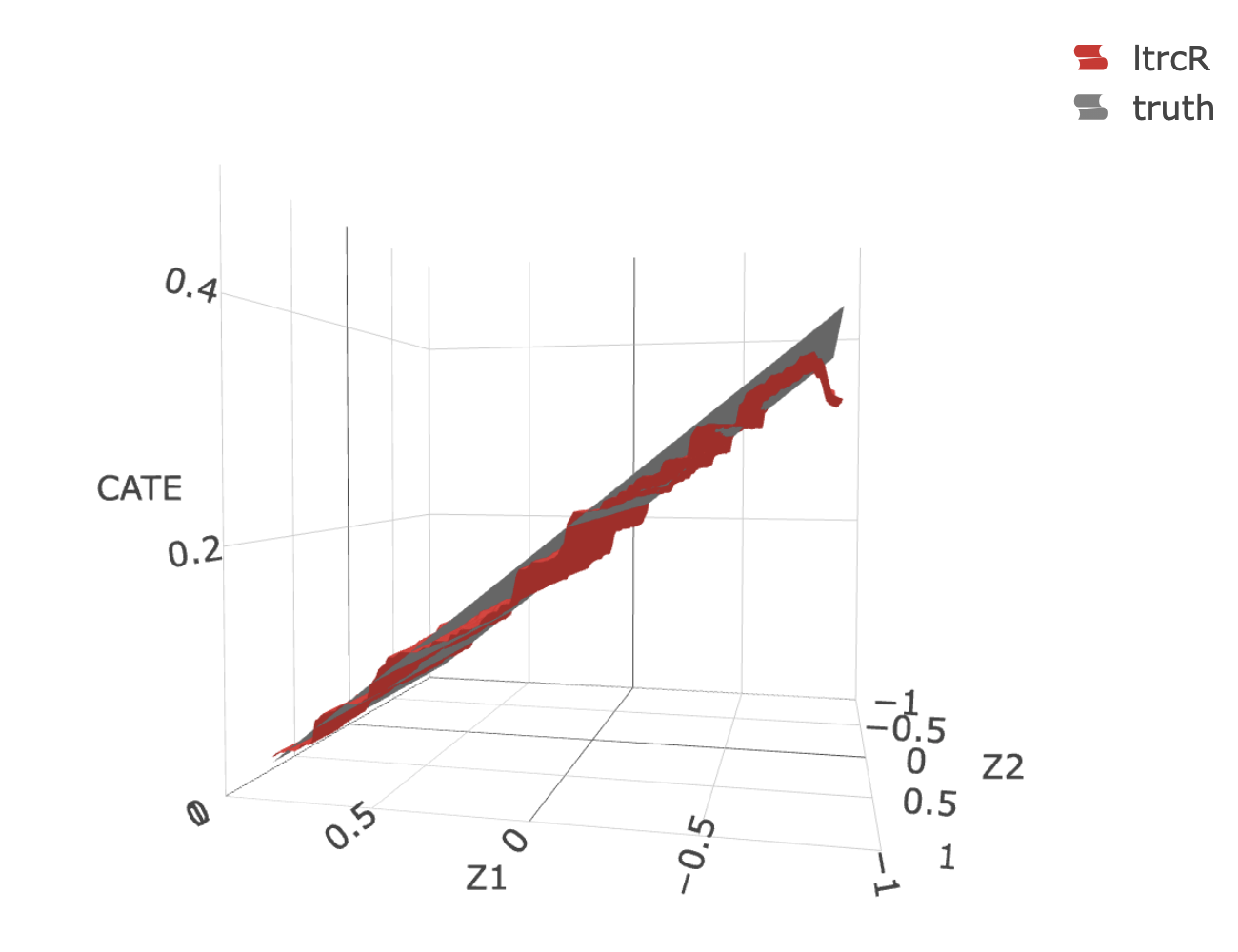}
	\end{subfigure}
    \hfill
 \begin{subfigure}[H]{0.48\textwidth}
		\includegraphics[width=1\textwidth]{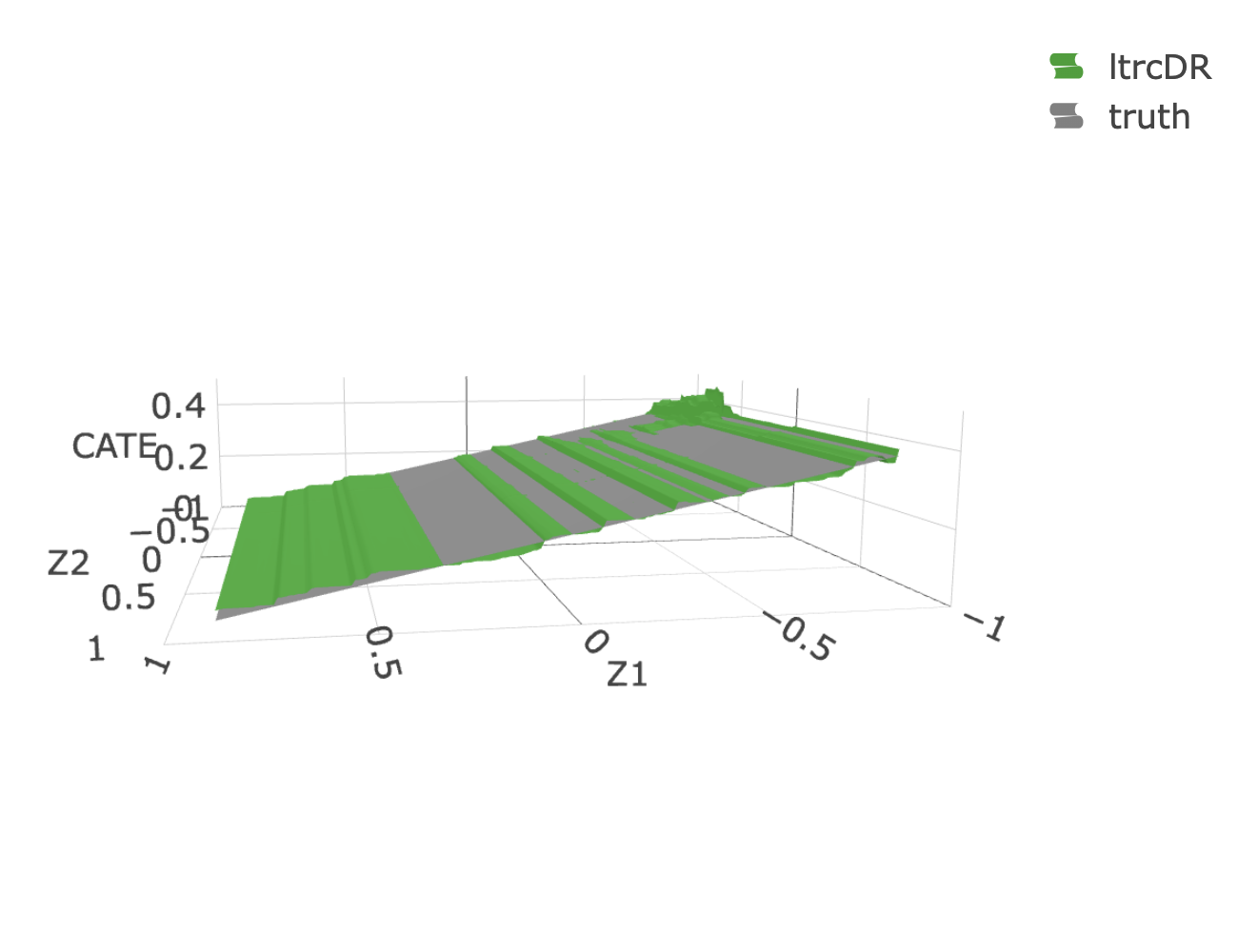}
\end{subfigure}

\vspace{2em}

 \begin{subfigure}[H]{0.48\textwidth}
		\includegraphics[width=1\textwidth]{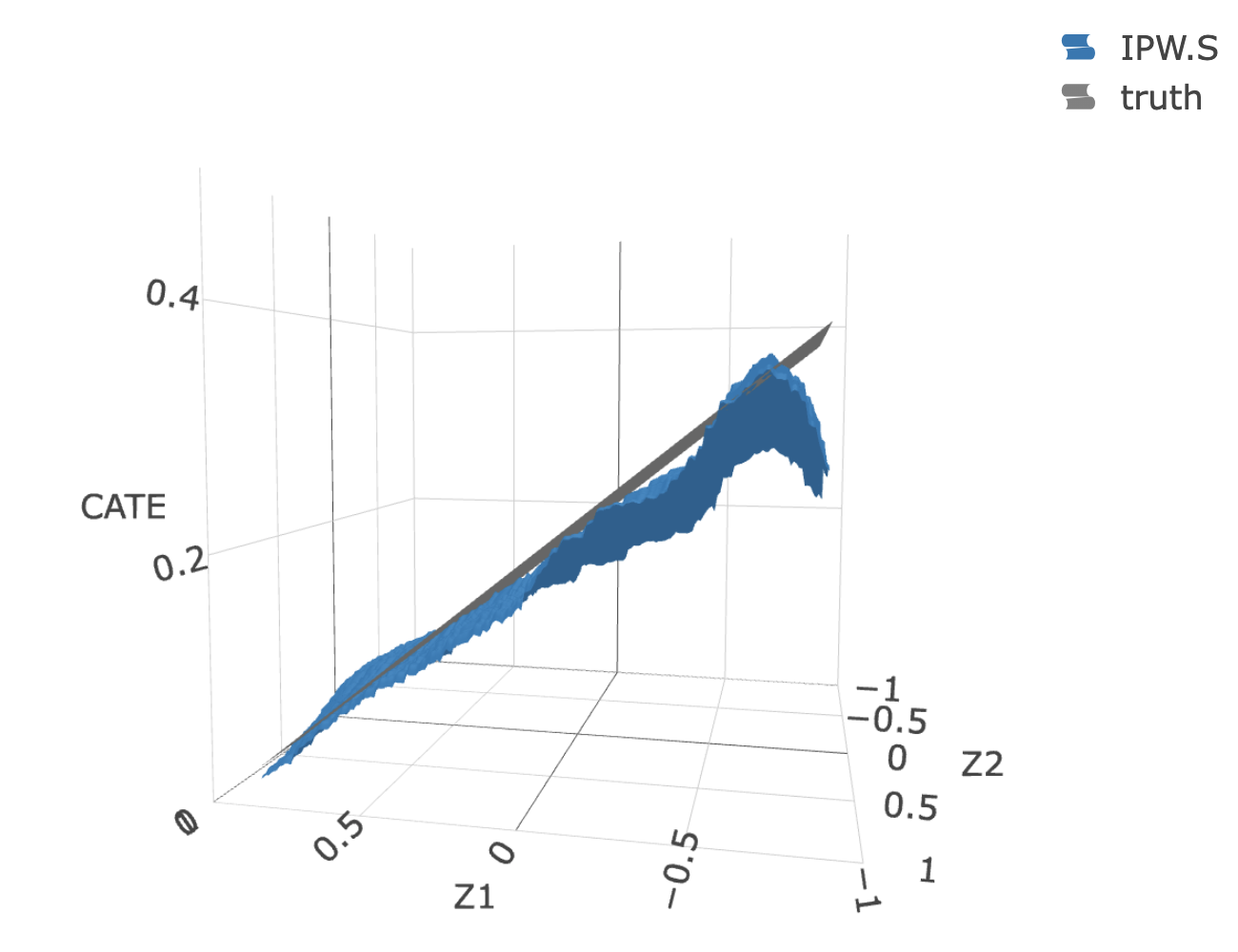}
\end{subfigure}
\caption{3D-plots for the true CATE surface and the estimated CATE surface for one simulated data set from Scenario (i).} 
\end{figure}

\begin{figure}[h]
\centering
\begin{subfigure}[H]{0.48\textwidth}
		\includegraphics[width=1\textwidth]{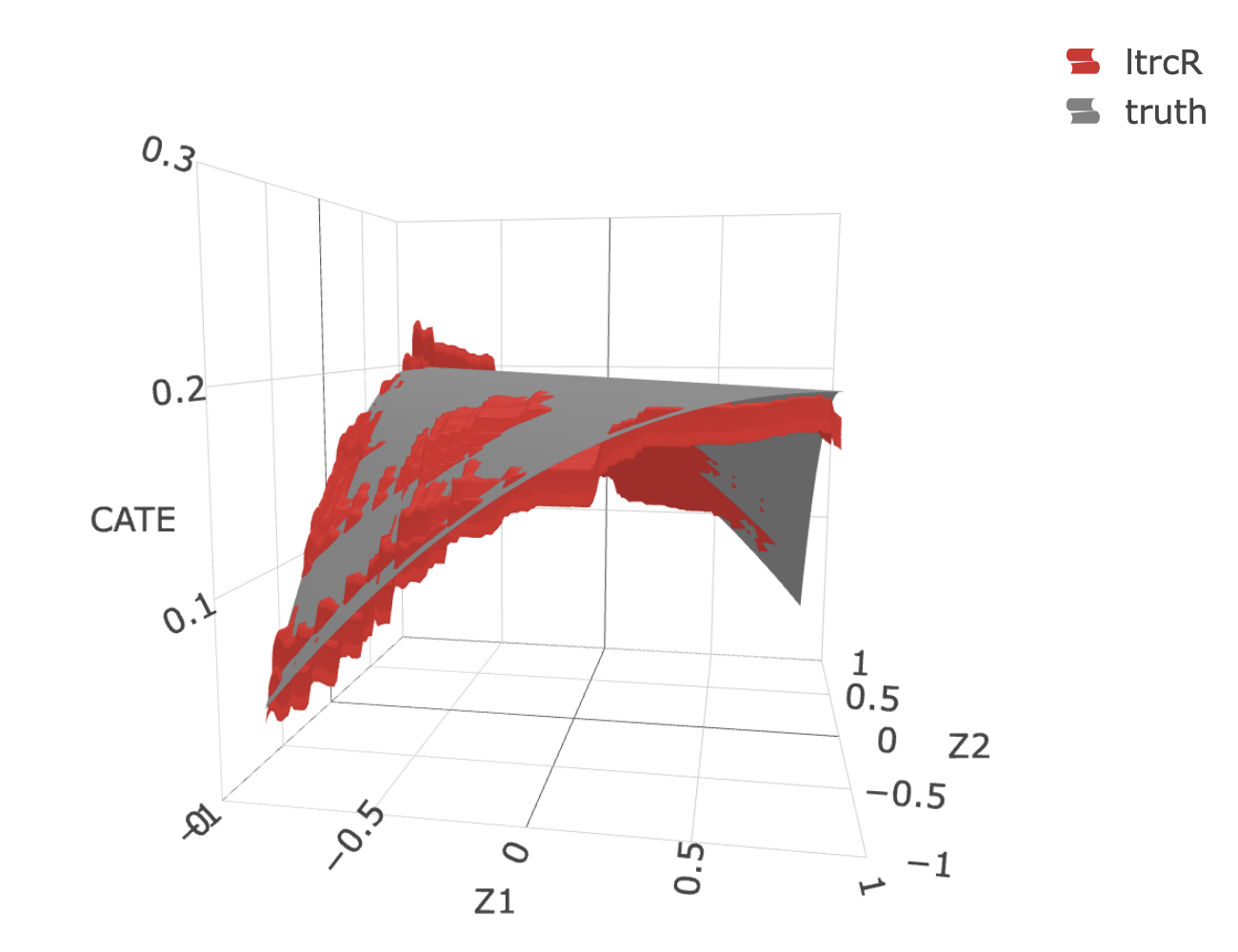}
	\end{subfigure}
    \hfill
 \begin{subfigure}[H]{0.48\textwidth}
		\includegraphics[width=1\textwidth]{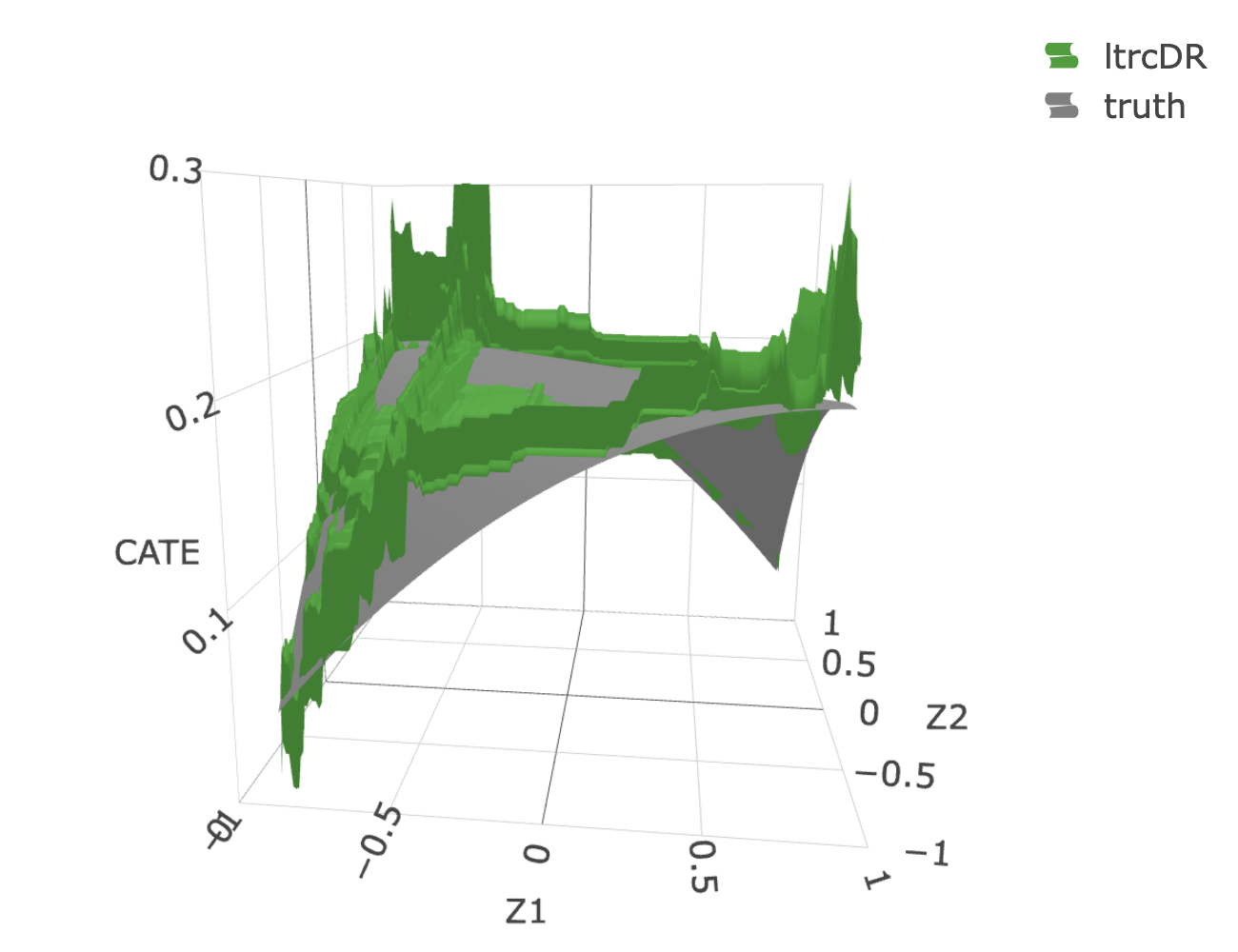}
\end{subfigure}

\vspace{2em}

 \begin{subfigure}[H]{0.48\textwidth}
		\includegraphics[width=1\textwidth]{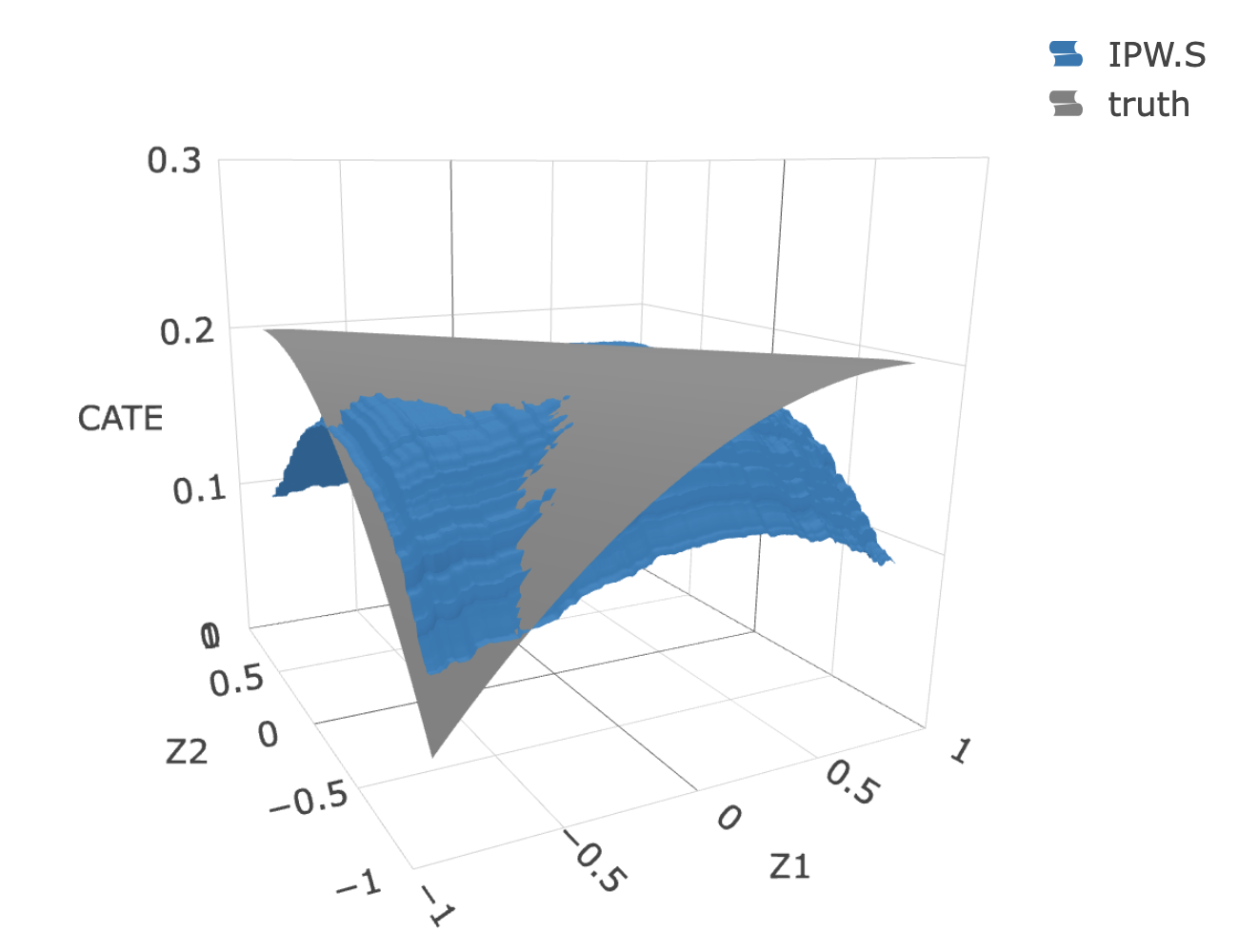}
\end{subfigure}
\caption{3D-plots for the true CATE surface and the estimated CATE surface for one simulated data set from Scenario (ii).} 
\end{figure}

\begin{figure}[h]
\centering
    \begin{subfigure}[H]{0.48\textwidth}
		\includegraphics[width=1\textwidth]{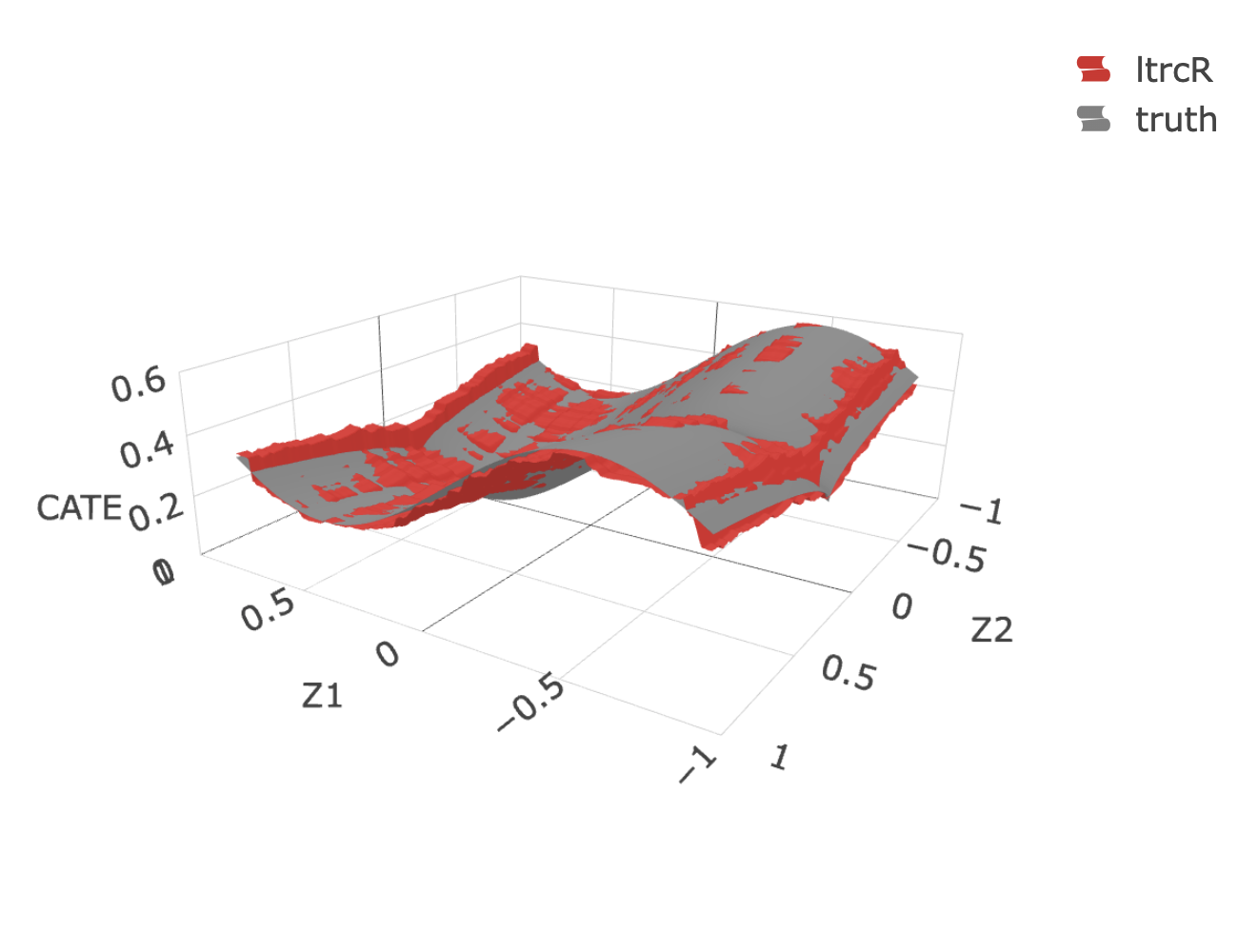}
	\end{subfigure}
    \hfill
    \begin{subfigure}[H]{0.48\textwidth}
		\includegraphics[width=1\textwidth]{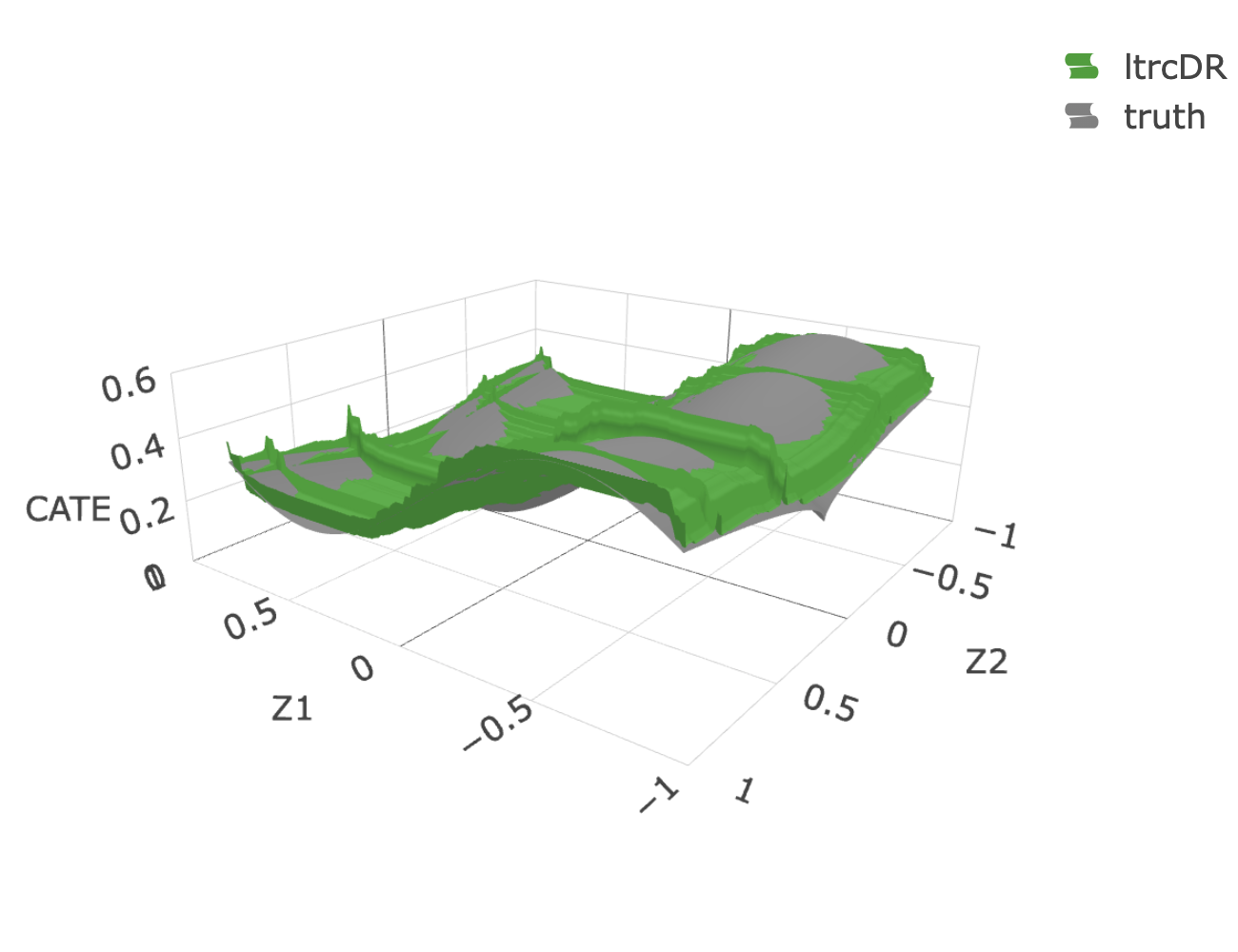}
    \end{subfigure}

    \vspace{2em}
    
 \begin{subfigure}[H]{0.48\textwidth}
		\includegraphics[width=1\textwidth]{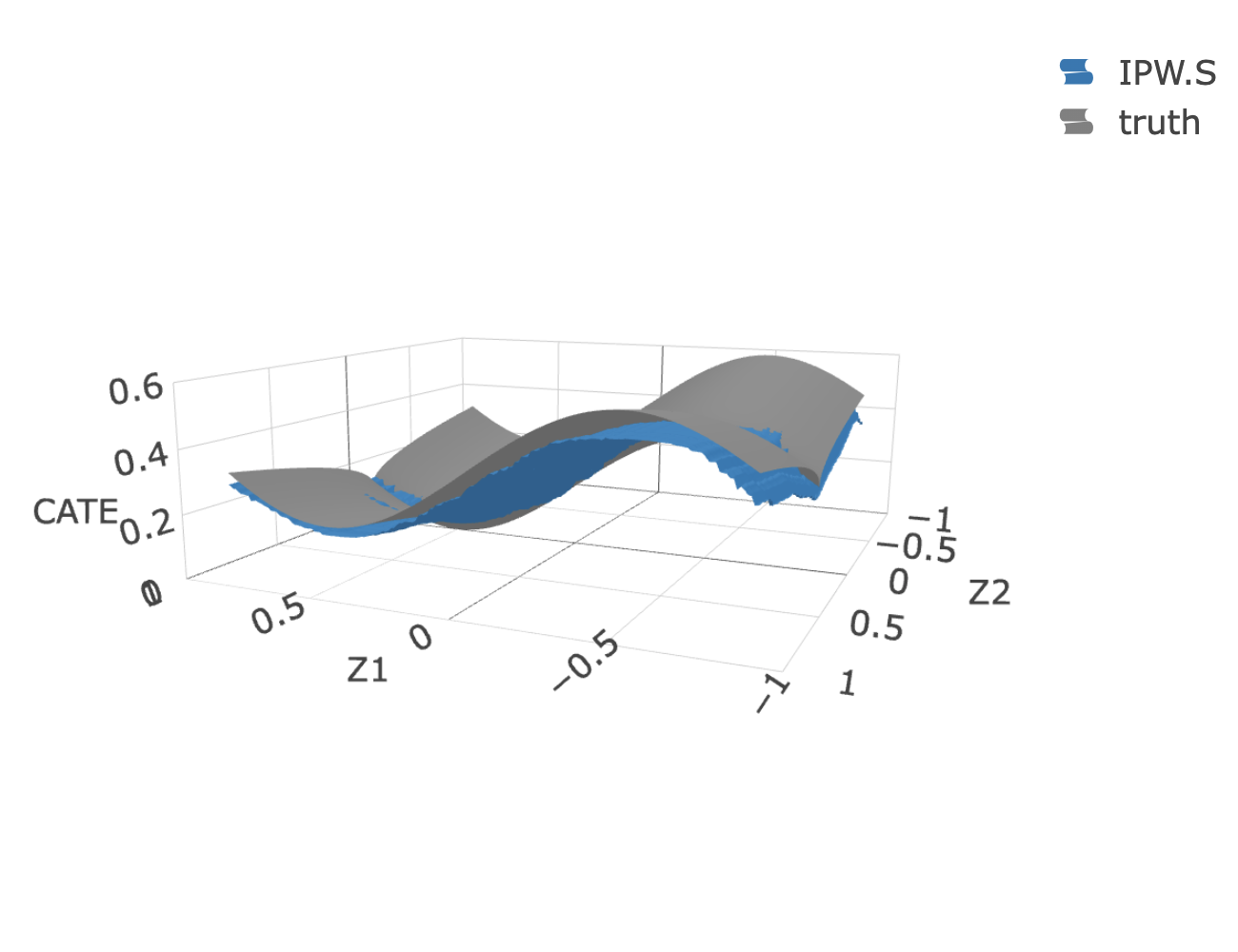}
\end{subfigure}
\caption{3D-plots for the true CATE surface and the estimated CATE surface for one simulated data set from Scenario (iii).} 
\end{figure}

\clearpage
\subsubsection{Additional simulation results for CATE}\label{app:simu_CATE_additional}

Figure \ref{fig:CATE_MSE_HTE4_trim05} shows the simulation results for CATE for the same learners and under the same scenarios as in Section \ref{sec:simu_CATE}, except that the probabilities involved in the denominators of the loss function expressions are truncated at 0.05 instead of 0.1.%

\begin{figure}[h]
        \centering
        \begin{subfigure}{0.43\textwidth}
        \includegraphics[width=1\textwidth]{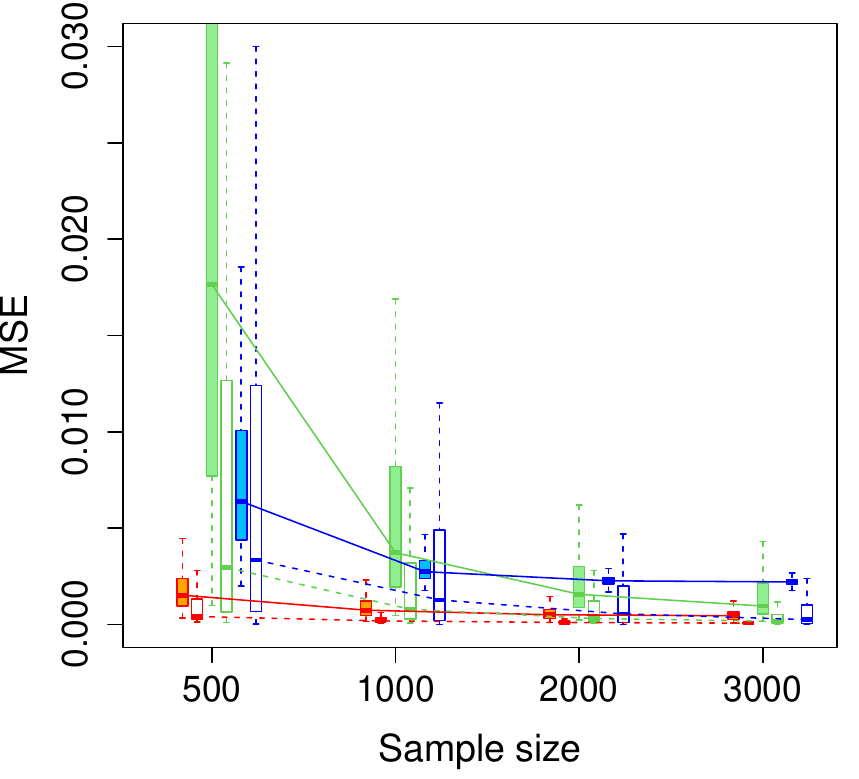}
         \caption*{Scenario (i)} 
         \end{subfigure}
        \hspace{1em}
        \begin{subfigure}{0.43\textwidth}
         \includegraphics[width=1\textwidth]{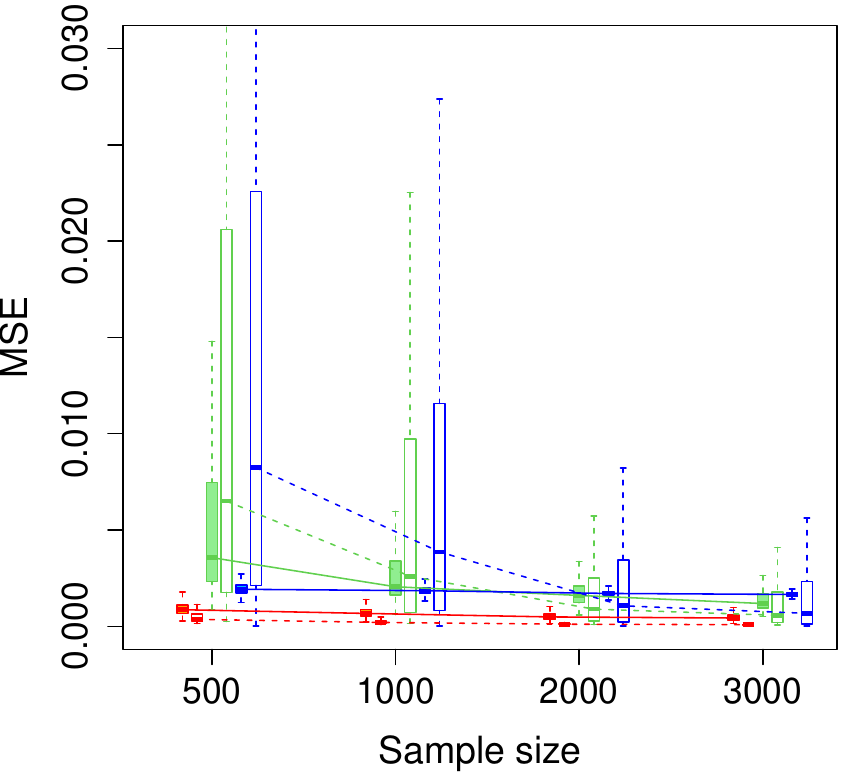}
         \caption*{Scenario (ii)}
        \end{subfigure}

        \vspace{0.7em} 

        \begin{subfigure}{0.43\textwidth}
         \includegraphics[width=1\textwidth]{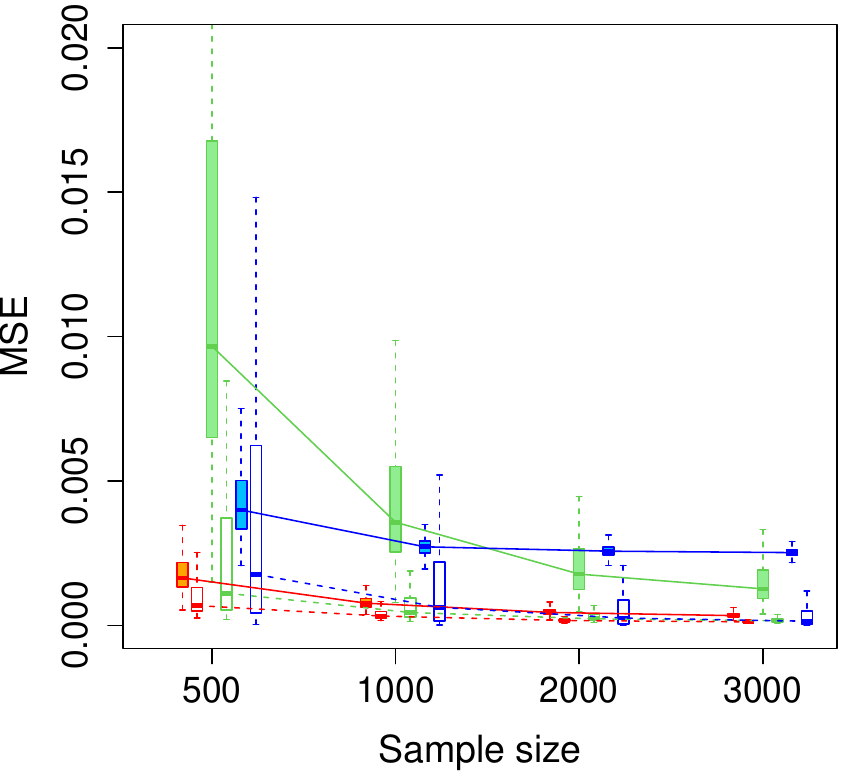}
         \caption*{Scenario (iii)}
         \end{subfigure}
        \hspace{0.7cm}
	\begin{subfigure}{0.3\textwidth}
         \includegraphics[width=0.5\textwidth]{figures/HTEx_MSE_legend.pdf}
         \end{subfigure}
        \caption{MSE for different learners under Scenarios (i) - (iii), respectively;  ``-o" indicates the oracle learner with the true nuisance parameters; the probabilities involved in the denominators of the loss function expressions are truncated at 0.05.} \label{fig:CATE_MSE_HTE4_trim05}
\end{figure}

\clearpage
\section{Additional HAAS data analysis results}\label{app:HAAS_additional}

\subsection{Alcohol consumptions and covariate distributions}\label{app:HAAS_table1}

\begin{table}[ht]
\centering
\caption{Covariate distribution in HAAS data  by alcohol consumption.}
\label{tab:table1}
\begin{tabular}{lcc}
  \toprule
  & \multicolumn{2}{c}{Mid-life alcohol consumption } \\
  \cmidrule(lr){2-3} 
  &  not heavy ($n=$1483) & heavy ($n=$470) \\ 
  \midrule
Education - count (\%)\\
\quad $\leq 12$ years  & 1128 (76.1) & 403 (85.7) \\ 
{\it APOE}  - count (\%)\\
\quad {\it APOE E4} positive & 278 (18.7) & 97 (20.6) \\ 
SystolicBP - mean (SD) & 149.05 (21.68) & 151.07 (22.92) \\ 
HeartRate - mean (SD) & 31.47 (4.63) & 31.90 (4.92) \\ 
   \bottomrule 
\end{tabular}
\end{table}

\begin{table}[ht]
\centering
\caption{Heavy alcohol consumption by education and {\it APOE}  status
}\label{tab:A1}
\begin{tabular}{lcc}
\toprule
 & {\it APOE E4} negative & {\it APOE E4}  positive \\
\midrule
Education $\leq$ 12 years & {\color{red} 317/1235 (25.7\%)} &  {\color{darkblue} 86/296 (29.1\%)} \\
Education $>$ 12 years    & {\color{darkgreen}  56/343 (16.3\%)} & {\color{brown}  11/79 (13.9\%)} \\
\bottomrule
\end{tabular}
\end{table}

\subsection{Tests for dependence}\label{app:dep}

For this data set we applied conditional Kendall's tau test \citep{tsai1990testing,martin2005testing} for potential violation of quasi-independence between age at HAAS study entry and age at moderate cognitive impairment or death. The p-value from the test is 0.0032, providing strong evidence against the quasi-independence between the ages.
We also assessed the dependence between age at HAAS study entry 
and the covariates, as well as the dependence between the residual censoring time and the covariates, by fitting two multivariate Cox models, respectively. 
The Wald test for the effect of education and systolic blood pressure on age at HAAS study entry has $p$-value $< 10^{-5}$ each, 
suggesting a strong association between age at HAAS study entry and these variables.
For the Cox model on the residual censoring time, 
the Wald test $p$-values are 0.08 and 0.07 for heart rate and alcohol exposure, respectively; although these p-values do not reach statistical significance at the 0.05 level, their magnitudes suggest a potential association between the residual censoring  time  and the covariates.

\subsection{Interesting features of estimated CATE}\label{app:PL_curves}

Other noticeable features of Figure \ref{fig:HAAS_CATE_plot3D} are, for example, 
in the upper left plot the estimated CATE for all four groups show a dip for heart rate  between 64 - 76, indicating that mid-life heavy drinking  has  more of a negative effect on DFS for subjects with normal  heart rates. 
On the upper right as well as bottom left we see that 
different education groups (red versus greeen, blue versus yellow) appear to have little difference in the estimated CATE for people with lower SBP ($<160$) and lower heart rate ($<68$). 
On the bottom two plots we see that the estimated CATE surfaces show a spike for SBP between 158 - 180,
and the spike remains even when we impose more smoothness to the estimated CATE surfaces.

Figures \ref{fig:HAAS_PL_curves_edu0_APOE0} - \ref{fig:HAAS_PL_curves_edu1_APOE1} contain the product-limit (PL) plots of DFS within categories of the four covariates which under Assumptions \ref{ass:noUnmConf} - \ref{ass:cen}, can be considered approximately valid estimates of the DFS probabilities. The categories for heart rate and SBP are informed by the CATE surfaces as discussed above, and the features such as the spike for SBP between 158 - 180, and the dip for heart rate  between 64 - 76, appear to be roughly consistent with the DSF differences as seen in the PL plots.  

\begin{figure}[h]
\centering


\includegraphics[width=1\textwidth]{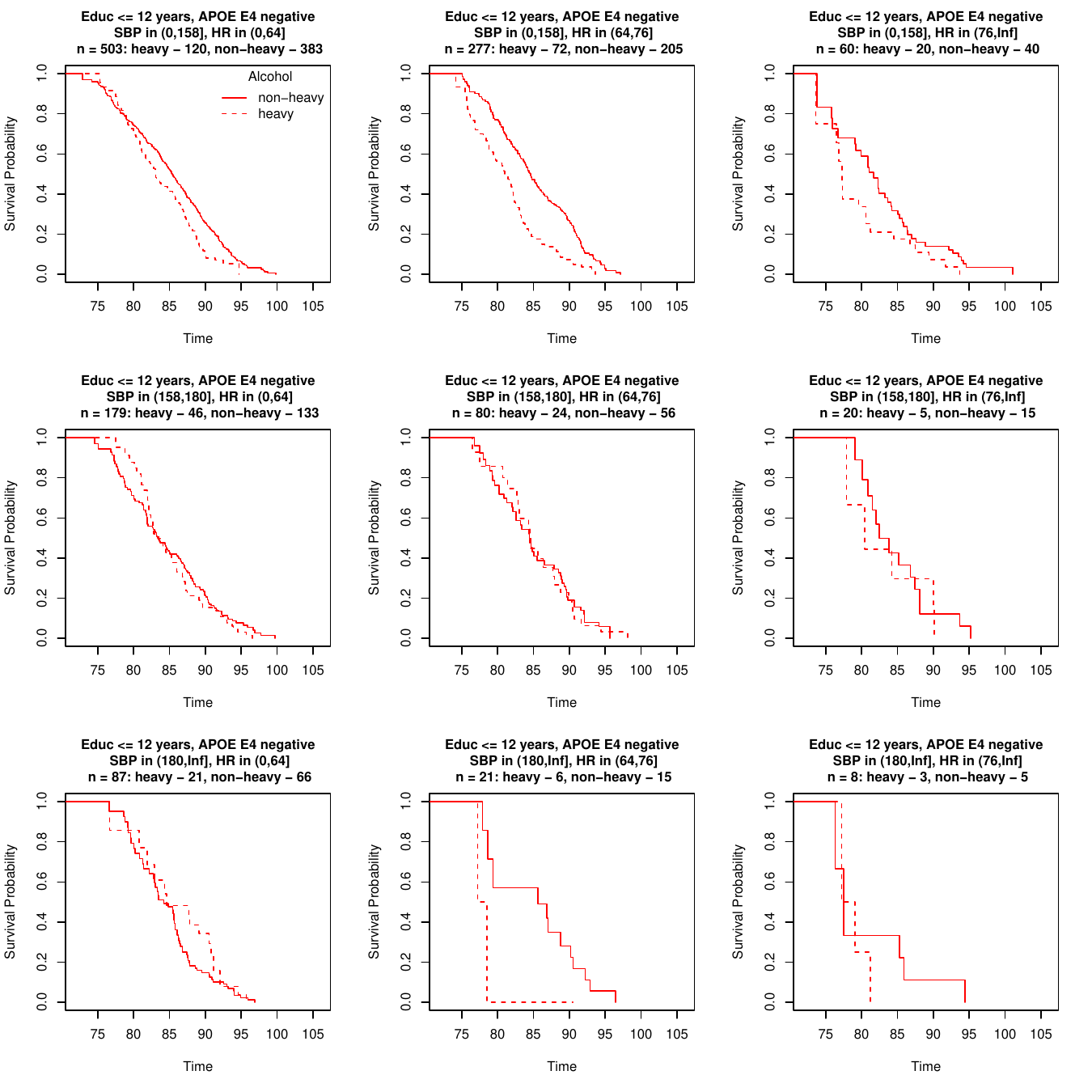}
\caption{Product-limit curves for heavy vs non-heavy drinking in the Education $\leq 12$ years and {\it APOE E4} negative group (correspond to the red surface).
}
\label{fig:HAAS_PL_curves_edu0_APOE0}
\end{figure}

\begin{figure}[h]
\centering


\includegraphics[width=1\textwidth]{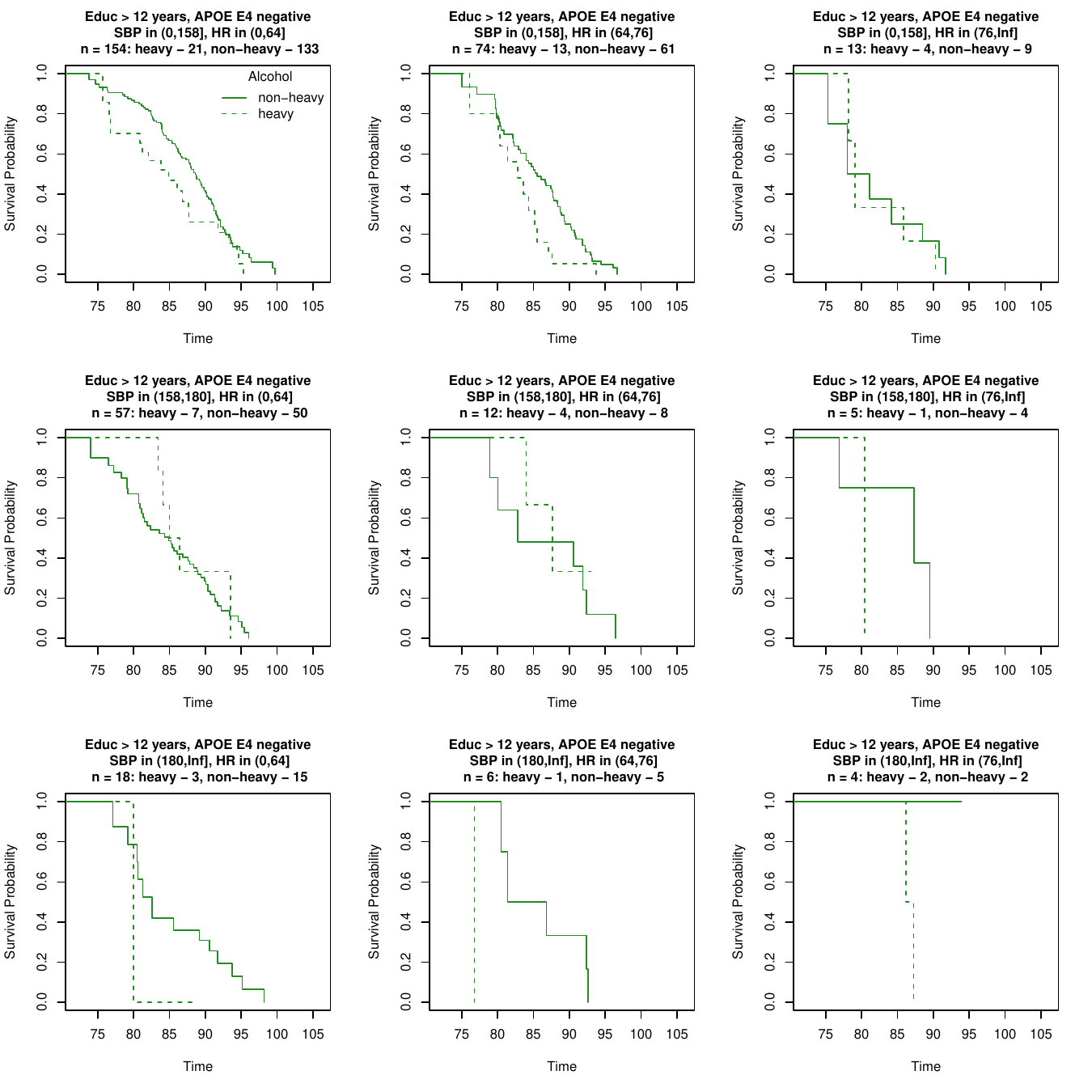}
\caption{Product-limit curves for heavy vs non-heavy drinking in the Education $>12$ years and {\it APOE E4} negative group (correspond to the green surface).  
}
\label{fig:HAAS_PL_curves_edu1_APOE0}
\end{figure}

\begin{figure}[h]
\centering
\includegraphics[width=1\textwidth]{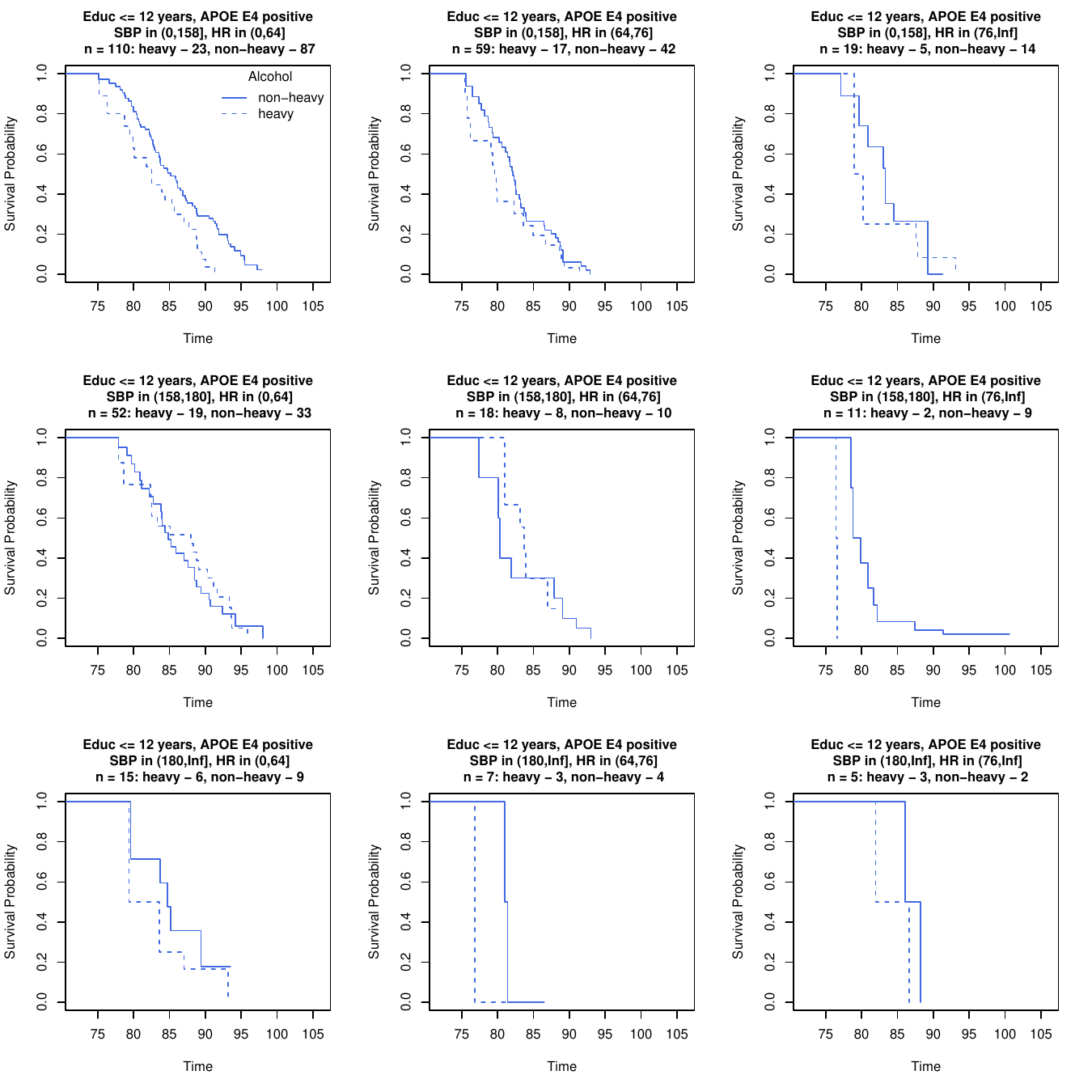}
\caption{Product-limit curves for heavy vs non-heavy drinking in the Education $\leq 12$ years and {\it APOE E4} positive group (correspond to the blue surface).
}
\label{fig:HAAS_PL_curves_edu0_APOE1}
\end{figure}

\begin{figure}[h]
\centering
\includegraphics[width=1\textwidth]{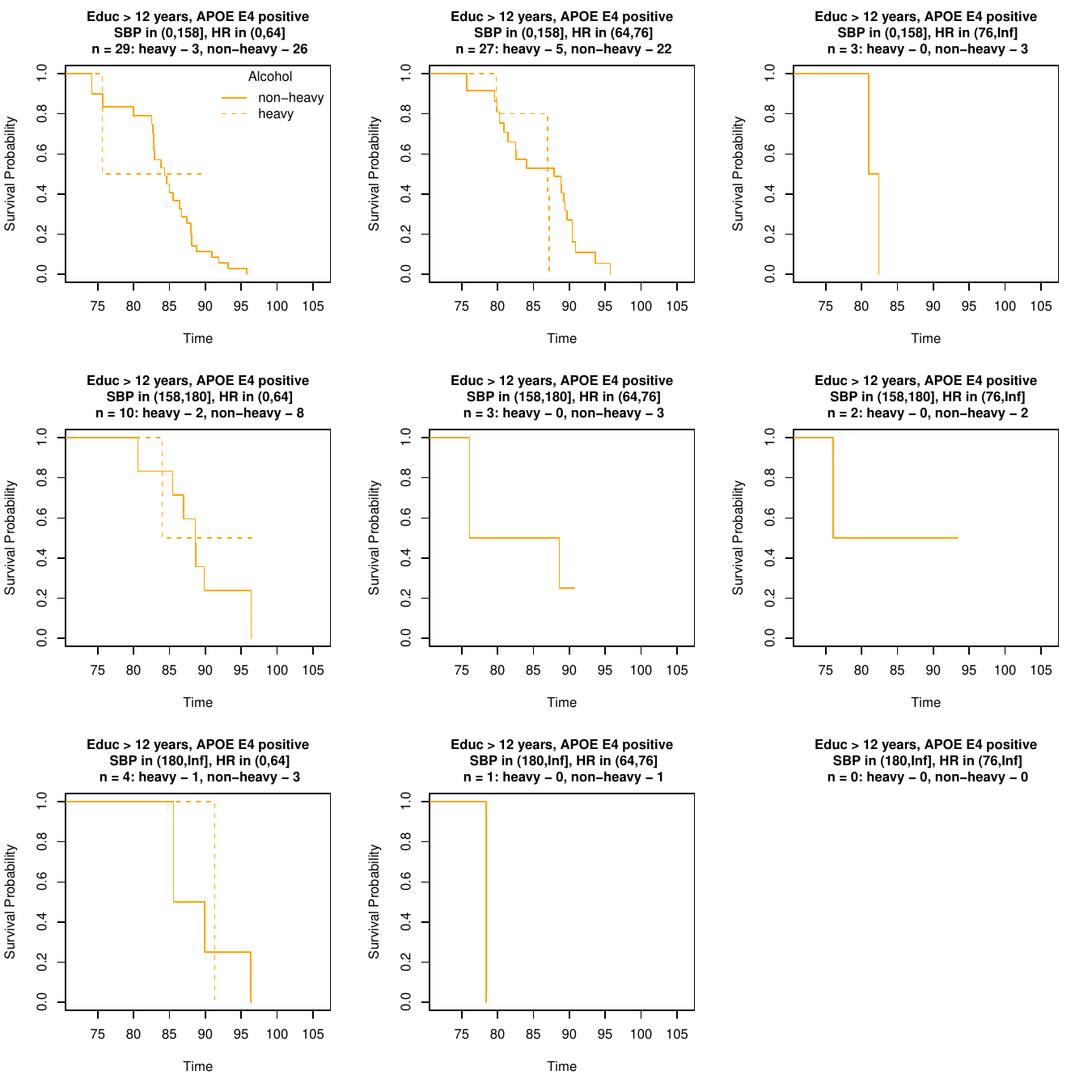}
\caption{Product-limit curves for heavy vs non-heavy drinking in the Education $> 12$ years and {\it APOE E4} positive group (correspond to the yellow surface).
}
\label{fig:HAAS_PL_curves_edu1_APOE1}
\end{figure}

\clearpage
\subsection{Additional analysis results for DFS at 90 years of age}\label{app:HAAS_CATE_surv90}

\subsubsection{CATE surfaces from the ltrcDR-learner}

Below are estimated CATE surfaces from the ltrcDR-learner with
the tuning parameters selection described in Section \ref{app:tuning_selection} for the extreme gradient boosting implemented in the R package \texttt{xgboost}.

\begin{figure}[h]
\centering
\includegraphics[width=0.49\textwidth]{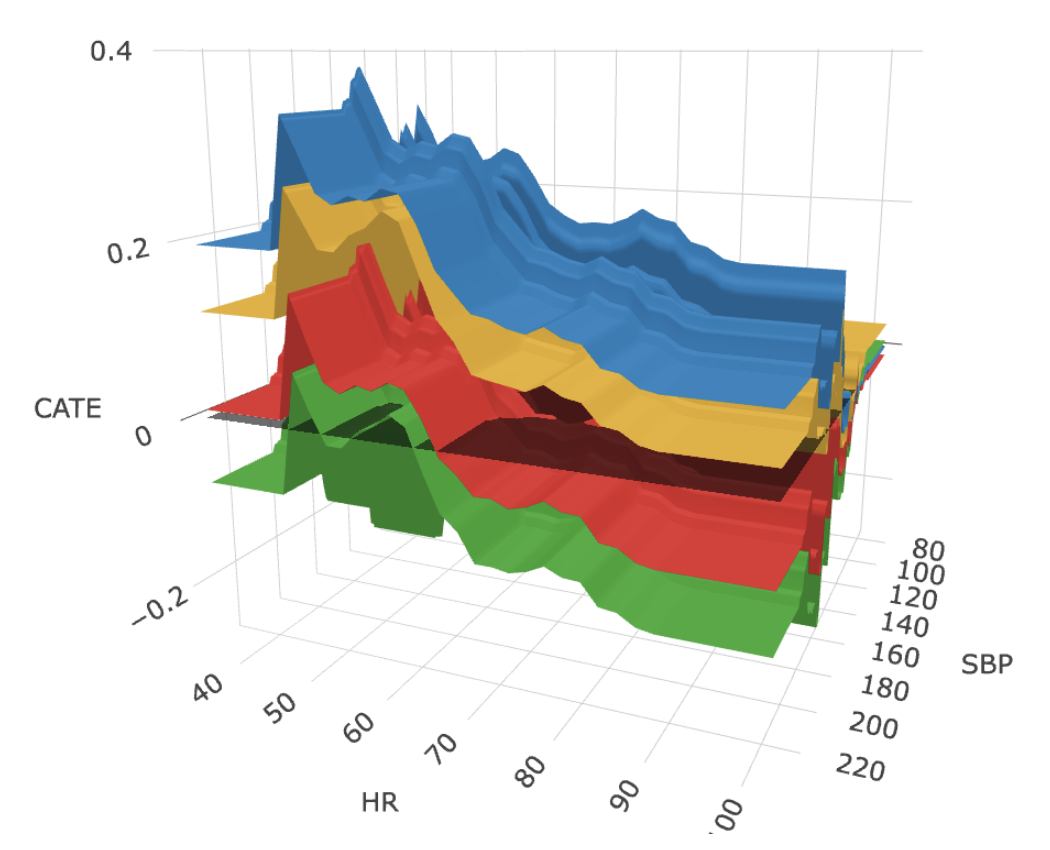}
\includegraphics[width=0.49\textwidth]{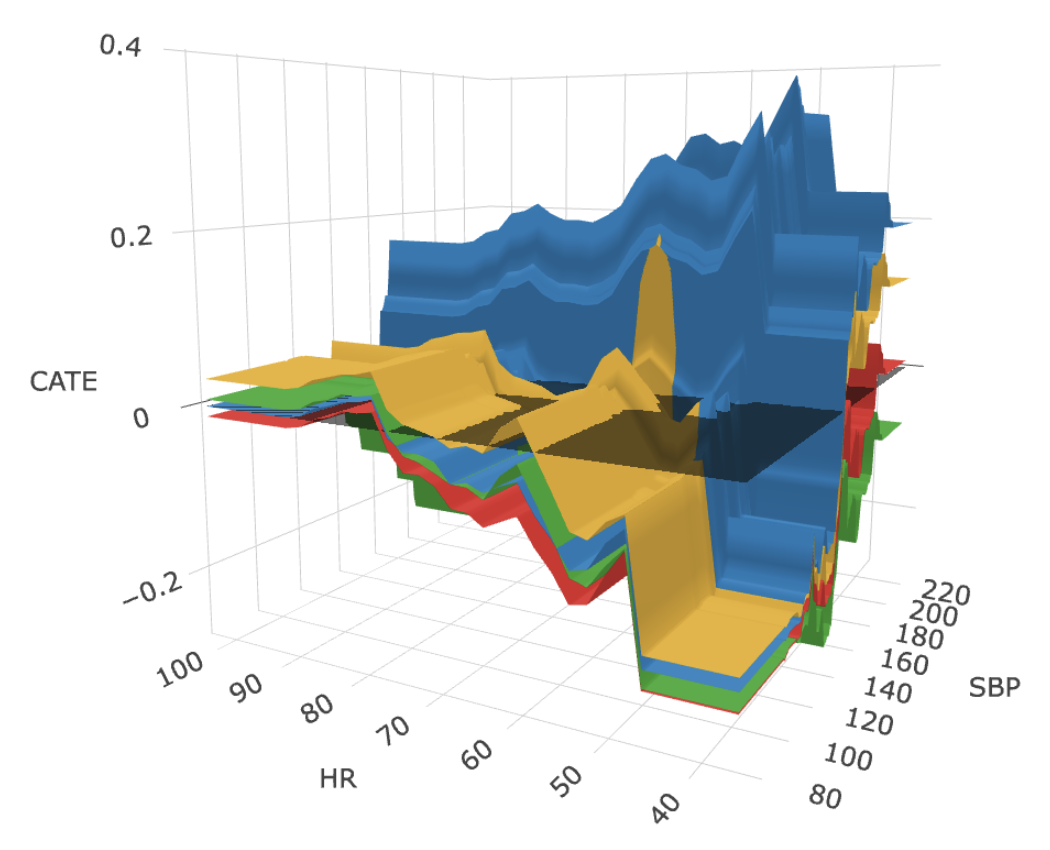}
\includegraphics[width=0.49\textwidth]{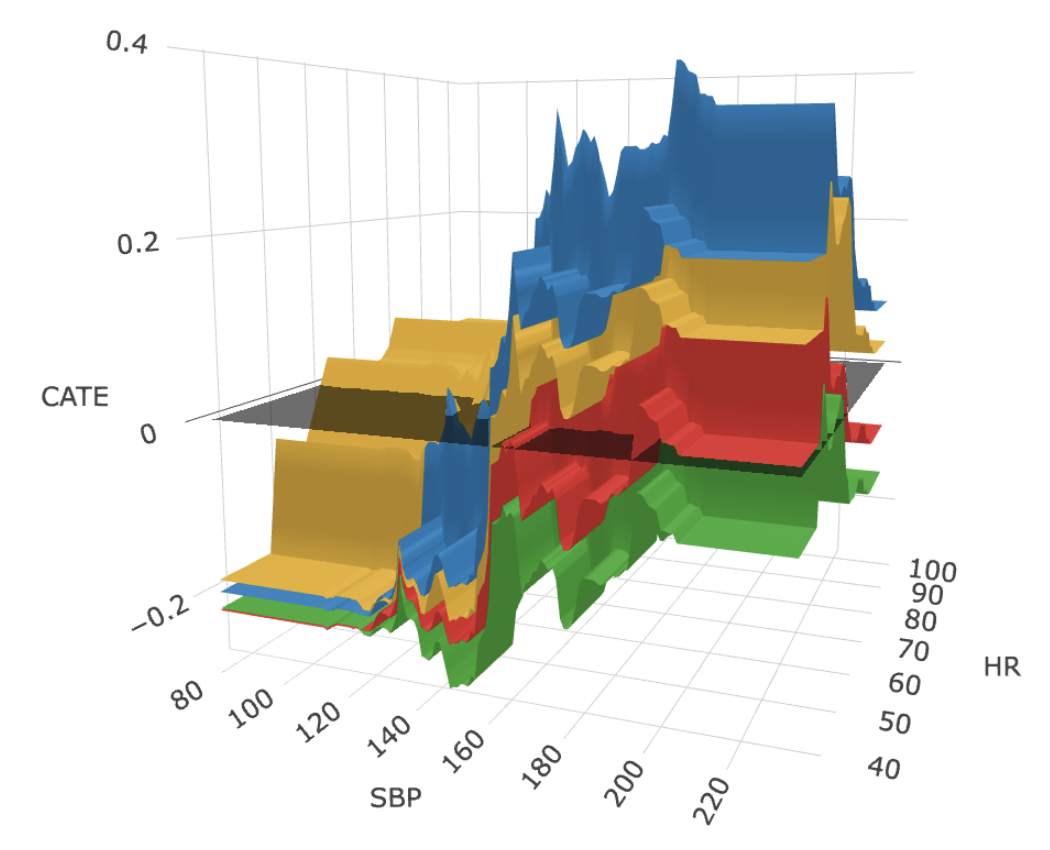}
\includegraphics[width=0.49\textwidth]{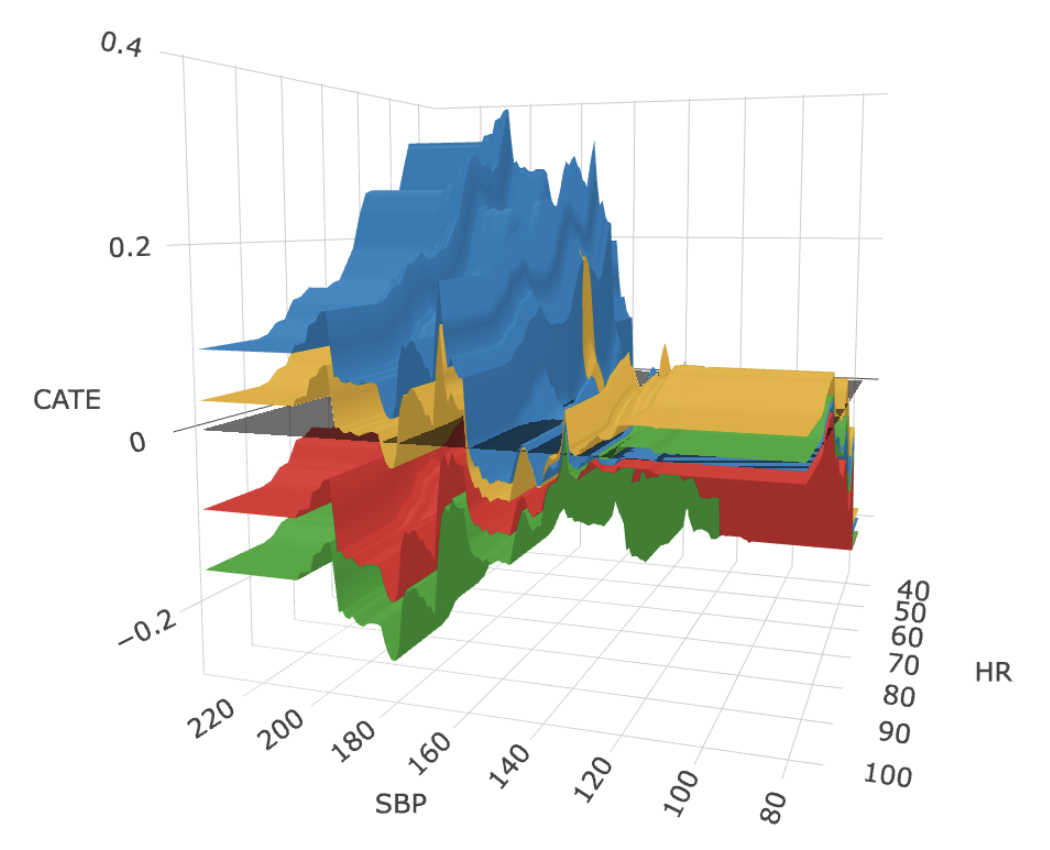}
\par 
\vspace{0.2em}
\hfill
\includegraphics[width=0.15\textwidth]{figures/HAAS_analysis/educ2/original_seed1/legend0.png}
\includegraphics[width=0.35\textwidth]{figures/HAAS_analysis/educ2/original_seed1/legend1.png}
\includegraphics[width=0.35\textwidth]{figures/HAAS_analysis/educ2/original_seed1/legend2.png}
\caption{Estimated CATE surfaces of DFS at age 90 from the ltrcDR-learner for different education and {\it APOE}  genotype subgroups (views from four different angles).
}
\label{fig:HAAS_CATE_plot3D_DR}
\end{figure}

\clearpage
\subsubsection{CATE surfaces that are less smooth}

Below are estimated CATE surfaces when 
the tuning parameters involved in the extreme gradient boosting implemented in the R package \texttt{xgboost} are chosen from the following candidate set:
\begin{itemize}
    \item \texttt{subsample}: 0.5, 0.75, 1;
    \item \texttt{colsample\_bytree}: 0.6, 0.8, 1;
    \item \texttt{eta}: 0.005, 0.01, 0.015, 0.025, 0.05, 0.08, 0.1, 0.2;
    \item \texttt{max\_depth}: 3, 4, ..., 20;
    \item \texttt{gamma}: randomly sampled from Unif(0, 0.2) distribution;
    \item \texttt{min\_child\_weight}: 1, 2, ..., 20;
    \item \texttt{max\_delta\_step}: 1, 2, ... , 10.
\end{itemize}

\begin{figure}[h]
\centering
\includegraphics[width=0.49\textwidth]{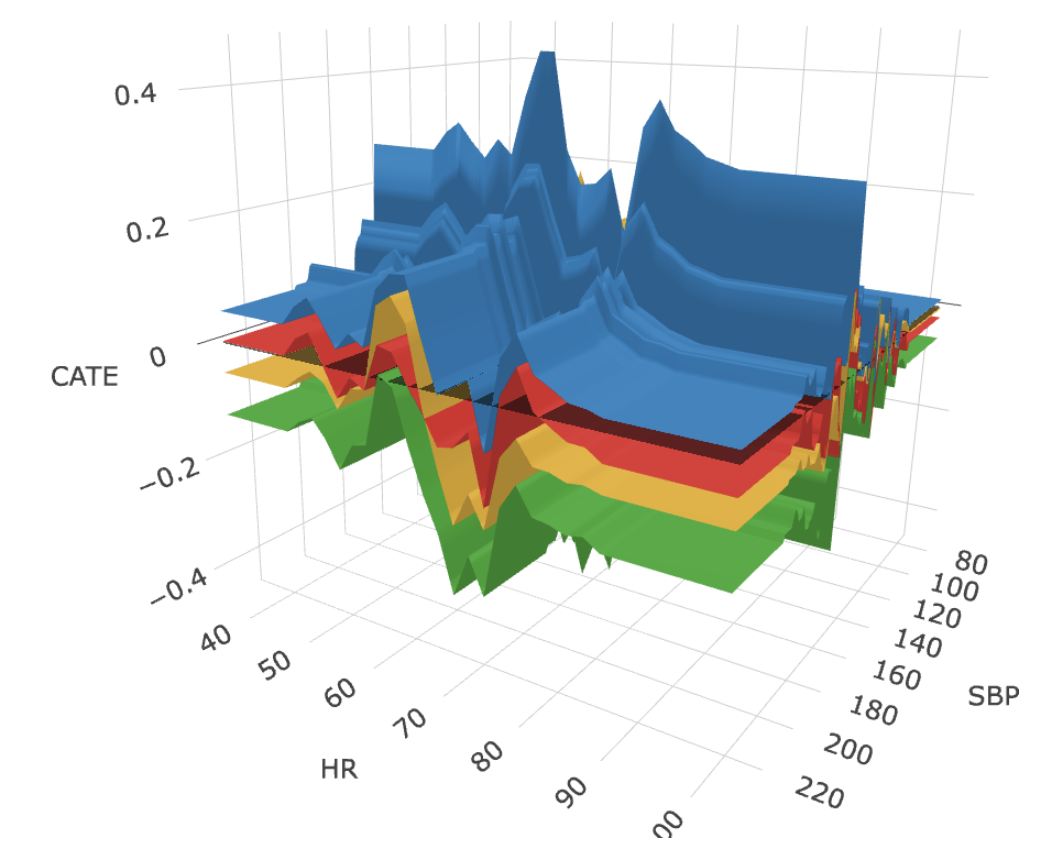}
\includegraphics[width=0.49\textwidth]{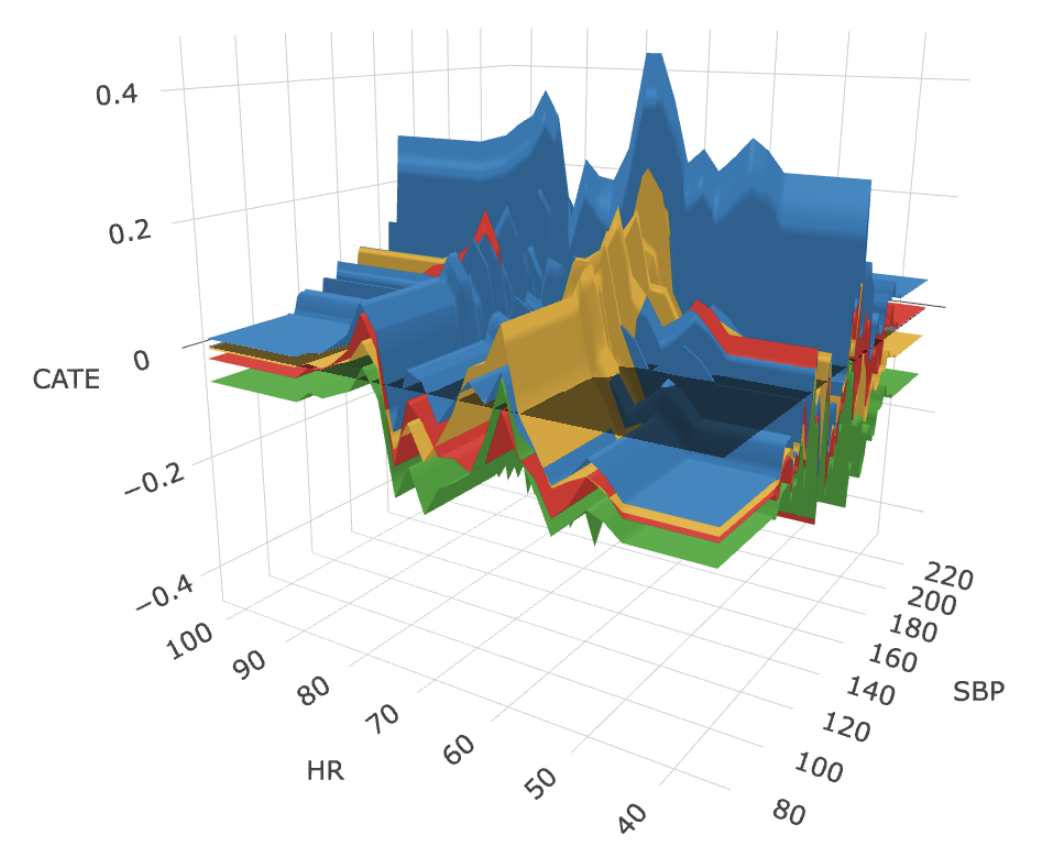}
\includegraphics[width=0.49\textwidth]{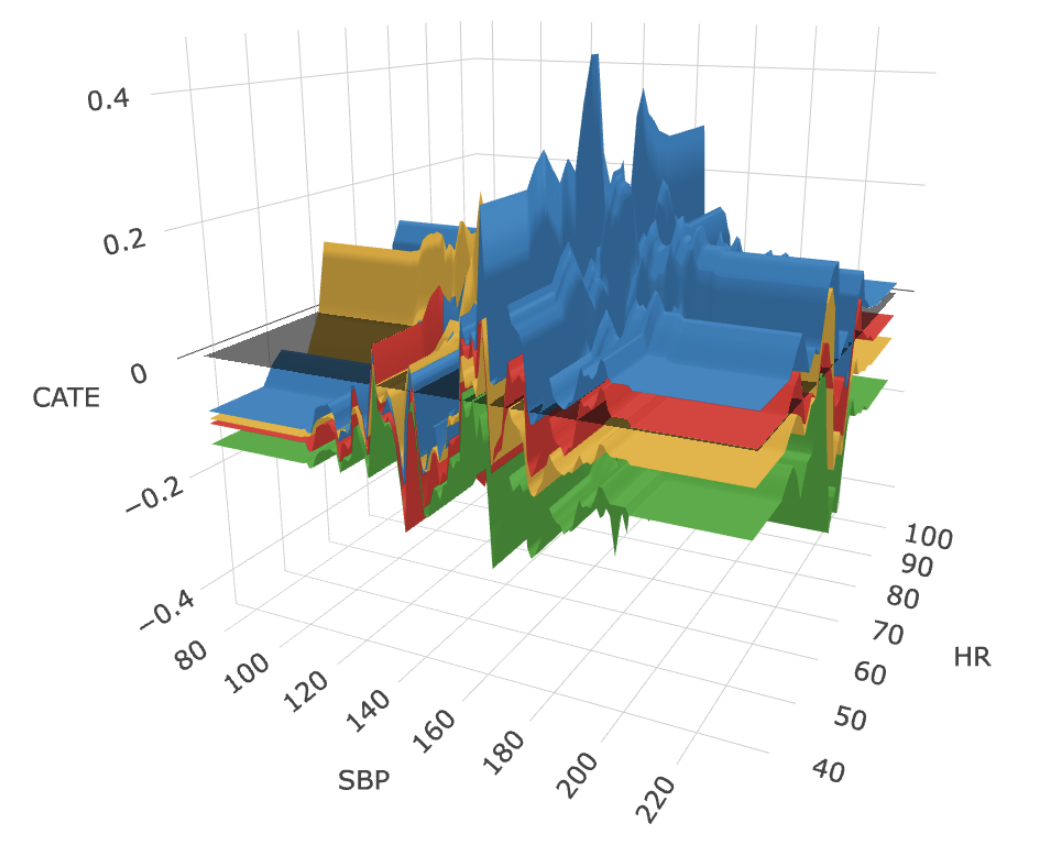}
\includegraphics[width=0.49\textwidth]{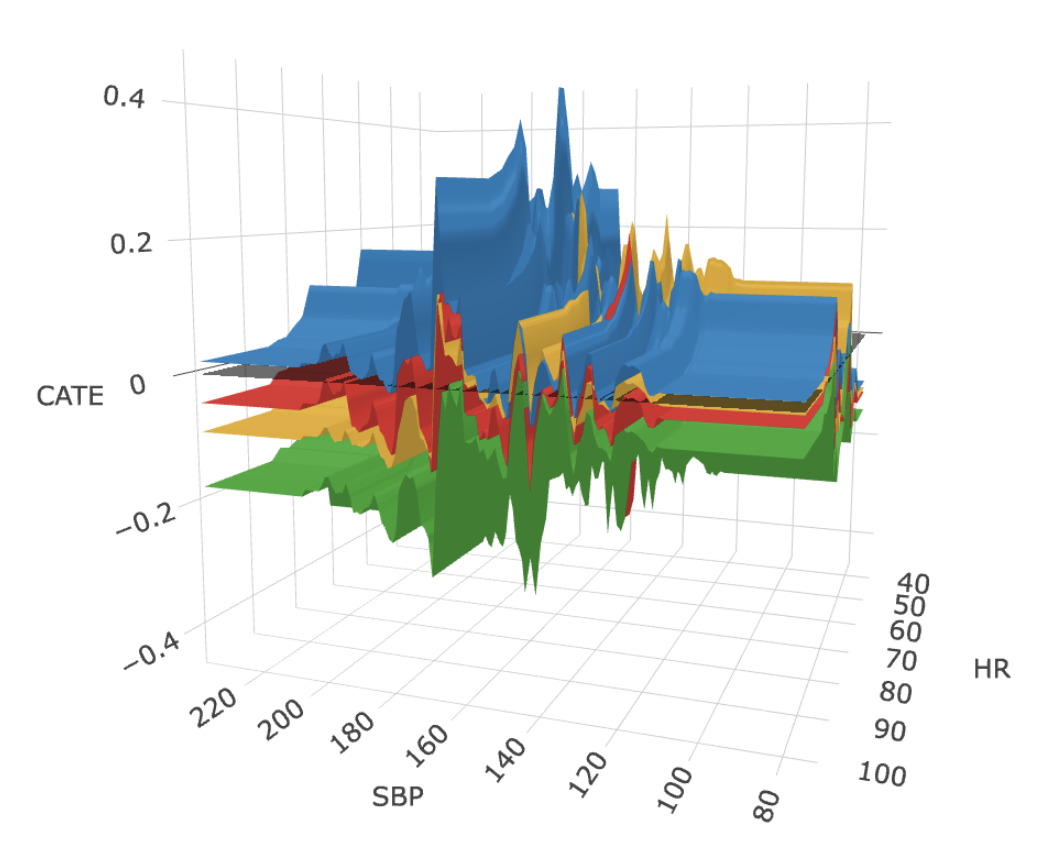}
\par 
\vspace{0.2em}
\hfill
\includegraphics[width=0.15\textwidth]{figures/HAAS_analysis/educ2/original_seed1/legend0.png}
\includegraphics[width=0.35\textwidth]{figures/HAAS_analysis/educ2/original_seed1/legend1.png}
\includegraphics[width=0.35\textwidth]{figures/HAAS_analysis/educ2/original_seed1/legend2.png}
\caption{Estimated less smooth CATE surfaces of DFS at age 90 from the ltrcR-learner for different education and {\it APOE}  genotype subgroups (views from four different angles).
}
\label{fig:HAAS_CATE_plot3D_try0}
\end{figure}


\end{document}